\documentclass{article}
\usepackage[utf8]{inputenc}
\usepackage{amsmath,amsthm,amsfonts}
\usepackage[margin=1in]{geometry}
\usepackage{mathrsfs}
\usepackage{graphicx}
\usepackage{hyperref}
\usepackage{makecell}
\usepackage{thm-restate}
\usepackage{booktabs} 

\allowdisplaybreaks

\usepackage{fancybox, pdfsync}

\newtheorem{definition}{Definition}

\usepackage[noend,linesnumbered]{algorithm2e}

\newenvironment{fminipage}%
  {\begin{Sbox}\begin{minipage}}%
  {\end{minipage}\end{Sbox}\fbox{\TheSbox}}

\newenvironment{algbox}[0]{\vskip 0.2in
\noindent 
\begin{fminipage}{6.3in}
}{
\end{fminipage}
\vskip 0.2in
}

\usepackage{bbm}

\newtheorem{theorem}{Theorem}

\newtheorem{lemma}[theorem]{Lemma}

\newtheorem*{remark*}{Remark}

\DeclareMathOperator*{\argmin}{arg\,min}

\usepackage{color}

\def\prob#1#2{\mbox{Pr}_{#1}\left[ #2 \right]}

\def\expec#1#2{{\mathbb{E}}_{#1}\left[ #2 \right]}

\def\defeq{\stackrel{\mathrm{def}}{=}}

\def\abs#1{\left|#1  \right|}

\def\norm#1{\left\| #1 \right\|}

\newcommand\Rbb{\mathbb{R}}

\newcommand\Otil{\tilde{O}}

\usepackage{enumitem}

\date{}
\begin{document}

\title{
Faster \texorpdfstring{$p$}{p}-Norm Regression Using Sparsity
}
\author{Mehrdad Ghadiri\footnote{Georgia Institute of Technology, \url{ghadiri@gatech.edu}} \and Richard Peng\footnote{Georgia Institute of Technology \& University of Waterloo, \url{rpeng@cc.gatech.edu}} \and Santosh S. Vempala\footnote{Georgia Institute of Technology, \url{vempala@gatech.edu}}}

\maketitle
\thispagestyle{empty}
\begin{abstract}
For a matrix $A\in \mathbb{R}^{n\times d}$ with $n\geq d$, we consider the dual problems of $\min \|Ax-b\|_p^p, \, b\in \mathbb{R}^n$  and $\min_{A^\top x=b} \norm{x}_p^p,\, b\in \mathbb{R}^d$. 
We improve the runtimes for solving these problems to high accuracy for every $p>1$ for sufficiently \emph{sparse} matrices.
We show that recent progress on fast sparse linear solvers can be leveraged to obtain faster than matrix-multiplication algorithms for any $p > 1$, i.e., in time $\tilde{O}(pn^\theta)$ for some $\theta < \omega$, the matrix multiplication constant.
We give the first high-accuracy input sparsity $p$-norm regression algorithm for solving $\min \|Ax-b\|_p^p$ with $1 < p \leq 2$, via a new row sampling theorem for the smoothed $p$-norm function. This algorithm runs in time $\tilde{O}(\text{nnz}(A) + d^4)$ for any $1<p\leq 2$, and in time $\tilde{O}(\text{nnz}(A) + d^\theta)$ for $p$ close to $2$, improving on the previous best bound where the exponent of $d$ grows with $\max\{p, p/(p-1)\}$. 
\end{abstract}
\newpage
\tableofcontents
\thispagestyle{empty}
\newpage
\pagenumbering{arabic}
\section{Introduction}
The $p$-norm regression problem can be stated in two ways that are dual to each other. The input consists of a matrix $A \in \Rbb^{n \times d}$, $n\geq d$ and a vector $b \in \Rbb^n$ or $b \in \Rbb^d$.
\begin{enumerate}[label=\arabic*., ref=\kern-0.85ex\arabic*]
    \item[(P1)]  \label{P1} $\min\norm{Ax-b}_p^p$.
    \item[(P2)]  \label{P2} $\min_{A^\top x=b} \norm{x}_p^p$. 
\end{enumerate}

For these problems, the \emph{square} setting (i.e., $n=O(d)$) and \emph{tall} setting (i.e., $n\gg d$) have both been intensively studied.
The case of $p=2$, or least squares, is classical linear regression, a problem that has been studied for centuries, and continues to be used in machine learning~\cite{Zhu05} and optimization~\cite{BV18}.
The case of $p=1$ or minimum absolute deviation, is also classical, and a topic that was studied by Laplace~\cite{L1774} and others.
Over the past few decades, $\ell_1$ minimization has turned out to be a very effective tool for sparse recovery and other applications~\cite{candes2006stable,candes2006robust, juditsky2011verifiable,juditsky2011accuracy}.

Regression in other norms, i.e., $p$ between $1$ and $2$, and $p>2$ has also been studied in Statistics for many decades. It is explicitly proposed and studied as a {\em robust} estimator, and the question of efficiently solving $p$-norm regression was stated as an important problem over 50 years ago~\cite{gentleman1965robust}.  
Regression in norms other than $p=1,2$ is widely used in network science~\cite{BH09,FWY20},
and the $p = O(\log{n})$ setting has been surprisingly useful as an algorithmic primitive for network flows~\cite{LS20,KLS20}.

For any $p\geq 1$, the $p$-norm regression problem is a convex optimization problem that can be solved in polynomial time. However, given its many applications and the increasing size of data sets, very efficient algorithms, ideally nearly linear in the input size, are desirable. For $p=2$, the problem can be solved to high accuracy in $\tilde{O}(\text{nnz}(A)+d^{\omega})$ time~\cite{CW13,NelsonN13,W14}\footnote{A complexity of $\tilde{O}(\text{nnz}+\text{poly}(d))$ is usually called {\em input sparsity} time.}, i.e., the dependence on the target error $\epsilon$ is only $O(\log(1/\epsilon))$.
This has been extended to any $p > 1$ but with a complexity of $\tilde{O}_p(\text{nnz}(A)+d^\omega + d^{0.5\max\{p,\frac{p}{p-1}\}+1})$~\cite{bubeck2018homotopy}, which becomes very large as $p$ deviates away from $2$. 
The state-of-the-art for general $p$, for high accuracy solutions, is that both variants above can be solved in time $\tilde{O}(n^{\omega})$, where $\omega$ is the matrix multiplication constant and the dependence on target error $\epsilon$ is logarithmic.
These complexities are a result of progress on fast optimization techniques over the past decade,
giving the current fastest runtimes of $\tilde{O}(n^{\omega})$~\cite{CLS21,B20} and 
$\tilde{O}(nd+d^{2.5})$~\cite{BrandLSS20,BLNPSS0W20,BLLSS0W21}
for solving linear programs.
For dense matrices, further improvements require either solving dense systems faster than matrix multiplication time ($n^\omega$),
or running in time sub-linear to input size ($nd$).
We note that solving linear programs is equivalent to $p$-norm minimization for $p=1$ or $p=\infty$.

Much less understood is the optimal asymptotic complexity of solving sparse optimization problems,
which represent an overwhelming majority of instances that arise in practice. Recent advances for dense matrices rely on fast matrix multiplication and matrix ``heavy-hitter" sketching, neither of which seem natural for sparse matrices. More specifically, (a) the complexity of solving sparse linear systems and (b) the existence of heavy-hitter matrices that could exploit sparsity are both unknown. 
This raises the following questions: (1) Can we go below matrix multiplication time for general $p$-norm regression for sufficiently sparse matrices? and (2) Is there an input sparsity algorithm for (P1) for $1< p \le 2$ (or for (P2) for $p > 2$)? In other words, 
a complexity of $\text{nnz}+\text{poly}(d)$ where the $\text{poly}(d)$ is a fixed polynomial independent of $p$. 

In this paper, we explore sparsity oriented tools and algorithms to speed up high-accuracy optimization. These include sparse linear system solvers, preconditioning, 
and row/column sampling, which can all exploit sparsity.
In fact the latter has lead to input sparsity algorithms for $p$-norm minimization, but with polynomial dependence on the target error $\epsilon$~\cite{DDHKM09,SW11,CDMMMW16,MM13,LMP13,CohenP15}.
Other methods lead to very high exponents in $d$, roughly $d^{0.5\max\{p,\frac{p}{p-1}\}+1}$~\cite{bubeck2018homotopy}.
Recent progress on sparse linear system solvers~\cite{PengVempala21} effectively uses rows of size $o(d)$, and sampling methods naturally preserve row-sparsity. Our goal is faster, high-accuracy algorithms for $p$-norm minimization. Our main contributions are faster algorithms for sparse $p$-norm regression. More specifically, we show that:
\begin{enumerate}
    \item Sparse linear system solvers can be adapted and used to go below the matrix multiplication threshold for general $p$-norm regression, i.e., a complexity of $n^\theta$ for $\theta < \omega$. In the special case of $p=2$ (and $p$ close to $2$), we get a runtime of $\text{nnz}(A)+d^\theta$. 
    \item We show that the sketch-and-precondition
    approach can be extended to general norms, giving the first input sparsity algorithm for (P1) for any $p \in (1,2]$ (and hence for (P2) for $p\geq 2$). The core of this result is a new sampling algorithm and analysis, where we use the smoothed $p$-norm (called the $\gamma$ function) and show that sampling with leverage scores can be used to approximate it. 
\end{enumerate}
We expect that the sparse optimization tools
we study here will be useful in broader settings, including for sparse linear programs.

\subsection{Results}

To state our results, we use the notation $\text{nnz}(A)$ to denote the number of nonzero entries of $A$ and $\text{nnz}_d(A)$ 
to denote the maximum number of nonzero entries in any $d$ rows of an $n \times d$ matrix $A$, with $n \ge d$. The following table provides a quick summary. 

\begin{table}[htb]
  \centering
  \begin{tabular}{p{0.15\textwidth}p{0.3\textwidth}p{0.3\textwidth}
  }
    \toprule
    Problem     & $\min \|Ax-b\|_p$ & $\min_{A^\top x = b} \|x\|_p$ \\
    \midrule
    Any $p$ &  $n^\theta$ & $n^\theta$ \\
    $1 < p < 2$ & $\text{nnz}(A)+d^4$  &  $d^{\frac{2-p}{2+p}}(\text{nnz}(A)+d^{\omega})$~\cite{jambulapati2021improved} \\
    $p=2$ &  $\text{nnz}(A)+d^\theta$   & $\text{nnz}(A)+d^\theta$   \\
    $p$ close to $2$ & $\text{nnz}(A)+d^\theta + d^{0.5\max\{p,\frac{p}{p-1}\}+1}$  & $\text{nnz}(A)+d^\theta + d^{0.5\max\{p,\frac{p}{p-1}\}+1}$ \\
    $p > 2$  & $d^{\frac{p-2}{3p-2}}(\text{nnz}(A)+d^{\omega})$~\cite{jambulapati2021improved}  & $\text{nnz}(A)+d^4$\\
    \bottomrule
  \end{tabular}
  \caption{The complexity of sparse $p$-norm regression. We assume the input is an $n \times d$ matrix $A$ with $n \ge d$. The second column is dual to the first column. The exponent $\theta$ represents a constant smaller than the current matrix multiplication exponent $\omega$. The first, third and fourth rows are the first improvements over fast matrix multiplication. For $1 < p < 2$ for the first problem (and $p > 2$ for the second problem), we get the first input sparsity algorithm. For $p > 2$ for the first problem (and $1 < p < 2$ for the second problem), the current best runtime is due to~\cite{jambulapati2021improved}.}
  \label{table:complexity}
\end{table}

Recent developments for regression problems are based on many ideas, including higher order smoothness~\cite{GDGVSU0W19}, homotopy methods~\cite{bubeck2018homotopy}, and generalized preconditioning~\cite{AdilKPS19,AdilS20}. At a high level, our algorithms are based on combining the latter two methods, which have $\log\left(1 / \epsilon\right)$ dependence on approximation error $\epsilon$, with recent developments of fast solvers for sparse linear systems~\cite{PengVempala21}.
To present our results in detail, we begin with linear regression. The following result of~\cite{CW13,NelsonN13,CohenLMMPS15} shows this problem can be solved in input-sparsity time up to a $d^\omega$ term.

\begin{theorem}[Input Sparsity Time Linear Regression Clarkson-Wooduff, Nelson-Nguyen~\cite{CW13,NelsonN13}]\label{thm:p-2}
The linear regression problem 
\[
\min\norm{Ax-b}_2^2
\]
can be solved to within relative error $1+\epsilon$ in time $\tilde{O}(\text{nnz}(A) + d^\omega)$ where d is the rank of $A$.
\end{theorem}

Our first theorem shows that for sparse matrices, we can go below the $d^\omega$ threshold and maintain input sparsity time.

\begin{restatable}{theorem}{linearRegression}
\label{thm:linear-regression}[Sparse Linear Regression]
Let $A\in \mathbb{R}^{n\times d}$, where $n\geq d$, be a matrix with condition number $\kappa$. Let 
$x^*=\argmin \norm{Ax-b}_2^2$. There is an algorithm that finds $\overline{x}$ such that  
\[
\norm{A\overline{x}-b}_2^2 \leq (1+\epsilon) \norm{Ax^* - b}_2^2
\]
in time 
\[
\tilde{O} \left( \left( \text{nnz}(A) + \text{nnz}_d(A)^{\frac{\omega-2}{\omega-1}} d^2 + d^{\frac{5\omega-4}{\omega+1}} \right)\log^{2} (\kappa/\epsilon) \log \left( \frac{\kappa \norm{b}_2}{\epsilon OPT} \right) \right),
\]
with probability at least $1- \tilde{O}(d^{-10})$.
\end{restatable}

We note that 
the theorem gives an improvement in the complexity of linear regression to $\text{nnz}(A)+o(d^\omega)$ for matrices that have $o(d^{\omega-2})$ nonzeros in each row (or $o(d^{\omega-1})$ entries in any $d$ rows). Moreover, these improvements hold for any value of $\omega > 2$. If each row has $O(1)$ entries, then the runtime with the current value of $\omega$ is bounded by $d^{2.331645}$ up to logarithmic terms. 
We next turn our attention to $p \neq 2$, starting with $p$ close to two. In this setting, \cite{bubeck2018homotopy} used a homotopy method to obtain the following result.

\begin{theorem}[Input Sparsity $p$-norm, $p$ near $2$, Bubeck-Cohen-Lee-Li~\cite{bubeck2018homotopy}]
\label{thm:p2}
Let $1<p<\infty$. The problem $\min \norm{Ax-b}_p^p$ can be solved to within relative error $(1+\epsilon)$ in time $\tilde{O}_p(\text{nnz}(A)+d^{\omega} + d^{0.5\max\{p,\frac{p}{p-1}\}+1})$.
\end{theorem}

We show that using a sparse inverse operator improves the time complexity for sparse matrices.

\begin{restatable}{theorem}{pCloseToTwo}
\label{thm:p-close-to-two-regression}[Sparse $p$-norm, $p$ near $2$]
Let $A\in \mathbb{R}^{n\times d}$, where $n\geq d$, be a matrix with condition number $\kappa$. Let $x^*=\argmin \norm{Ax-b}_p^p$. 
There is an algorithm that finds $\overline{x}$ such that  
\[
\norm{A\overline{x}-b}_p^p \leq (1+\epsilon) \norm{Ax^* - b}_p^p
\]
in time 
\[
\tilde{O}_p \left(
\left(
\text{nnz}\left(A\right) + d^{0.5\max\left\{p,\frac{p}{p-1}\right\}+1}
+  \text{nnz}_d\left(A\right)^{\frac{\omega-2}{\omega-1}} d^2 + d^{\frac{5\omega-4}{\omega+1}} \right)
\log^{2} \left(\kappa/\epsilon\right) \log \left( \frac{\kappa\norm{b}_2}{\epsilon OPT}
\right)
\right),
\]
with probability at least $1- \tilde{O}(d^{-10})$.
\end{restatable}

The above theorem gives improvements for $\frac{2\omega-2}{2\omega-3}<p<2\omega-2$ and matrices that have $o(d^{\omega-2})$ nonzeros in each row.
Next, we turn to the square case for arbitrary $p > 1$.
The current best bound is that of \cite{AdilS20}, which builds on \cite{AdilKPS19}.

\begin{theorem}
[$p$-Norm Regression in Matrix Multiplication Time, Adil-Kyng-Peng-Sachdeva, Adil-Sachdeva \cite{AdilKPS19,AdilS20}]\label{thm:p>2}
The $p$-norm regression problem of the form $\min_{A^\top x=b} \norm{x}_p^p$ can be solved in time $\tilde{O}(p (n^\omega + n^{7/3}))$ to high accuracy.
\end{theorem}

Our next result is an improvement of the above for general $p$ for sufficiently sparse matrices.

\begin{restatable}{theorem}{pNorm}
\label{thm:p-norm}[Sparse general $p$-norm faster than Matrix Multiplication]
Let $A\in \mathbb{R}^{n\times d}$ be a matrix with condition number $\kappa$. Let $x^*=\argmin_{A^\top x=b} \norm{x}_p^p$. Let $m<n^{1/4}$ be the number of blocks in the block Krylov matrix used by the sparse linear system solver. For $2 < p<\infty$, there is an algorithm that finds $\overline{x}$ such that $A \overline{x} = b$ and 
\[
\norm{\overline{x}}_p^p \leq (1+\epsilon) \norm{x^*}_p^p
\]
in time 
\begin{align*}
\tilde{O} \Big(
\Big( \text{nnz}\left(A\right) \cdot n \cdot m^{\frac{(p+2)}{(3p-2)}} + n^2 \cdot m^{3+\frac{(p-2)}{(3p-2)}} & + n^{2+\frac{p-(10-4\omega)}{3p-2}}+ n^\omega m^{2+\frac{(p-2)}{(3p-2)}-\omega} \Big) \\ & \cdot
 n^{o(1)} \left( p\log p \right) \log^{2} (\kappa/\epsilon) \log \left( \frac{\kappa\norm{b}_2}{\epsilon OPT} \right) \Big)
\end{align*}
with probability at least $1- \tilde{O}(n^{-10})$. For $1<p\leq 2$, there is an algorithm that finds $\overline{x}$ such that $A \overline{x} = b$ and $\norm{\overline{x}}_p^p \leq (1+\epsilon) \norm{x^*}_p^p$
in time 
\begin{align*}
\tilde{O} \Big(
\Big( \text{nnz}\left(A\right) \cdot n \cdot m^{\frac{(3p-2)}{(2+p)}} + n^2 \cdot m^{3+\frac{(2-p)}{(2+p)}} & + n^{2+\frac{p/(p-1)-\left(10-4\omega\right)}{3p/\left(p-1\right)-2}}+ n^\omega m^{2+\frac{(2-p)}{(2+p)}-\omega} \Big) \\ & \cdot
 n^{o(1)} \left( p\log p \right) \log^{2} (\kappa/\epsilon) \log \left( \frac{\kappa\norm{b}_2}{\epsilon OPT} \right) \Big)
\end{align*}
with probability at least $1- \tilde{O}(n^{-10})$.
\end{restatable}

For $\omega>\frac{7}{3}$ and $p>2$, $2+\frac{p-(10-4\omega)}{3p-2}<\omega$. Moreover,
by choosing $m$ to be a suitably small power of $n$, and noting that the exponent of $m$ in the last term with the $n^\omega$ factor is negative, we can ensure that the overall complexity is $n^\theta$ for some $\theta < \omega$. 
We note that for $1<p<2$, we can instead solve the dual problem for $\frac{p}{p-1}$ norm --- see Section 7.2 of \cite{AdilKPS19}; if $1<p<2$ and $\omega > \frac{7}{3}$, then 
\[
\frac{p/(p-1)-\left(10-4\omega\right)}{3p/\left(p-1\right)-2} <\omega-2.
\]
Surprisingly, the improvement of Theorem \ref{thm:p-norm} grows as $p$ deviates from $2$. 
Although our improvements for $p$ near $2$ (Theorems \ref{thm:linear-regression} and \ref{thm:p-close-to-two-regression}) are obtained by directly substituting the linear system solver in existing algorithms \cite{CohenLMMPS15,bubeck2018homotopy} with one tailored to sparse matrices, any improvement for general $p>1$ (Theorem \ref{thm:p-norm}) appears to require modifying the ``inverse maintenance" steps of \cite{AdilKPS19}. Inverse maintenance is a data structural based approach
for speeding up optimization algorithms.
It hinges upon the observation that the linear systems
arising from second-order optimization algorithms
are slowly changing.
It dates back to the early papers on interior point
methods~\cite{Karmarkar84,vaidya1989speeding,DBLP:conf/stoc/CohenLS19}, and is also at the core of
recent $d^{\omega}$-time optimization
algorithms for linear programming and other optimization problems~\cite{CLS21,B20,BrandLSS20,JSWZ21,JiangKLP020,BLNPSS0W20,BLLSS0W21}.
The main difference between the algorithm of \cite{AdilKPS19}
and the linear programming ones is that the total relative change
per step is bounded in $3$-norm instead of $2$-norm,

In the sparse setting, efficient inverse maintenance is not immediate because the output of the sparse linear system solver \cite{PengVempala21} is a representation of the inverse as a multiplication operator with polynomially large (e.g., $m=n^{0.01}$) number of bits. This introduces restrictions on the rank of the updates on the inverse that can be done using Sherman-Morrison-Woodbury identity. Our guarantee for $p>1$ requires opening up the
inverse maintenance steps, and directly associating the
size of the update maintained with the cost of solving. 

Finally, we turn to input sparsity algorithms for general $p$. An input sparsity time algorithm for linear regression was first presented by Clarkson and Woodruff \cite{CW13} by using sparse sketching tools. Later \cite{bubeck2018homotopy} presented a homotopy algorithm that runs in time $\Otil_p(\text{nnz}(A) + d^{\omega} + d^{0.5\max\{p,\frac{p}{p-1}\}+1})$, for all $1<p<\infty$. Note that when $p$ tends to one or infinity, the exponent of the third term tends to infinity. In contrast, our algorithm has a running time of $\Otil_p(\text{nnz}(A) + d^4)$ for all $p\in (1,2]$ for P(1), and for all $p\in [2,\infty)$ for (P2).

\begin{restatable}{theorem}{inputSparsityPNorm}
\label{thm:input-sparsity-p-norm}[Input sparsity time $p$-norm]
Let $A\in\Rbb^{n\times d}$, $C\in \Rbb^{d\times d}$, $b\in \Rbb^{n}$, and $v\in \Rbb^{d}$. Let $\kappa$ be an upper bound for the condition numbers of $A$ and $C$. Then, for $1< p \leq 2$, there is an algorithm that finds $\widetilde{x}$ such that $C\widetilde{x}=v$ and
\begin{align*}
\norm{A\widetilde{x} - b}_p^p \leq (1+\epsilon) \min_{Cx=v} \norm{Ax - b}_p^p,
\end{align*}
with high probability and in time $\Otil_p(\text{nnz}(A) + d^4)$. For $2\leq p<\infty$, there is an algorithm that finds $\widetilde{x}$ such that $A^\top\widetilde{x} = v$ and
\begin{align*}
\norm{x}_p^p \leq (1+\epsilon) \min_{A^\top x = v} \norm{x}_p^p,
\end{align*}
with high probability and in time $\Otil_p(\text{nnz}(A) + d^4)$.
\end{restatable}

Recently, \cite{adil2021almost} considered the case of $p\in [2,4]$ and presented an algorithm with a running time that matches with that of \cite{bubeck2018homotopy}.
The techniques we study also give an improvement on earlier work in the remaining range of $p$ ($p > 2$ for (P1) and $1 < p < 2$ for (P2)). 

\begin{restatable}{theorem}{almostInputSparsityPNorm}
\label{thm:almost-input-sparsity-p-norm}[Almost input sparsity time $p$-norm]
Let $A\in\Rbb^{n\times d}$, $C\in \Rbb^{d\times d}$, $b\in \Rbb^{n}$, and $v\in \Rbb^{d}$. Let $\kappa$ be an upper bound for the condition numbers of $A$ and $C$. Then, for $2 \leq p < \infty$, there is an algorithm that finds $\widetilde{x}$ such that $C\widetilde{x}=v$ and
\begin{align*}
\norm{A\widetilde{x} - b}_p^p \leq (1+\epsilon) \min_{Cx=v} \norm{Ax - b}_p^p,
\end{align*}
with high probability and in time $\Otil_p\left(\min_{q\in[2,p-1)} n^{1/q}\left(\text{nnz}(A) + d^{(q/2)+1} \right)\right)$. For $1<p\leq 2$, there is an algorithm that finds $\widetilde{x}$ such that $A^\top\widetilde{x} = v$ and
\begin{align*}
\norm{x}_p^p \leq (1+\epsilon) \min_{A^\top x = v} \norm{x}_p^p,
\end{align*}
with high probability and in time $\Otil_p\left(\min_{q\in[2,\frac{1}{p-1})} n^{1/q}\left(\text{nnz}(A) + d^{(q/2)+1} \right)\right)$.
\end{restatable}

When $\text{nnz}(A) = O(n)$, compared to $\min\{nd+\text{poly}(d), \text{nnz}(A)+d^\omega+d^{0.5\max\{p,\frac{p}{p-1}\}}\}$, Theorem \ref{thm:almost-input-sparsity-p-norm} gives improvements when $d^2<n<d^p$. Very recently, \cite{jambulapati2021improved}, presented an algorithm for all $p\in [2,\infty)$ for (P1) that runs in time $\Otil_p\left(d^{(p-2)/(3p-2)}\left( \text{nnz}(A) + d^\omega\right)\right)$.

\paragraph{Discussion of results.}
Our new results are Theorem~\ref{thm:linear-regression} (sparse linear regression), Theorem~\ref{thm:p-close-to-two-regression} (sparse regression for $p$ close to $2$), Theorem~\ref{thm:p-norm} ($p$-norm regression faster than matrix multiplication), Theorem~\ref{thm:input-sparsity-p-norm} (input sparsity $p$-norm regression) and Theorem~\ref{thm:almost-input-sparsity-p-norm} (almost input-sparsity $p$-norm regression). Of these, the first two, for $p=2$ and $p$ close to $2$ are relatively straightforward extensions of the sparse linear system solver of~\cite{PengVempala21} to first produce a fast spectral sparsifier and then adapt existing regression algorithms. The next result, Theorem~\ref{thm:p-norm} for general $p$-norm, is less immediate: beating the current bound of $n^\omega$ needs a combination of the sparse solver together with appropriate inverse maintenance so that the solver is effectively called only $\tilde{O}(1)$ times. The next result, Theorem~\ref{thm:input-sparsity-p-norm} about input sparsity $p$-norm regression, is perhaps the most surprising, as existing algorithms scale with $d^{\Omega(p/(p-1))}$; moreover, it was unclear if a row-sampling method could work, as the quantity that needs to be preserved by sampling is not the $p$-norm of $Ax$ for arbitrary $x$, but rather a smoothed version of it (called the $\gamma$-norm). We believe this sampling result is of independent interest.

\section{Technical Overview}

In this section, we give an overview of our algorithms and analysis. Our improvements rely on novel application of several tools for sparse matrices combined in a careful manner. For the first part, to go below matrix multiplication time, we use sampling techniques, inverse maintenance, and the newly introduced sparse linear system solvers. In the second part, we give an overview of our novel framework to do sampling for the quadratically smoothed $p$-norm functions. For this part, we use leverage scores and iterative sampling for our algorithm and we analyze the algorithm using $\epsilon$-nets.

\subsection{Faster than $n^\omega$}
Our starting point is the improvement in the complexity of solving linear systems for sufficiently sparse matrices. We emphasize that here (and throughout the paper), our complexities refer to the total {\em bit} complexity. In many cases, for numerical algorithms to return reliable results, the size of the bit representations might have to get larger along the way.

\begin{restatable}[\cite{PengVempala21}]{theorem}{SparseInverse}
\label{thm:sparse_inverse}
Given a sparse $d\times d$ matrix $A$ with max entry-wise magnitude at most $1$,
a diagonal $d\times d$ matrix $W$ with entry-wise magnitude at most $1$
and $m \leq d^{1/4}$, along with $\kappa$ that upper bounds the condition numbers of $A$ and $W$,
we can obtain in time 
\[
\Otil\left(\left(d\cdot \text{nnz}\left(A\right)\cdot m + d^2\cdot m^3+ \left(\frac{d}{m}\right)^\omega m^2\right)\log\left(\kappa\right)\right)
\]
a linear operator $Z_{AWA^\top}$ such that
\[
\norm{Z_{AWA^\top}
- \left(A W A^\top\right)^{-1}
}_F
\leq
\kappa^{-10} n^{-10}.
\]
Moreover, for a $d\times r$ matrix $B$, $Z_{AWA^\top} B$ can be computed in time $\tilde{O}((r \cdot \text{nnz}(A)\cdot m + d^{2} r^{\omega-2})\log(\kappa))$.
\end{restatable}
In the above theorem, $m$ denotes the number of blocks in the block Krylov space approach used by Peng and Vempala \cite{PengVempala21}. $m$ is also the number of bits in the sparse inverse representation. In the main theorem of \cite{PengVempala21}, $m$ is chosen as a value that optimizes the running time of Theorem \ref{thm:sparse_inverse}. However, we need to exploit the flexibility of $m$ in our running times because the cost of low rank updates on matrices with $\tilde{O}(m)$ bits also comes to play. In all of our results, one can see the improvement by setting $m$ to a small polynomial in $d$ (or $n$), e.g., $m=d^{0.01}$. However the best value of $m$ depends on multiple factors including the value of $\omega$ and the sparsity of the matrix. 

For brevity of notation, we sometime denote the running time of the sparse linear system solver by $d^\theta$. In these cases one can replace $d^\theta$ with $\Otil((d\cdot \text{nnz}_d(A)\cdot m + d^2\cdot m^3+ (\frac{d}{m})^\omega m^2)\log(\kappa))$.
Since the above statement is a bit more general than the main theorem of~\cite{PengVempala21}, we show how their methods easily extend to this version in Section~\ref{sec:accessInverse}.

For linear regression, the method of \cite{CohenLMMPS15}, uses a sequence of linear system solves to approximate leverage scores, samples the given matrix according to these scores, then applies the Richardson iteration to compute a high accuracy approximation. As we will see, the main ingredient we need is a spectral approximation with $\Otil(d)$ rows to the given matrix. 
The algorithm of \cite{bubeck2018homotopy} also needs spectral approximations of a set of $\tilde{O}(1)$ matrices. 
To handle both, we introduce the following efficient sparse spectral approximation --- see Definition \ref{def:spectral-approx}. 

\begin{theorem}
\label{thm:spectral_approx}
Let $A\in\mathbb{R}^{n\times d}$ such that the condition number of $A^\top A$ is $\kappa$. Let $m<d^{1/4}$ be the number of blocks in the block Krylov matrix used by the sparse linear system solver.
There exists an algorithm that finds a constant-factor spectral approximation $\widetilde{A}$ with $\tilde{O}(d)$ rows of $A$ in time 
\[
\tilde{O} \left(
\left( \text{nnz}\left(A\right) + d \cdot \text{nnz}_d(A) \cdot m + d^2 \cdot m^3 + \left( \frac{d}{m} \right)^\omega m^2 \right)
\cdot \log^{2} (\kappa) \right)
,
\]
with probability at least $1- \tilde{O}(d^{-10})$
\end{theorem}

We use sampling (as opposed to sketching) because sampling preserves the row sparsity. This results in a sparse matrix and enables us to use the sparse linear system solver.

For general $p$, the algorithm of \cite{AdilKPS19,AdilS20} is more complicated. When $p$ is large, it first reduces the $p$-norm regression problem to a small number of $q$-norm regression problems for $q=O(\sqrt{\log n})$. The latter problem is reduced to a sequence of weighted $2$-norm regression problems. Since it is too costly to solve each $2$-norm regression problem individually, they maintain an efficient preconditioner which is used in each iteration to solve the new problem by running Richardson's iterations to get a high accuracy solution. To maintain the preconditioner efficiently, they use the Sherman-Morrison-Woodbury identity, but the updates happen only when the changes are ``significant" as determined by the rank of the update, and a bucketing strategy. The amortized cost of their approach per iteration is $n^{\omega - (1/3)}$, with $n^{1/3}$ iterations overall. 

To improve on this, we want to use the sparse inverse. However, the theorem of \cite{PengVempala21} gives an inverse operator in the form $Z_{A^\top A}$ that involves matrices with $O(m\log(\kappa))$ bits. Thus, naively applying the Woodbury formula could be too expensive, resulting again in runtime that grows as $n^\omega$. To get around this, we set a threshold for the rank of the updates. We compute the sparse inverse entirely from scratch every $(n/m)^{1/3}$ iterations. This ensures that the rank of an update is at most $n/m$. This reduces the cost of the Woodbury update to below $n^\omega$ using fast rectangular matrix multiplication, in spite of $m$ bits per entry. These ideas are described precisely in Section~\ref{subsec:p-norm}.

\subsection{Input Sparsity Time}

The dual of $\min_{A^\top x=b} \norm{x}_p$ is $\max_{\norm{A y}_{p/(p-1)} \leq 1 } b^\top y$ which is equivalent to solving
\[
\min_{b^\top y = 1} \norm{A y}_{p/(p-1)}.
\]
Note that when $2\leq p < \infty$, $1<\frac{p}{p-1}\leq 2$. Therefore in this paper we consider the following general problem
\[
\min_{Cx = v} \norm{Ax}_p,
\]
where $C\in \mathbb{R}^{d\times d}$, $A \in \mathbb{R}^{n\times d}$, $v\in \mathbb{R}^d$, and $1<p\leq 2$.

Previous input sparsity time algorithms obtain their running times
by sketching/sampling this matrix,
and returning the solution $x$ computed on the
smaller, $\text{poly}(d)$-sized instance.
This approach, known as sketch-to-solve,
leads to runtimes of the form of $\text{nnz} + \text{poly}(d, \epsilon^{-1})$.
Errors are directly transferred between the
sketch and the original matrix,
and typically $\text{poly}(\epsilon^{-1})$ samples
are needed to obtain $(1 \pm \epsilon)$ relative error.

Our algorithm obtains an $\left(\text{nnz} + \text{poly}(d) \right) \log(1 / \epsilon)$
runtime via the sketch-and-precondition approach.
This randomized numerical linear algebra approach
has only been rigorously analyzed for problems
closely related to quadratic minimization problem.
Specifically, we use leverage score sampling to
produce approximations suitable for
$p$-norm preconditioning algorithms~\cite{AdilKPS19,AdilS20}.
It is shown by~\cite{AdilKPS19}
that the $p$-norm problem can be solved to high accuracy by approximately solving 
a sequence of ``residual" problems, 
defined via the following smoothed $p$-norm function (introduced by \cite{bubeck2018homotopy}), which combines the $p$-norm with a quadratic function.

\begin{definition}\label{def:quadratic-smooth}
We define the following quadratically smoothed $p$-norm function for scalars $x \in \mathbb{R}, t \in \mathbb{R}_{\geq 0}$
\begin{align}
\gamma_p(t, x) = \begin{cases}
\frac{p}{2} t^{p-2} x^2 & \text{ if } |x|\leq t,\\
|x|^p + (\frac{p}{2} - 1) t^p & \text{ otherwise.}
\end{cases}
\end{align}
Overloading notation, for vectors $x\in\mathbb{R}^d$ and $t\in\mathbb{R}^{n}_{\geq 0}$ and matrix $A\in\mathbb{R}^{n\times d}$, we define
\begin{align}
[\gamma_p(t, Ax)]_i = \begin{cases}
\frac{p}{2} t_i^{p-2} (Ax)_i^2 & \text{ if } |(Ax)_i|\leq t_i,\\
|(Ax)_i|^p + (\frac{p}{2} - 1) t_i^p & \text{ otherwise.}
\end{cases}
\end{align}
\end{definition}

Generally iterative methods that work with derivatives (first order or higher order) need to have information about how far we are from the optimum in order to adjust the step size. If we are far from the optimum, we want to take large steps and when we are close to the optimum, we want to take small steps to converge. In a quadratic function, this information is in the gradient.
However, for $p$-norm functions,
using either a quadratic, or a $p$-norm, to find this information can be inaccurate.
This is why we use the $\gamma$ function defined above, which accounts for both the locally linear and quadratic behavior of $f(t) = |t|^{p}$, and the long term $p$-norm behavior.
Specifically, this function is designed to take into account both the $p$-th power behavior of $|t + \delta|^{p} - t^{p}$ when $|\delta| \gg |t|$, as well as the locally linear + quadratic behavior of it when $|\delta| \ll |t|$.
These conditions allow us to precondition with it everywhere, instead of only in certain regions.

\begin{figure}[t]
    \centering
    \includegraphics[width=0.4\textwidth]{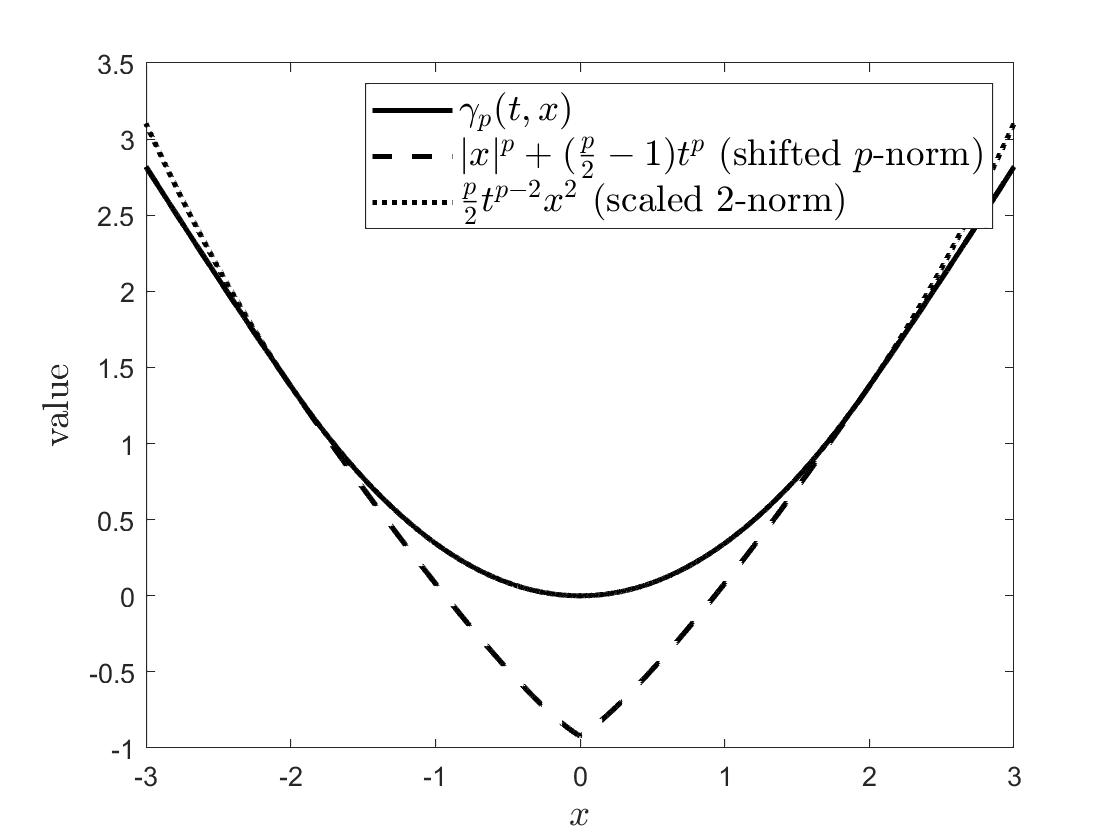}
    \caption{The $\gamma$ function for $p=1.2$ and $t=2$.}
    \label{fig:gamma}
\end{figure}

Although there are sampling algorithms that preserve $p$-norm ($1< p\leq 2$) and only need $\Otil(d)$ rows of the matrix \cite{bourgain1989approximation,CohenP15}, these cannot be used to get a high accuracy algorithm, and do not immediately translate to a sampling algorithm for the $\gamma$ function due to the introduction of threshold values.
To address this, we show in Theorem~\ref{thm:main-sampling}
that we can construct a matrix
$\widetilde{A}$ consisting of $\Otil(d^3)$ sampled and rescaled
rows of $A$ such that with high probability
the value of the $\gamma_p$ function is approximately preserved for all $x$ in a polynomial range.
This means solving the sampled problem allows us
to solve the residual problem to constant approximation.
That in turn, via $p$-norm preconditioning,
reduces the distance to optimum by a constant
factor, leading to convergence after $\Otil_p(\log(1 / \epsilon))$ iterations of an outer loop.

We analyze this sampling process using $\epsilon$-nets
in a manner similar to~\cite{DDHKM09}:
the $d$-dimensional space of all possible $x$
vectors (in a polynomial range that contains the optimal solution) is discretized into $n^{\Otil(d)}$ vectors,
and we show that the sample approximately
preserves the $\gamma_p$-function value
with probability at least 
$1 - \exp(-O(d \log{n}))$ for each vector.
However, a significant complication arise
because the $\gamma_p$ function is no longer
homogeneous: the function value at $x$
is not directly relatable to the value at $2x$.

This non-homongeniety prevents us from
directly working with the unit sphere:
it's also difficult to `decouple' the
$2$-norm and $p$-norm terms because of the
thresholding.
It is possible for $\gamma_{p}(t, y)$ to
be significantly less than both $\norm{y}_{p}^{p}$
and $\norm{y}_{2}^2$, so bounding
variance by either of those norms
is insufficient.
To address this, we explicitly consider the
contribution from the coordinates where the $\gamma$
function behave quadratically.
Specifically, we set up multiple cases based
on how this contribution compares to the ones
from entries where the $\gamma$-function
takes the $p$th power terms. Moreover, because of other problems that arise because of thresholding and the structure of the $\gamma$ function, the size of our $\epsilon$-net also depends on the condition number of the matrix and we have to do bucketing on the values of $t_i$'s and perform the sampling for each bucket separately.

This analysis leads
us to the choice of sampling by square-roots
of leverage scores.
While there are sampling probabilities better
tailored to preserving $p$-norm
functions~\cite{CohenP15,CDMMMW16},
the multiple conversions between $2$-norm
and $p$-norm we make in our analyses
precludes us from making gains using these more
general values.
Instead, we obtain a bound of sampling $\Otil(\sqrt{n} d^{1.5})$
rows, after which we iterate this process
in a manner similar to~\cite{LMP13} to
obtain a final $\text{poly}(d)$ bound.

\section{Preliminaries}

We denote the $i$'th row of a matrix $A$ with $a_i$ --- note that $a_i$ is a column vector. For a matrix $A\in \mathbb{R}^{n\times d}$ and a set $E
\subseteq [n]$, the matrix with rows of $A$ with indices in $E$ is denoted by $A_E$. We denote the psedoinverse of a matrix $A$ by $A^+$ while throughout the paper we assume that all input matrices are full-rank. Our results are generalizable to all matrices. We use $\Otil$ to hide polylogarithmic terms. We use $\Otil_p$ to hide polylogarithmic terms as well as terms that only depend on $p$, for example $p^2$ and $2^p$. For a vector $v$, we denote by $|v|$ a vector where $(|v|)_i = |v_i|$. Also for a vector $v$, we denote its corresponding diagonal matrix with the capital letter $V$.

\begin{definition}
\label{def:spectral-approx}
For $\lambda\geq 1$, $\widetilde{A}\in \mathbb{R}^{t\times d}$ is a $\lambda$-spectral approximation of $A\in\mathbb{R}^{n\times d}$ if,
\[
\frac{1}{\lambda} A^\top A \preceq \widetilde{A}^\top \widetilde{A} \preceq A^\top A,
\]
where $\preceq$ is the  Loewner ordering.
\end{definition}

A spectral approximation of a matrix is useful because it can be used as a preconditioner to solve linear regression problems.

\begin{lemma}[Richardson's iteration with preconditioning~\cite{Saad03}]
\label{lemma:richardson}
Given a matrix $M$ such that $A^\top A \preceq M \preceq \lambda \cdot A^\top A$ for some $\lambda\geq 0$. Let $x^{(k+1)}=x^{(k)} - M^{-1} (A^\top A x^{(k)}- A^\top b)$. Then we have
\[
\norm{x^{(k)}-x^*}_M
\leq
\left(1-\frac{1}{\lambda}\right)^k
\norm{x^{\left(0\right)} - x^*}_M,
\]
where $x^*=\argmin_x \norm{Ax-b}_2^2$.
\end{lemma}

A useful tool to find a small-sized spectral approximation of a matrix is the concept of statistical leverage scores.

\begin{definition}
The leverage score of the $i$'th row $a_i^\top$ of a matrix $A$ is $\tau_i(A)=a_i^\top (A^\top A)^+ a_i$. The \emph{generalized leverage score} of $i$'th row with respect to $B\in \mathbb{R}^{t\times d}$ is
\[
\tau_i^B(A)=\begin{cases}
a_i^\top (B^\top B)^+ a_i & ~~ \text{if} ~ a_i \perp \text{ker}(B), \\
\infty &  ~~ \text{otherwise.}
\end{cases}
\]
\end{definition}

The following lemma regarding the sum of leverage scores is key to bounding the number of samples in our sampling algorithms.

\begin{lemma}
[Foster's theorem \cite{foster1953stochastic}]
\label{lemma:sum-leverage-bound}
For a matrix $A\in\mathbb{R}^{n\times d}$,
$
\sum_{i=1}^n \tau_i(A) \leq d.
$
\end{lemma}

\begin{lemma}[\cite{CohenLMMPS15}]
\label{lemma:sampling-leverage}
The leverage scores of a matrix $A\in\mathbb{R}^{n\times d}$ can be computed to high accuracy in time $\Otil(\text{nnz}(A) + d^\omega)$. Moreover sampling (and scaling) $\Otil(d/\epsilon^2)$ rows according to leverage scores gives a $(1+\epsilon)$ spectral approximation with high probability.
\end{lemma}

\noindent
The following generalizes the concept of leverage scores to general $p$-norms.

\begin{definition}
The $p$-norm Lewis weights of a matrix $A\in\mathbb{R}^{n\times d}$ are defined as the unique weights $w$ such that for all $i\in [n]$,
\[
w_i = \tau_i(W^{1/2-1/p} A)
\]
\end{definition}

The following result gives a low-accuracy algorithm for finding the $p$-norm Lewis weights which is enough for our purposes. However a high-accuracy algorithm is presented very recently \cite{fazel2021computing}.

\begin{lemma}[\cite{bourgain1989approximation,CohenP15}]
\label{lemma:sampling-lewis}
For $p>2$, the $p$-norm Lewis weights of a matrix $A\in\mathbb{R}^{n\times d}$ can be computed in $\Otil(\text{nnz}(A) + d^{p/2})$ time to constant approximation. Moreover sampling (and rescaling) $\Otil(d^{p/2})$ rows according to the Lewis weights preserves the $p$-norm up to constant factors with high probability.
\end{lemma}

The following lemma, which is due to Johnson and Lindenstrauss, is useful for computing the leverage scores of a matrix fast.

\begin{lemma}[Random Projection \cite{johnson1984extensions,arriaga2006algorithmic}]
\label{lemma:random-projection}
Let $x\in\mathbb{R}^d$. Assume the entries in $G\in \mathbb{R}^{r\times d}$ are sampled independently from $N(0,1)$. Then,
\[
Pr \left((1-\epsilon)\norm{x}_2^2 \leq \norm{\frac{1}{\sqrt{r}}Gx}_2^2 \leq (1+\epsilon) \norm{x}_2^2\right)
\geq
1- 2e^{-\left(\epsilon^2-\epsilon^3\right) r /4}
\]
\end{lemma}

We use the following notation for the running time of fast matrix multiplication and fast rectangular matrix multiplication.

\begin{definition}
We denote the running time of multiplying an $r\times s$ matrix by an $s \times t$ matrix with $\textsc{MM}(r,s,t)$. The exponent of matrix multiplication is denoted by $\omega$. In other words, $d^\omega = \textsc{MM}(d,d,d)$. 
\end{definition}

\noindent
We will use the following concentration inequality.

\begin{lemma}[Chernoff bound \cite{chung2006complex}]
\label{lemma:chernoff-bound}
Let $X_1,\ldots,X_n$ be independent random variables with $\prob{}{X_i=1}=p_i$ and $\prob{}{X_i=0}=1-p_i$. For $X=\sum_{i=1}^n a_i X_i$, $a_i > 0$, we have $\expec{}{X} = \sum_{i=1}^n a_i p_i$ and we define $\nu = \sum_{i=1}^n p_i a_i^2$ and $a=\max\{a_1,\ldots,a_n\}$. Then we have
\begin{align*}
\prob{}{X \leq \expec{}{X} - \lambda} \leq \exp(-\lambda^2/2\nu)
\end{align*}
and
\begin{align*}
\prob{}{X \geq \expec{}{X} + \lambda} \leq \exp(-\lambda^2/(2\nu+2a\lambda/3)).
\end{align*}
\end{lemma}

\section{Tall $p$-Norm Regression in Input Sparsity Time}

In this section, we present our input sparsity results. We first present an iterative sampling algorithm that uses leverage scores. We show that the output of this algorithm preserves the value of $\gamma$ function up to constant factors in a polynomial range. This polynomial range is essentially a set that contains the optimal solution to the residual problem needed to be solved to solve the $p$-norm regression problem. We then discuss how the $p$-norm regression problem can be solved by approximately solving a small number of residual problems.

\subsection{Iterative Sampling for the Gamma Function}
\label{sec:iterative-sampling}

In this section, we present our iterative sampling algorithm to decrease the size of the residual problem while preserving the value of the function up to constant factors.

Without loss of generality, throughout this section, by rescaling, we can assume the first argument $t$ of $\gamma_q$ function is at least one. In other words, by definition of $\gamma_q$, for any $r>0$,
\begin{align*}
r\gamma_q(t, y) = \gamma_q(r^{1/q} t, r^{1/q} y).
\end{align*}

Given a matrix $A\in\mathbb{R}^{n\times d}$, a vector $t\in\mathbb{R}^n$, $1<q\leq 2$, and an oversampling parameter $h=\tilde{O}(d)$, our sampling algorithm performs the following iterative process,

\begin{enumerate}
    \item For $k = 1,\ldots, O(\log\log(n))$:
    \begin{enumerate}
        \item For $j\in\mathbb{Z}$, let $A^{(j)}$ be the matrix obtained by taking rows of $A$ that have $2^{j-1} \leq t_i < 2^j$.
        \item Sample each row $i$ of $A^{(j)}$ independently with probability $p_i = \min \{1, h \cdot \sqrt{\tau_i(A^{(j)})}\}$.
        \item Scale the selected rows and the corresponding $t_i$'s by $(1/p_i)^{1/q}$.
        \item Set $A$ to be the matrix comprised of the sampled and scaled rows.
    \end{enumerate}
\end{enumerate}

This algorithm gives the following result. We aggregate all the scaling factors in a weight vector $w$, where if row $i$ is in the final sample, $w_i$ is equal to the product of all $(1/p_i)$ for this row over the iterations, and $w_i$ is zero, otherwise.

\begin{theorem}
\label{thm:main-sampling}
Let
$t\in\mathbb{R}^n$, $A\in\mathbb{R}^{n\times d}$ such that $t\geq 1$ and $\beta$ be an upper bound for the condition number of $A$ and $\max_{i} t_i$. There exists an oversampling parameter $h=\Otil(d)$ and a vector of weights $w$ such that 
$|\text{supp}(w)|=\Otil(d^3)$ and for all $y=Ax$ such that $\norm{x}_2\in\left(\frac{1}{\beta n^{10}}, \beta n^{10}\right)$,
\[
\left\vert \sum_{i \in [n]} w_i \gamma_q(t_i,y_i) - \sum_{i\in [n]} \gamma_q(t_i,y_i)\right\vert \leq \frac{3}{4}\sum_{i\in [n]} \gamma_q(t_i,y_i)
\]
with probability at least $1 - O(\exp(-d))$.
\end{theorem}

\begin{proof}
The probability bound for the error of the $\gamma$ function follows from Theorem \ref{thm:main-eps-net} and the size of $\text{supp}(w)$ follows from Theorem \ref{thm:one-point-prob}.
\end{proof}

\subsubsection{Bounding the second moment}

We start by bounding the second moment of
sampling by square roots of leverage scores for a single $x$. The following shows that reweighting according to leverage scores makes the entries of $\ell_2$ norm uniform.

\begin{lemma}
\label{lemma:lev-score-uniform}
Let $A$ be a matrix and
$\tau_i$ be the leverage score of row $i$ of $A$.
Then for any $x$, all entries of the vector $y = Ax$
satisfy
\[
y_i^2 \leq \tau_i \norm{y}_2^2.
\]
\end{lemma}
\begin{proof}
Let $a_i$ be the $i$'th row of matrix $A$. Note that $a_i$ is a column vector.
We have 
\[
y_i = a_i^T x = a_i^T (A^\top A)^{-1/2} (A^\top A)^{1/2}  x,
\]
so by Cauchy-Schwarz inequality:
\[
y_i^2 \leq a_i^\top (A^\top A)^{-1} a_i (x^\top A^\top A x) = \tau_i \norm{Ax}_2^2 = \tau_i \norm{y}_2^2.
\]
\end{proof}

This shows that the mass of $y$ cannot
be concentrated in a few entries.
Below we will give two lemmas that bound
the variance of sampling $\gamma_q(t, y)$.
They are at the core of our proof
of the concentration of our row sampling.

\begin{lemma}
\label{lemma:variance-bound}
Let $y,\tau\in\mathbb{R}^n$ such that $1 \geq \tau> 0$ and for all $i\in[n]$, $y_i^2 \leq \tau_i \norm{y}^2$. Let $h\geq 1$ and $p_i = h\sqrt{\tau_i}$. Moreover let $1\leq q\leq 2$, $S\subseteq [n]$ and $\overline{S} = [n]\setminus S$ such that $\norm{y_{\overline{S}}}_2\geq 0.5$, and for all $i\in S$, $|y_i| \leq 1$. Then for some absolute constant $\bar{c}$,
\[
\sum_{i\in \overline{S}} \frac{1}{p_i} |y_i|^{2q} + \sum_{i \in S} \frac{1}{p_i} y_i^4 \leq
\frac{\bar{c}}{h} \left(\norm{y_{\overline{S}}}_q^q + \norm{y_S}_2^2 \right)^2.
\]
\end{lemma}
\begin{proof}
We have three cases.

\textbf{Case 1.}
$\norm{y_{\overline{S}}}_q \geq \norm{y_S}_q$.
In this case
\[
\label{eq:case1-2norm-bound}
\norm{y}_2 \leq \norm{y}_q \leq \norm{y_{\overline{S}}}_q + \norm{y_S}_q \leq 2 \norm{y_{\overline{S}}}_q,
\]
where the last inequality follows from the assumption and the first inequality holds since $1\leq q\leq 2$. 
On the other hand, the uniformity assumption,
plus the
conditions of $0\leq\tau_i\leq 1$
and $1\leq q \leq 2$ gives:
\[
\sum_{i\in \overline{S}} \frac{1}{p_i} |y_i|^{2q} = \sum_{i\in \overline{S}} \frac{1}{h \sqrt{\tau_i}} |y_i|^{2q} \leq \sum_{i\in \overline{S}} \frac{1}{h \sqrt{\tau_i}} \tau_i^{q/2} \norm{y}_2^q |y_i|^q \leq 
\frac{1}{h} \norm{y}_2^q \norm{y_{\overline{S}}}_q^q
\]
\noindent
So incorporating the bound above gives
\[
\leq 2^q \cdot \frac{1}{h} \norm{y_{\overline{S}}}_q^{2q} \leq 4 \cdot \frac{1}{h}
\left(\norm{y_{\overline{S}}}_q^q + \norm{y_S}_2^2\right)^2.
\]

\noindent
For the second term,
combining 
$|y_i|\leq 1$ for all $i \in S$
and $1\leq q \leq 2$
with the uniformity assumption gives:
\[
\sum_{i \in S} \frac{1}{p_i} y_i^4 
=
\sum_{i\in S} \frac{1}{h \sqrt{\tau_i}} y_i^4
\leq
\sum_{i\in S} \frac{1}{h \sqrt{\tau_i}}
\abs{y_i}^{2q}
\leq \sum_{i\in S} \frac{1}{h \sqrt{\tau_i}} \tau_i^{q/2} \norm{y}_2^q \abs{y_i}^q
\]
Again invoking $0< \tau_i \leq 1$ and $1\leq q \leq 2$ gives:
\[
\leq
\frac{1}{h} \norm{y}_2^q \sum_{i\in S}\abs{y_i}^q
=
\frac{1}{h} \norm{y}_2^q \norm{y_S}_q^q.
\]
Then substituting in the case assumption
($\norm{y_{\overline{S}}}_q \leq \norm{y_{S}}_q$)
and the initial bound
($\norm{y}_2 \leq 2 \norm{y_{\overline{S}}}_q$)
gives:
\[
\leq
2^q \cdot \frac{1}{h} \norm{y_{\overline{S}}}_q^{2q}
\leq
4 \cdot \frac{1}{h}
\left(\norm{y_{\overline{S}}}_q^q + \norm{y_S}_2^2\right)^2
\]

\textbf{Case 2.} $ \norm{y_S}_2 \geq \norm{y_{\overline{S}}}_2$.
In this case we have
\[
\norm{y}_2 \leq \norm{y_{S}}_2 + \norm{y_{\overline{S}}}_{2} \leq 2 \norm{y_S}_2,
\]
\noindent
Combining this with the uniformity assumption
and $0< \tau_i \leq 1$ and $1\leq q \leq 2$ gives:
\[
\sum_{i\in \overline{S}} \frac{1}{p_i} \abs{y_i}^{2q}
=
\sum_{i\in \overline{S}} \frac{1}{h \sqrt{\tau_i}} \abs{y_i}^{2q}
\leq
\sum_{i\in \overline{S}} \frac{1}{h \sqrt{\tau_i}} \tau_i^{q/2} \norm{y}_2^q \abs{y_i}^q
\leq
\frac{1}{h} \norm{y}_2^q \norm{y_{\overline{S}}}_q^q.
\]
Incorporating $1 \leq q \leq 2$ and
the assumption that $\norm{y_S}_2 \geq \norm{y_{\overline{S}}}_2\geq 0.5$ then gives
the bound on the first term:
\[
\leq 2^q \cdot \frac{1}{h} \norm{y_S}_2^q \norm{y_{\overline{S}}}_q^q \leq 4\cdot \frac{1}{h} \norm{y_S}_2^2 \norm{y_{\overline{S}}}_q^q
\leq 2 \cdot \frac{1}{h}
\left(\norm{y_{\overline{S}}}_q^q + \norm{y_S}_2^2\right)^2.
\]
\noindent
For the second term, the uniformity assumption
and the fact that $\tau_i \leq 1$ gives
\[
\sum_{i \in S} \frac{1}{p_i} y_i^4
=
\sum_{i\in S} \frac{1}{h \sqrt{\tau_i}} y_i^4
\leq \sum_{i\in S} \frac{1}{h \sqrt{\tau_i}} \tau_i \norm{y}_2^2 y_i^2
\leq \frac{1}{h} \norm{y}_2^2 \sum_{i\in S} y_i^2
= \frac{1}{h} \norm{y}_2^2 \norm{y_S}_2^2
\]
Substituting in the bound on $\norm{y}_2$ at
the start of this case then gives
\[
\leq 4 \cdot \frac{1}{h} \norm{y_S}_2^4 \leq 4 \cdot \frac{1}{h}\left(\norm{y_{\overline{S}}}_q^q + \norm{y_S}_2^2\right)^2,
\]

\textbf{Case 3.} $\norm{y_{\overline{S}}}_q \leq \norm{y_S}_q$ and $ \norm{y_S}_2 \leq \norm{y_{\overline{S}}}_2$.
In this case,
\[
\label{eq:case3-2norm-bound}
\norm{y}_2 \leq \norm{y_{S}}_2 + \norm{y_{\overline{S}}}_{2} \leq 2 \norm{y_{\overline{S}}}_q,
\]
where the second inequality comes from 
$\norm{y_{\overline{S}}}_2 \leq \norm{y_{\overline{S}}}_q$ 
which implies
\[
\norm{y_S}_2
\leq
\norm{y_{\overline{S}}}_2
\leq
\norm{y_{\overline{S}}}_q.
\]
\noindent
For the first term, the uniformity assumption gives
\[
\sum_{i\in \overline{S}} \frac{1}{p_i} \abs{y_i}^{2q}
=
\sum_{i\in \overline{S}} \frac{1}{h \sqrt{\tau_i}} \abs{y_i}^{2q}
\leq
\sum_{i\in \overline{S}} \frac{1}{h \sqrt{\tau_i}} \tau_i^{q/2} \norm{y}_2^q \abs{y_i}^q,
\]
which combined with
and $0\leq \tau_i \leq 1$ and $1\leq q \leq 2$ 
and the bound on $\norm{y}_2$ in this case gives:
\[
\leq
\frac{1}{h} \norm{y}_2^q \norm{y_{\overline{S}}}_q^q
\leq 2^q \cdot \frac{1}{h} \norm{y_{\overline{S}}}_q^{2q} \leq 4 \cdot \frac{1}{h}\left(\norm{y_{\overline{S}}}_q^q + \norm{y_S}_2^2\right)^2.
\]
\noindent
For the second term,
because for all $i \in S$,
$|y_i| \leq 1$,
we get via the uniformity assumption
\[
\sum_{i \in S} \frac{1}{p_i} y_i^4
= \sum_{i\in S} \frac{1}{h \sqrt{\tau_i}} y_i^4
\leq \sum_{i\in S}
\frac{1}{h \sqrt{\tau_i}} \abs{y_i}^{2+q}
\leq  
\sum_{i\in S} \frac{1}{h \sqrt{\tau_i}} \tau_i^{q/2} \norm{y}_2^q y_i^2
\]
after which 
incorporating $0< \tau_i \leq 1$
and $1\leq q \leq 2$,
and the $\ell_2$ bound for this case gives:
\[
\leq
\frac{1}{h}\norm{y}_2^q \norm{y_S}_2^2 \leq 2^q \cdot \frac{1}{h} \norm{y_{\overline{S}}}_q^q \norm{y_S}_2^2 \leq 2 \cdot \frac{1}{h}
\left(\norm{y_{\overline{S}}}_q^q + \norm{y_S}_2^2\right)^2,
\]
\end{proof}

The following gives a bound on the second moment of our random variable.

\begin{theorem}
\label{thm:main-variance-bound}
Let $t\in\mathbb{R}^n_{\geq 0}$, $A\in\mathbb{R}^{n\times d}$ be a full column rank matrix, $x\in\mathbb{R}^d$, and $y= Ax$. Let $j\in\mathbb{Z}$ such that for all $i\in [n]$, $2^{j-1}\leq t_i\leq 2^j$. Let $h\geq 1$ be an oversampling parameter. For all $i\in [n]$,
let $p_i=h \sqrt{\tau_i}$, where $\tau_i$ is the leverage score of row $i$ of $A$, i.e., $\tau_i=a_i^\top (A^\top A)^{-1} a_i$. Suppose we put index $i\in [n]$ in set $R$ independently with probability $p_i$. Then for $1\leq q \leq 2$, there is an absolute constant $c$ such that
\[
\sum_{i \in \left[ n \right]}
\frac{1}{p_i}
\left(\gamma_q\left(t_i,y_i\right)\right)^2
\leq \frac{c}{h}
\left(\sum_{i \in \left[ n \right]}
\gamma_q\left(t_i,y_i\right)\right)^2
\]
\end{theorem}

\begin{proof} 

First note that by definition of $\gamma_q$, we have
\[
\gamma_q\left(t,y\right)
=
\left(2^j\right)^q \gamma_q\left(t/2^j,y/2^j\right).
\]
Therefore, to prove the lemma,
it is enough to prove
\[
\sum_{i \in \left[ n \right]} \frac{1}{p_i}
\left( \gamma_q\left(t_i/2^j,y_i/2^j\right)\right)^2
\leq
\frac{c}{h} \left(\sum_{1 \leq i \leq n} \gamma_q\left(t_i/2^j,y_i/2^j\right)\right)^2.
\]

Due to the assumption of
$2^{j-1}\leq t_i\leq 2^j$
for all $i$,
without loss of generality,
for the rest of the proof,
we assume $\frac{1}{2}\leq t_i\leq 1$,
for all $i \in [n]$.
By definition of $\overline{S}$, we have $t_i^q \leq |y_i|^q$. 
Moreover because $1\leq q\leq 2$,
$-1\leq \frac{q}{2} - 1 \leq 0$.
Hence
\[
0 \leq \left(
\abs{y_i}^q + \left(\frac{q}{2} - 1\right)
t_i^q\right)^2
\leq
\abs{y_i}^{2q}.
\]
Moreover since $1\leq q \leq 2$ and $\frac{1}{2}\leq t_i \leq 1$,
\[
\frac{q}{2} t_i^{q-2} y_i^2 \leq 2 y_i^2.
\]
Therefore
\[
\sum_{i \in \left[n\right]}
\frac{1}{p_i} \left( \gamma_q(t_i,y_i)\right)^2
=
\sum_{i\in \overline{S}} \frac{1}{p_i}
\left( \abs{y_i}^q +
\left(\frac{q}{2} - 1\right)t_i^q\right)^2
+
\sum_{i \in S} \frac{1}{p_i}
\left( \frac{q}{2} t_i^{q-2} y_i^2 \right)^2
\leq
4\left(\sum_{i\in \overline{S}}
\frac{1}{p_i} \abs{y_i}^{2q}
+
\sum_{i \in S} \frac{1}{p_i} y_i^4\right).
\]
There are two cases remaining depending on
whether $\overline{S}$ is empty:

\paragraph{If $\overline{S} \neq \emptyset$:}

Since $t_i\geq 0.5$, for all $i\in[n]$,
and $|y_i|\geq t_i$, for all $i\in\overline{S}$,
we get
\[
\norm{y_{\overline{S}}}_2 \geq 0.5
\]
So since $|y_i| \leq t_i \leq 1$ for all $i\in S$,  by Lemma \ref{lemma:variance-bound} we get
\[
\sum_{i\in \left[n\right]}
\frac{1}{p_i}
\left( \gamma_q\left(t_i,y_i\right)\right)^2
\leq
\frac{4\bar{c}}{h}
\left(\norm{y_{\overline{S}}}_q^q + \norm{y_S}_2^2 \right)^2.
\]

\paragraph{If $\overline{S}=\emptyset$:}

In this case, the uniformity assumption gives
\[
\sum_{i\in \left[n\right]}
\frac{1}{p_i} \left( \gamma_q(t_i,y_i)\right)^2
\leq
4\sum_{i \in \left[n\right]} \frac{1}{p_i} y_i^4
=
4\sum_{i \in \left[n\right]}
\frac{1}{h\sqrt{\tau_i}} y_i^4 
\leq
4\sum_{i \in \left[n\right]}
\frac{1}{h\sqrt{\tau_i}} \tau_i \norm{y}_2^2 y_i^2
\]
which upon incorporating
$0\leq \tau_i \leq 1$ and $1\leq q \leq 2$ gives:
\[
\leq
\frac{4}{h} \norm{y}_2^4
=
\frac{4}{h} \norm{y_S}_2^4
\leq
\frac{4\bar{c}}{h}
\left(\norm{y_{\overline{S}}}_q^q + \norm{y_S}_2^2\right)^2.
\]
Taking square roots of both sides gives
\[
\sqrt{\sum_{i\in \left[n\right]}
\frac{1}{p_i}
\left( \gamma_q(t_i,y_i)\right)^2}
\leq 2\sqrt{\frac{\bar{c}}{h}}
\left( \sum_{i\in \overline{S}}
\abs{y_i}^q + \sum_{i \in S} y_i^2 \right)
\]
where upon incorporating
$|y_i|>t_i$ for all $i\in\overline{S}$,
$0.5\leq t_i\leq 1$, for all $i\in [n]$,
and $1 \leq q \leq 2$ gives
\[
\leq 4 \sqrt{\frac{\bar{c}}{h}}
\left( \sum_{i\in \overline{S}}
\left(\abs{y_i}^q + \left(\frac{q}{2}-1\right) t_i^q \right)
+
\sum_{i \in S} \frac{q}{2} t_i^{q-2} y_i^2 \right)
= 4 \sqrt{\frac{\bar{c}}{h}}
\sum_{i\in \left[n\right]} \gamma_q\left(t_i,y_i\right),
\]
The result follows by setting $c = 16\bar{c}$.
\end{proof}

\subsubsection{Concentration 
for a single $x$}

In addition to an upper bound on the second moment, the upper tail of the Chernoff bound (Lemma \ref{lemma:chernoff-bound}) needs an upper bound on the maximum of the random variable.
These variance and magnitude bounds in turn
allow us to show concentration for a particular vector $x$.

\begin{theorem}
\label{thm:concentation}
Let $t\in\mathbb{R}^n_{\geq 0}$, $A\in\mathbb{R}^{n\times d}$, $x\in\mathbb{R}^d$, and $y= Ax$. Let $j\in\mathbb{Z}$ such that for all $i\in [n]$, $2^{j-1}\leq t_i\leq 2^j$. Let $h\geq 1$ be an oversampling parameter. For all $i\in [n]$,
let $p_i=\min\{1, h \sqrt{\tau_i}\}$, where $\tau_i$ is the leverage score of row $i$ of $A$, i.e., $\tau_i=a_i^\top (A^\top A)^+ a_i$. Suppose we put index $i\in [n]$ in set $R$ independently with probability $p_i$. Then for $1\leq q \leq 2$, there is an absolute constant $\bar{c}$ such that for any $0<\epsilon<1$
\[
\prob{}{\left\vert \sum_{i\in R} \frac{1}{p_i} \gamma_q\left(t_i,y_i\right) - \sum_{i\in \left[n\right]} \gamma_q\left(t_i,y_i\right) \right\vert 
\geq \epsilon \sum_{i\in \left[n\right]} \gamma_q\left(t_i,y_i\right)} \leq 2\exp\left(-\frac{h\epsilon^2}{\bar{c}}\right).
\]
\end{theorem}

\begin{proof}
First without loss of generality, we assume that for all $i$, $p_i=h\sqrt{\tau_i}$ because if $T=\{i: h\sqrt{\tau_i}>1\}$, then
\[
\left\vert \sum_{i\in R} \frac{1}{p_i} \gamma_q\left(t_i,y_i\right)
-
\sum_{i\in \left[n\right]} \gamma_q\left(t_i,y_i\right) \right\vert
=
\left\vert \sum_{i\in R\setminus T} \frac{1}{p_i} \gamma_q\left(t_i,y_i\right)
-
\sum_{i\in \left[n\right]\setminus T} \gamma_q\left(t_i,y_i\right) \right\vert.
\]
Therefore 
\[
\left\vert \sum_{i\in R\setminus T} \frac{1}{p_i} \gamma_q\left(t_i,y_i\right) - \sum_{i\in \left[n\right]\setminus T} \gamma_q\left(t_i,y_i\right) \right\vert
\leq
\epsilon \sum_{i\in \left[n\right]\setminus T} \gamma_q\left(t_i,y_i\right)
\]
implies
\[
\left\vert \sum_{i\in R} \frac{1}{p_i} \gamma_q\left(t_i,y_i\right) - \sum_{i\in \left[n\right]} \gamma_q\left(t_i,y_i\right) \right\vert \leq \epsilon \sum_{i\in \left[n\right]} \gamma_q\left(t_i,y_i\right).
\]
First note that
\[
\expec{}{\sum_{i\in R} \frac{1}{p_i} \gamma_q(t_i,y_i)} = \sum_{i\in \left[n\right]} \gamma_q(t_i,y_i),
\]
therefore by Chernoff bound (Lemma \ref{lemma:chernoff-bound}),
for the variance value $\sigma$ chosen so that
\[
\sigma^2 = \sum_{i\in [n]} \frac{1}{p_i} \left(\gamma_q(t_i,y_i)\right)^2.
\]
we have
\[
\prob{}{\sum_{i\in R} \frac{1}{p_i} \gamma_q\left(t_i,y_i\right)
\leq \left(1-\epsilon\right)
\sum_{i\in \left[n\right]} \gamma_q\left(t_i,y_i\right)}
\leq
\exp\left(
-\frac{\left(\epsilon \sum_{i\in \left[n\right]}
\gamma_q\left(t_i,y_i\right)\right)^2}
{2\sigma^2}
\right).
\]
Hence by the variance bound from Theorem \ref{thm:main-variance-bound},
\begin{align}
\label{eq:lower-tail-chernoff}
\prob{}{\sum_{i\in R} \frac{1}{p_i} \gamma_q(t_i,y_i) \leq \frac{1}{2} \sum_{i\in [n]} \gamma_q(t_i,y_i)} \leq \exp(-\frac{h\epsilon^2}{2c})
\end{align}
Moreover, by Chernoff bound (Lemma \ref{lemma:chernoff-bound}) we have
\[
\prob{}{\sum_{i\in R} \frac{1}{p_i} \gamma_q(t_i,y_i)
  \geq \left(1+\epsilon\right)
  \sum_{i\in \left[n\right]} \gamma_q(t_i,y_i)}
\leq
\exp\left(-\frac{\left(\epsilon \sum_{i\in \left[n\right]} \gamma_q\left(t_i,y_i\right)\right)^2}
{2\left(\sigma^2 +  \frac{\alpha}{3}\cdot\epsilon
\sum_{i\in \left[n\right]} \gamma_q\left(t_i,y_i\right)\right)}\right),
\]
where
\[
\alpha
=
\max_{i\in \left[n\right]}
\frac{1}{p_i} \gamma_q\left(t_i,y_i\right).
\]
We show that
$\alpha \leq  \frac{16}{h}\sum_{i\in [n]} \gamma_q(t_i,y_i)$
by showing that
\[
\max_{i\in \left[n\right]}
\frac{1}{p_i} \gamma_q\left(t_i/2^j,y_i/2^j\right)
\leq 16 \frac{1}{h}\sum_{i\in \left[n\right]}
\gamma_q\left(t_i/2^j,y_i/2^j\right)
\]
This is sufficient because we have $\gamma_q(t_i,y_i)=(2^{j})^q\gamma_q(t_i/2^j,y_i/2^j)$. 

Therefore, from here on out in the proof we assume $\frac{1}{2}\leq t_i \leq 1$, for all $i\in [n]$.
Once again, we split the entries based on small and
large values, let:
\begin{align*}
S
& \defeq
\{i\in \left[n\right]: \abs{y_i}\leq t_i\}\\
\overline{S}
& \defeq
\left[n\right]\setminus S,
\end{align*}
this thresholding allows us to bound the norm of $y$ via:
\begin{equation}
\label{eq:concentration-thm-norm-bound}
\norm{y}_2 \leq \norm{y_S}_2 + \norm{y_{\overline{S}}}_2 \leq \norm{y_S}_2 + \norm{y_{\overline{S}}}_q,
\end{equation}
where the second inequality follows from $1\leq q\leq 2$. Moreover note that we only consider the entries where $p_i=h\sqrt{\tau_i}$ because the indices where $p_i=1$ do not contribute to the second moment and only increase the mean which means adding them only improves the probability of concentration. We have two cases.

\textbf{Case 1.} $\norm{y_{\overline{S}}}_q \geq \norm{y_S}_2$ and $\norm{y_{\overline{S}}}_q\geq 0.5$. First note that because $|y_i|\geq t_i \geq \frac{1}{2}$, for all $i\in\overline{S}$, if $\norm{y_{\overline{S}}}_q< 0.5$, then $\overline{S}=\emptyset$. We deal with the $\overline{S}=\emptyset$ later in the proof.

In this case $\norm{y}_2 \leq 2\norm{y_{\overline{S}}}_q$. Moreover, for all $i\in \overline{S}$, we have
\begin{align*}
\frac{1}{p_i} \gamma_q(t_i, y_i) & = \frac{1}{p_i} \left(|y_i|^q + \left( \frac{q}{2} - 1 \right) t_i^q\right) \leq \frac{1}{p_i} |y_i|^q =
\frac{1}{h\sqrt{\tau_i}} |y_i|^q \leq \frac{1}{h\sqrt{\tau_i}} \tau_i^{q/2} \norm{y}_2^q \\ & \leq 2^q\frac{1}{h}\norm{y_{\overline{S}}}_q^q \leq 4 \frac{1}{h} \left(\sum_{i\in \overline{S}}|y_i|^q\right) \leq 8 \frac{1}{h} \left(\sum_{i\in \overline{S}}|y_i|^q + \left(\frac{q}{2} - 1\right) t_i^q\right) \leq 8 \frac{1}{h}\sum_{i\in [n]} \gamma_q(t_i,y_i),
\end{align*}
where the first inequality follows from $1\leq q\leq 2$ and $t_i\geq 0$. The second inequality follows from Lemma \ref{lemma:lev-score-uniform}. The third inequality follows from Equation~\eqref{eq:concentration-thm-norm-bound} and the case assumption. The fourth inequality holds because $1\leq q \leq 2$. The fifth inequality holds because $1\leq q \leq 2$ and $t_{i}\leq |y_i|$ for all $i\in \overline{S}$.

For $i\in S$, we have
\begin{align*}
\frac{1}{p_i} \gamma_q(t_i, y_i) & = \frac{1}{p_i} \left(\frac{q}{2} t_i^{q-2}y_i^2 \right) \leq 2\frac{1}{p_i} y_i^2 = 2\frac{1}{h\sqrt{\tau_i}} y_i^2 \leq 2\frac{1}{h\sqrt{\tau_i}} |y_i| \leq 2\frac{1}{h\sqrt{\tau_i}} \sqrt{\tau_i}\norm{y}_2 \\ & \leq 4\frac{1}{h}\norm{y_{\overline{S}}}_q \leq 2^{1+q} \frac{1}{h}\norm{y_{\overline{S}}}_q^q \leq 8 \frac{1}{h} \left(\sum_{i\in \overline{S}}|y_i|^q\right) \leq 16 \frac{1}{h} \left(\sum_{i\in \overline{S}}|y_i|^q + \left(\frac{q}{2} - 1\right) t_i^q\right) \\ &\leq 16 \frac{1}{h}\sum_{i\in [n]} \gamma_q(t_i,y_i),
\end{align*}
where the first inequality holds because $1\leq q\leq 2$, and $\frac{1}{2} \leq t_i\leq 1$ for all $i\in S$. The second inequality holds because $|y_i|\leq t_i \leq 1$ for all $i\in S$. The third inequality follows from Lemma \ref{lemma:lev-score-uniform}. The fourth inequality follows from \eqref{eq:concentration-thm-norm-bound} and the case assumption. The fifth inequality follows from $\norm{y_{\overline{S}}}_q \geq 0.5$ and $q\geq 1$. The sixth inequality follows from $q\leq 2$. The seventh inequality holds because $1\leq q \leq 2$ and $t_{i}\leq |y_i|$ for all $i\in \overline{S}$.

\textbf{Case 2.} $\norm{y_S}_2 \geq \norm{y_{\overline{S}}}_q$. In this case $\norm{y}_2 \leq 2 \norm{y_S}_2$. For $i\in\overline{S}$, we have
\begin{align*}
\frac{1}{p_i} \gamma_q(t_i, y_i) & = \frac{1}{p_i} \left(|y_i|^q + \left( \frac{q}{2} - 1 \right) t_i^q\right) \leq \frac{1}{p_i} |y_i|^q \leq 2^{2-q} \frac{1}{p_i} y_i^2  \leq 
2 \frac{1}{h\sqrt{\tau_i}} y_i^2 \leq 2\frac{1}{h\sqrt{\tau_i}} \tau_i \norm{y}_2^2 \\ &
\leq 2\frac{1}{h}\norm{y}_2^2 \leq 8\frac{1}{h}\norm{y_S}_2^2 = 8 \frac{1}{h} \left( \sum_{i\in S} y_i^2 \right) \leq 16 \frac{1}{h} \left( \sum_{i\in S} \frac{q}{2} t_i^{q-2} y_i^2 \right) \\ & \leq 16 \frac{1}{h}\sum_{i\in [n]} \gamma_q(t_i,y_i),
\end{align*}
where the first inequality follows from $1\leq q\leq 2$ and $t_i\geq 0$. The second inequality holds because $|y_i|\geq t_i \geq \frac{1}{2}$ for all $i\in \overline{S}$. The third inequality holds because $1\leq q \leq 2$. The fourth inequality follows from Lemma \ref{lemma:lev-score-uniform}. The fifth inequality holds because $0\leq \tau_i\leq 1$. The sixth inequality follows from Equation~\eqref{eq:concentration-thm-norm-bound} and the case assumption. The seventh inequality follows from $0.5 \leq t_i \leq 1$ and $1\leq q \leq 2$.

For $i\in S$, we have
\begin{align*}
\frac{1}{p_i} \gamma_q(t_i, y_i) & = \frac{1}{p_i} \left(\frac{q}{2} t^{q-2}y_i^2 \right) \leq 2\frac{1}{p_i} y_i^2 = 2 \frac{1}{h \sqrt{\tau_i}} y_i^2 \leq 2 \frac{1}{h \sqrt{\tau_i}} \tau_i \norm{y}_2^2 \\ &
\leq 2\frac{1}{h}\norm{y}_2^2 \leq 8\frac{1}{h}\norm{y_S}_2^2 = 8 \frac{1}{h} \left( \sum_{i\in S} y_i^2 \right) \leq 16 \frac{1}{h} \left( \sum_{i\in S} \frac{q}{2} t_i^{q-2} y_i^2 \right) \\ & \leq 16 \frac{1}{h}\sum_{i\in [n]} \gamma_q(t_i,y_i),
\end{align*}
where the first inequality holds because $1\leq q\leq 2$, and $\frac{1}{2} \leq t_i\leq 1$ for all $i\in S$. The second inequality follows from Lemma \ref{lemma:lev-score-uniform}. The third inequality holds because $0\leq \tau_i \leq 1$. The fourth inequality follows from \eqref{eq:concentration-thm-norm-bound} and the case assumption. The fifth inequality follows from $0.5 \leq t_i \leq 1$ and $1\leq q \leq 2$.

Now we consider the case where $\overline{S}=\emptyset$. In this case we only need to address $i\in S = [n]$, for which we have
\begin{align*}
\frac{1}{p_i} \gamma_q(t_i, y_i) & = \frac{1}{p_i} \left(\frac{q}{2} t^{q-2}y_i^2 \right) \leq 2\frac{1}{p_i} y_i^2 = 2 \frac{1}{h \sqrt{\tau_i}} y_i^2 \leq 2 \frac{1}{h \sqrt{\tau_i}} \tau_i \norm{y}_2^2 \\ &
\leq 2\frac{1}{h}\norm{y}_2^2 = 2 \frac{1}{h}\norm{y_S}_2^2 = 2 \frac{1}{h} \left( \sum_{i\in S} y_i^2 \right) \leq 4 \frac{1}{h} \left( \sum_{i\in S} \frac{q}{2} t_i^{q-2} y_i^2 \right) \\ & \leq 4 \frac{1}{h}\sum_{i\in [n]} \gamma_q(t_i,y_i),
\end{align*}
where the first inequality holds because $1\leq q\leq 2$, and $\frac{1}{2} \leq t_i\leq 1$ for all $i\in S$. The second inequality follows from Lemma \ref{lemma:lev-score-uniform}. The third inequality holds because $0\leq \tau_i \leq 1$. The fourth inequality follows from \eqref{eq:concentration-thm-norm-bound} and the case assumption. The fifth inequality follows from $0.5 \leq t_i \leq 1$ and $1\leq q \leq 2$.

Therefore by the above argument and case analysis, we have
\[
\alpha \leq 16 \frac{1}{h}\sum_{i\in [n]} \gamma_q(t_i/2^j,y_i/2^j).
\]
Hence by Chernoff bound (Lemma \ref{lemma:chernoff-bound}) and Theorem~\ref{thm:main-variance-bound}, we have
\begin{equation}
\label{eq:upper-tail-chernoff-2}
\prob{}{\sum_{i\in R} \frac{1}{p_i} \gamma_q(t_i,y_i) \geq (1+\epsilon) \sum_{i\in [n]} \gamma_q(t_i,y_i)} \leq \exp\left(-\frac{h\epsilon^2}{2c+16\epsilon/3}\right) \leq \exp\left(-\frac{h\epsilon^2}{2c+16/3}\right),
\end{equation}
where the last inequality follows from $0\leq \epsilon \leq 1$.
Combining the lower bound from
Equation~\eqref{eq:lower-tail-chernoff} and
the upper bound from
Equation~\eqref{eq:upper-tail-chernoff-2}
gives that the result follows
from a constant choice of $\bar{c} \leftarrow 2c + 16/3$.
\end{proof}

\subsubsection{Sampling Algorithm and $\epsilon$-Net}
In this section, we present our algorithm for sampling the $\gamma_q$ function for $1<q\leq 2$. Our sampling is based on leverage scores. However, we use the square root of the leverage scores as opposed to the classical application. This results in $\Otil(\sqrt{n d})$ samples instead of $\Otil(d)$ {\em but with smaller variance}.
Therefore we utilize an iterative approach to sample that decrease the number of samples in each iteration. We perform this iterative algorithm for $\log\log(n)$ iterations. Our analysis to show the concentration of our sampling scheme is based on $\epsilon$-nets.

\begin{algbox}
\textbf{Algorithm 1 - Iterative Sampling for $\gamma_q$ function with Square Root of Leverage Scores}

\textbf{Input: matrix $A\in \mathbb{R}^{n\times d}$, vector $t\in \mathbb{R}^{n}$, oversampling parameter $h$}
\begin{enumerate}
    \item Set $t^{(1)} = t$.
    \item Set $T^{(1)} = [n]$.
    \item Set $A^{(1)} = A$
    \item Set $z = \lceil\log\log(n) \rceil$.
    \item For $k=1,\ldots,z$ do
    \begin{enumerate}
        \item Set $\beta = \min_{i\in T^{(k)}} t^{(k)}_i$ and $\eta^{(k)} = \left\lceil\log\left(\frac{\max_{i\in T^{(k)}} t^{(k)}_i}{ \beta}\right)\right\rceil + 1$.
        
        \item For $j\in[\eta^{(k)}]$, set $T^{(k)}_j = \{i \in T^{(k)}: 2^{j-1} \beta \leq t^{(k)}_i < 2^j \beta\}$
        
        \item For $j\in[\eta^{(k)}]$, let $A^{(k,j)}$ be the matrix consisting of rows of $A^{(k)}$ in $T^{(k)}_j$.
        
        \item For each $i\in T^{(k)}$, set $\tau_i^{(k)} = a_i^\top \left(\left( A^{(k,j)} \right)^\top A^{(k,j)} \right)^+ a_i$, where $j$ is the index such that $i\in T_j^{(k)}$, i.e., the leverage score of row $i$ in matrix $A^{(k,j)}$.
        
        \item For $j\in [\eta^{(k)}]$, form set $S_j^{(k)}$ by sampling each member $i\in T_j^{(k)}$ independently with probability $p_i^{(k)}:=\min\left\{1, h \sqrt{\tau_i^{(k)}}\right\}$.
        
        \item Set $T^{(k+1)} = \bigcup_{j\in[\eta^{(k)}]} S_j^{(k)}$.
        
        \item For $i \in T^{(k+1)}$, set $t_i^{(k+1)} = \left(\frac{1}{p_i^{(k)}}\right)^{1/q} t_i^{(k)}$.
        
        \item Set $A^{(k+1)}$ to a matrix with rows in $T^{(k+1)}$ such that for $i \in T^{(k+1)}$, $A^{(k+1)}_{i:} = \left(\frac{1}{p_i^{(k)}}\right)^{1/q} A^{(k)}_{i:}$
        
    \end{enumerate}
    
    \item For $i \in T^{(z+1)}$, set $w_i = \prod_{k=1}^z\frac{1}{p_i^{(k)}}$, and for $i\in[n]\setminus T^{(z+1)}$, set $w_i = 0$.
    
    \item Return $T^{(z+1)}$ and $w$.
    \end{enumerate}
    \end{algbox}

\begin{theorem}
\label{thm:one-point-prob}
Let $t\in\mathbb{R}^n_{\geq 0}$, $A\in\mathbb{R}^{n\times d}$, $x\in\mathbb{R}^d$, $y= Ax$, and $1\leq q\leq 2$. Let $h \in [1, n]$ be an oversampling parameter. Let $T^{(z+1)}$ and $w$ be the outputs of Algorithm 1 for $A$, $t$, and $h$. Moreover, suppose the condition number of $A$ and $\max_{i} t_i / \min_{i} t_i$ are bounded by a polynomial in $n$.
Then for some absolute constant $c$ we have
\[
\prob{}{\left\vert \sum_{i \in \left[n\right]}
w_i \gamma_q\left(t_i,y_i\right)
-
\sum_{i\in \left[n\right]} \gamma_q\left(t_i,y_i\right) \right\vert
\leq
\frac{1}{2} \sum_{i\in \left[n\right]} \gamma_q\left(t_i,y_i\right)
}
\geq
1 - O\left(\log^5{n} \exp\left(-\frac{h}{c\log^2(n)}\right)\right).
\]
Moreover,
$|T^{(z+1)}| \leq O( h^2 d \log^8{n})$
with high probability.
\end{theorem}

\begin{proof}
Let $y^{(k)}=A^{(k)}x$.

We first obtain crude condition number bounds of
$A^{(k)}$'s produced in the algorithm.
The condition number of $A$ is bounded by $n^{O(1)}$,
the algorithm iterates for $\log\log(n)$ iterations,
and in each iteration we scale rows of the previous matrix
by $\frac{1}{h\sqrt{\tau^{(k)}_i}}$.
Furthermore, $1/\tau_i^{(k)}$ is bounded by the
condition number of the previous matrix.
Therefore, the condition number of $A^{(k)}$,
as well as $\max_{i} t_i^{(k)} / \min_{i} t_i^{(k)}$ are
both at most 
\[
n^{O(\log n)}\cdot h^{O(\log\log(n))}.
\]
This in turn implies that the number of buckets, i.e. $\eta^{(k)}$,
in each iteration is $O(\log^2(n)\log\log(n)\log(h)) \leq O(\log^{5}n)$.

By Theorem \ref{thm:concentation},
for each such bucket $j$ and $k\in [z]$,
we have
\[
\prob{}{ \left(1-\epsilon\right) \sum_{i\in T_j^{\left(k\right)}} \gamma_q\left(t_i^{(k)},y_i^{\left(k\right)}\right)
\leq
\sum_{i\in S_j^{\left(k\right)}} \frac{1}{p_i^{\left(k\right)}}
\gamma_q(t_i^{\left(k\right)},y_i^{\left(k\right)})
\leq
\left(1+\epsilon\right) \sum_{i\in T_j^{\left(k\right)}}
\gamma_q\left(t_i^{\left(k\right)},y_i^{\left(k\right)}\right)}
\geq
1
-
2\exp\left(-\frac{h\epsilon^2}{c}\right)
\]
Union bounding over the $O(\log^{5}n)$ buckets
of different values of $t_i$ gives
that with probability at least $1 - O(\log^5{n}) \exp( -\frac{h\epsilon^2}{c})$ we have:
\[
\left(1-\epsilon\right)
  \sum_{i\in T^{\left(k\right)}}
    \gamma_q\left(t_i^{\left(k\right)},y_i^{\left(k\right)}\right)
  \leq
  \sum_{i\in T^{\left(k+1\right)}}
  \frac{1}{p_i^{\left(k\right)}}
  \gamma_q\left(t_i^{\left(k\right)},y_i^{\left(k\right)}\right)
  \leq
  \left(1+\epsilon\right) \sum_{i\in T^{\left(k\right)}} \gamma_q\left(t_i^{\left(k\right)},y_i^{\left(k\right)}\right)
\]
Note that by definition of the $\gamma_q$ function
and the setting of the algorithm
\[
\frac{1}{p_i^{\left(k\right)}}
\gamma_q\left(t_i^{\left(k\right)},y_i^{\left(k\right)}\right)
=
\gamma_q\left(
  \left(\frac{1}{p_i^{(k)}}\right)^{1/q} t_i^{\left(k\right)},
  \left(\frac{1}{p_i^{(k)}}\right)^{1/q} y_i^{\left(k\right)}\right)
=
\gamma_q\left(t_i^{\left(k+1\right)},y_i^{\left(k+1\right)}\right)
\]
Therefore by union bounding over all
$O(\log\log{n}) \leq O(\log{n})$ iterations,
we have that with probability at least
$1 - O(\log^6{n}) \exp( -\frac{h\epsilon^2}{c})$
that the following holds for all $k \in [z]$:
\[
\left(1-\epsilon\right)^{\log\log{n}}
\sum_{i\in T^{\left(k\right)}}
\gamma_q\left(t_i^{\left(k\right)},y_i^{\left(k\right)}\right)
\leq
\sum_{i\in T^{\left(k+1\right)}}
\gamma_q\left(t_i^{\left(k+1\right)},y_i^{\left(k+1\right)}\right)
\leq
\left(1+\epsilon\right)^{\log\log{n}}
\sum_{i\in T^{\left(k\right)}}
\gamma_q\left(t_i^{\left(k\right)},y_i^{\left(k\right)}\right)
\]
Now note that we can pick $\epsilon=\theta(\log(n))$ such that $(1+\epsilon)^{\log\log(n)}\leq \frac{3}{2}$ and $(1-\epsilon)^{\log\log(n)}\geq \frac{1}{2}$.
Then the first part of the result follows by noting that $T^{(1)}=[n]$, $t^{(1)}=t$, $y^{(1)}=y$ and $\gamma_q(t_i^{(z+1)},y_i^{(z+1)})=w_i \gamma_q(t_i,y_i)$.

Now we bound the size of $T^{(z+1)}$.
Applying Cauchy-Schwarz to its expected size gives:
\[
\expec{}{\abs{T^{\left(k+1\right)}}}
=
\sum_{j\in \left[\eta^{\left(k\right)}\right]}
\sum_{i\in T_j^{\left(k\right)}} p_i^{\left(k\right)}
\leq
\sum_{j\in \left[\eta^{\left(k\right)}\right]}
\sum_{i\in T_j^{\left(k\right)}} h \sqrt{\tau_i^{\left(k\right)}}
\leq
\sum_{j\in \left[\eta^{\left(k\right)}\right]}
h \sqrt{\abs{T_j^{\left(k\right)}}}
\sqrt{\sum_{i\in T_j^k} \tau_i^{\left(k\right)}}.
\]
Incorporating in
$T_j^{(k)}\subseteq T^{(k)}$
and
Lemma \ref{lemma:sum-leverage-bound} then gives:
\[
\leq h \eta^{\left(k\right)} \sqrt{\abs{T^{\left(k\right)}}}\sqrt{d}
\leq
O \left( \log^{4}n \sqrt{d}\sqrt{\abs{T^{\left(k\right)}}} \right)
\]
where the last inequality follows from the argument on
the number of buckets in the previous part of the proof. 

Therefore by Chernoff bound (Lemma \ref{lemma:chernoff-bound}),
with probability at least $1 - \exp(h \sqrt{d})$, we have
\[
\abs{T^{\left(k+1\right)}}
\leq O \left( \log^4{n} h \sqrt{d}\sqrt{\abs{T^{\left(k\right)}}} \right).
\]
Doing union bound on this over all iterations gives that
with a probability of at least
$1-\log\log(n)\exp(h\sqrt{d})$, we have for all $k\in[z]$:
\[
\abs{T^{\left(k+1\right)}}
\leq
O \left( \log^4{n} h \sqrt{d} \sqrt{\abs{T^{\left(k\right)}}}
\right).
\]

We aggregate this via induction to show that
$|T^{(z+1)}| \leq O (\log^{8}n h^2 d)$.
Specifically, our induction hypothesis is:
\[
\abs{T^{\left(k+1\right)}}
\leq
O \left(\left( \log^4{n} h \sqrt{d}\right)^{2-1/2^k} n^{1/2^k} \right).
\]
The base case for $k=0$ is trivially true as $\abs{T^{\left(1\right)}}=n$.
For the inductive case, we have:
\[
\abs{T^{\left(k+2\right)}}
\leq
O \left( \log^{4}n h \sqrt{d}
\sqrt{\left( \log^4{n} h \sqrt{d}\right)^{2-1/2^k} n^{1/2^k}} \right) 
\leq O \left(\left( \log^{4}{n} h \sqrt{d}\right)^{2-1/2^{k+1}} n^{1/2^{k+1}} \right),
\]
which means the hypothesis holds for $k + 1$ as well.
Applying this with $k = z$ then gives:
\[
\abs{T^{(z+1)}}
\leq
O
\left(
\left( \log^4{n} h \sqrt{d}\right)^{2} n^{1/2^{\log\log{n}+1}} \right)
\leq
O \left(\left( \log^{4}n h \sqrt{d}\right)^{2}
n^{\frac{1}{\log{n}}} \right)
\leq
O \left(\left( \log^{4}n h \sqrt{d}\right)^{2} \right).
\]
\end{proof}

We can use Theorem \ref{thm:one-point-prob} to prove concentration for the points in an $\epsilon$-net. However to give such a bound for all points of interest, we need the following two lemmas to bound the error for all points in a polynomial range.

\begin{lemma}
\label{lemma:gamma-close}
Let $y,\tilde{y},t\in\mathbb{R}^{n}_{\geq 0}$ such that $\norm{y-\tilde{y}}_{1}\leq \alpha$, $\theta \geq\norm{y}_1$,
and $t\geq 1$. Let $1\leq q \leq 2$.
Then
\[
\left\vert\sum_{i\in\left[n\right]}
\gamma_q\left(t_i,\tilde{y}_i\right)
-
\sum_{i\in\left[n\right]}
\gamma_q\left(t_i,y_i\right)\right\vert
\leq
\sum_{i\in\left[n\right]}
\left\vert \gamma_q\left(t_i,\tilde{y}_i\right)
-
\gamma_q\left(t_i,y_i\right)
\right\vert
\leq
4n\alpha\left(\alpha+\theta\right)
\]
\end{lemma}

\begin{proof}
The first inequality of the result
will be via the triangle inequality.
Define the sets:
\begin{align*}
& S_1 = \{i\in [n]: y_i < t_i \text{ and } \tilde{y}_i < t_i \} \\ &
S_2 = \{i\in [n]: y_i \geq t_i \text{ and } \tilde{y}_i < t_i \} \\ &
S_3 = \{i\in [n]: y_i < t_i \text{ and } \tilde{y}_i \geq t_i \} \\ &
S_4 = \{i\in [n]: y_i \geq t_i \text{ and } \tilde{y}_i \geq t_i \}
\end{align*}
Then the difference we wish to bound can be decomposed into
\[
\sum_{i\in S_1} \frac{q}{2} t_i^{(q-2)} \left(\tilde{y}_i^2-y_i^2\right) + \sum_{i\in S_2} \frac{q}{2}t_i^{q-2}\tilde{y}_i^2 - y_i^q - (\frac{q}{2}-1) t_i^q
+ \sum_{i\in S_3} \tilde{y}_i^q + (\frac{q}{2}-1) t_i^q -\frac{q}{2}t_i^{q-2}y_i^2
+ \sum_{i\in S_4} \tilde{y}_i^q - y_i^q
\]

We bound these cases separately.

For $i\in S_1$,
applying $1\leq q\leq 2$, $t_i \geq 1$, and the identity
$a^2 - b^2 = (a - b)(a + b)$ gives
\[
\left\vert \frac{q}{2} t_i^{(q-2)} \left(\tilde{y}_i^2-y_i^2\right) \right\vert
\leq \left\vert \tilde{y}_i - y_i \right\vert \cdot \left\vert
\tilde{y}_i + y_i \right\vert
\]
by which incorporating the assumptions
$\norm{y-\tilde{y}}_{1}\leq \alpha$ and $\theta \geq\norm{y}_1$ leads to
\[
\leq \alpha\left(\alpha + 2\theta\right).
\]

For $i\in S_2$, we have $\tilde{y}_i \leq t_i \leq y_i$
and therefore by assumption $|\tilde{y}_i-t_i|,|t_i-y_i|,|\tilde{y}_i-y_i| \leq \alpha$.
Substituting this in via triangle inequality, and
incorporating $1\leq q\leq 2$, $y_i \geq t_i \geq 1$ gives
\[
\left\vert \frac{q}{2}t_i^{q-2}
\tilde{y}_i^2 - y_i^q - \left(\frac{q}{2}-1\right) t_i^q \right\vert
\leq
\left\vert \frac{q}{2}t_i^{q-2}\left(\tilde{y}_i^2 -t_i^2\right) \right\vert
+ \left\vert y_i^q - t_i^q \right\vert
\]
The identity
$\left(\tilde{y}_i^2-y_i^2\right)=(\tilde{y}_i-y_i)(\tilde{y}_i+y_i)$
and the assumptions
$\norm{y-\tilde{y}}_{1}\leq \alpha$ and $\theta \geq\norm{y}_1$
then give
\[
\leq \left\vert \tilde{y}_i^2 - t_i^2 \right\vert
+ \left\vert y_i^2 - t_i^2 \right\vert
\leq
\left\vert \tilde{y}_i - t_i \right\vert \cdot \left\vert \tilde{y}_i + t_i \right\vert + \left\vert y_i - t_i \right\vert \cdot \left\vert y_i + t_i \right\vert \leq 4\alpha \theta.
\]

For $i\in S_3$,
triangle inequality
and $1\leq q\leq 2$, $\tilde{y}_i \geq t_i \geq 1$ give
\[
\left\vert \tilde{y}_i^q + (\frac{q}{2}-1) t_i^q -\frac{q}{2}t_i^{q-2}y_i^2 \right\vert
\leq
\left\vert \frac{q}{2}t_i^{q-2}(y_i^2 -t_i^2) \right\vert + \left\vert \tilde{y}_i^q - t_i^q \right\vert
\leq
\left\vert \tilde{y}_i^2 - t_i^2 \right\vert + \left\vert y_i^2 - t_i^2 \right\vert
\]
Also, because $y_i \leq t_i \leq \tilde{y}_i$,
the assumption also implies
$|\tilde{y}_i-t_i|,|t_i-y_i|,|\tilde{y}_i-y_i| \leq \alpha$.
Combining this with
$(\tilde{y}_i^2-y_i^2)=(\tilde{y}_i-y_i)(\tilde{y}_i+y_i)$
then gives
\[
\leq
\left\vert \tilde{y}_i - t_i \right\vert \cdot \left\vert \tilde{y}_i+ t_i \right\vert
+
\left\vert y_i - t_i \right\vert \cdot \left\vert y_i + t_i \right\vert
\leq
4\alpha\left(\alpha + \theta\right).
\]

For $i\in S_4$, $y_i,\tilde{y}_i \geq t_i \geq 1$ gives
\[
\left\vert \tilde{y}_i^q - y_i^q \right\vert
\leq
\left\vert \tilde{y}_i^2 - y_i^2 \right\vert
\]
after which factorizing the squares and incorporating
the assumptions
$\norm{y-\tilde{y}}_{1}\leq \alpha$ and $\theta \geq\norm{y}_1$
gives
\[
=
\left\vert \tilde{y}_i - y_i \right\vert
\cdot
\left\vert \tilde{y}_i + y_i \right\vert
\leq
\alpha \left(\alpha + 2\theta\right).
\]

Combining all of the above gives
\[
\left\vert\sum_{i\in\left[n\right]}
\gamma_q\left(t_i,\tilde{y}_i\right)
-
\sum_{i\in\left[n\right]}
\gamma_q\left(t_i,y_i\right)\right\vert
\leq
4n\alpha\left(\alpha+\theta\right).
\]
\end{proof}

Rounding to an $\epsilon$-net incurs
additive differences.
In order to convert this additive difference into
a multiplicative error, we need the following crude
lower bound on the value of the gamma function.

\begin{lemma}
\label{lemma:lower-bound-gamma}
Let $y,t\in\mathbb{R}^n_{\geq 0}$
such that 
$1\leq t\leq \beta$, then
\[
\sum_{i\in\left[n\right]}
\gamma_q\left(t_i,y_i\right)
\geq
\min\left\{\frac{1}{8\beta} \norm{y}_2^2,
\frac{1}{8n} \norm{y}_q^q\right\}.
\]
\end{lemma}

\begin{proof}
Let $S=\{i\in [n]: |y_i|\leq t_i\}$
and
$\overline{S} = [n]\setminus S$ its complement.
We have two cases.

\textbf{Case 1.}
$\norm{y_S}_2 \geq \norm{y_{\overline{S}}}_q$.
In this case it suffices to lower bound
\[
\sum_{i\in\left[n\right]} \gamma_q\left(t_i,y_i\right)
\geq
\sum_{i\in S} \gamma_q\left(t_i,y_i\right)
=
\sum_{i\in S} \frac{q}{2} t_i^{q-2} y_i^2.
\]
Incorporating $1\leq q \leq 2$ and $t_i \leq \beta$,
the case assumption, and $1\leq q \leq 2$ then gives
\begin{align*}
\geq
\frac{1}{2\beta}\sum_{i\in S} y_i^2
& =
\frac{1}{8\beta} (2\norm{y_S}_2)^2
\geq
\frac{1}{8\beta}
\left(\norm{y_S}_2 + \norm{y_{\overline{S}}}_q\right)^2
\\ &
\geq
\frac{1}{8\beta}
\left(\norm{y_S}_2 + \norm{y_{\overline{S}}}_2\right)^2
\geq \frac{1}{8\beta}
\left(\norm{y_S}_2^2 + \norm{y_{\overline{S}}}_2^2\right)
= \frac{1}{8\beta} \norm{y}_2^2.
\end{align*}

\textbf{Case 2.} $\norm{y_{\overline{S}}}_q \geq \norm{y_S}_2$.
In this case, we lower bound
\[
\sum_{i\in\left[n\right]}
\gamma_q\left(t_i,y_i\right)
\geq
\sum_{i\in \overline{S}} \gamma_q\left(t_i,y_i\right)
=
\sum_{i\in \overline{S}} \abs{y_i}^q
+
\left(\frac{q}{2} - 1\right) t_i^q
\]
Because $1\leq q \leq 2$ and $t_i\leq |y_i|$
for all $i\in \overline{S}$,
this simplifies to lower bounding
\[
\geq \frac{1}{2}\sum_{i\in \overline{S}} y_i^q
\geq \frac{1}{8} \left(2\norm{y_{\overline{S}}}_q\right)^q.
\]
Substituting in the case assumption then gives
\begin{align*}
\geq \frac{1}{8} \left(\norm{y_S}_2 + \norm{y_{\overline{S}}}_q\right)^q
&
\geq \frac{1}{8} \left(\frac{1}{\sqrt{n}}\norm{y_S}_q + \norm{y_{\overline{S}}}_q\right)^q
\geq \frac{1}{8n} \left(\norm{y_S}_q + \norm{y_{\overline{S}}}_q\right)^q
\\ &
\geq  \frac{1}{8n} \left(\norm{y_S}_q^q + \norm{y_{\overline{S}}}_q^q\right)
=
\frac{1}{8n} \norm{y}_q^q.
\end{align*}

The result follows by combining the cases.
\end{proof}

Now, we are equipped to bound the concentration for all points in a polynomial range. 

\begin{theorem} 
\label{thm:main-eps-net}
Let
$t\in\mathbb{R}^n$, $A\in\mathbb{R}^{n\times d}$ such that $t\geq 1$ and $\beta$ be an upper bound for the condition number of $A$ and $\max_{i} t_i$.
For constants $\alpha$, and $c$ as given in
Theorem \ref{thm:one-point-prob}, and a sampling overhead
\[
h
\leftarrow
\alpha \log \left(2\beta n\right) \log^{3}{n} \cdot d,
\]
let $T^{(z+1)}$ and $w$ be the outputs of Algorithm 1
for $A$, $t$, and $h$, then with probability at least $1 - O(\exp(-d))$
we have for all $y=Ax$ such that
$\norm{x}_2\in (\frac{1}{\beta n^{10}}, \beta n^{10})$:
\[
\left\vert \sum_{i \in \left[n\right]}
w_i \gamma_q\left(t_i,y_i\right)
-
\sum_{i\in \left[n\right]} \gamma_q\left(t_i,y_i\right)\right\vert
\leq 
\frac{3}{4}\sum_{i\in \left[n\right]} \gamma_q\left(t_i,y_i\right).
\]
\end{theorem}

\begin{proof}
We form an $\epsilon$-net $N$ in $\mathbb{R}^d$
with granularity of $\frac{1}{\beta^{9} n^{36\log\log n}}$
on the coordinates where for any $x\in N$, $|x_i| \leq \beta n^{10}$.
It has size at most
\[
\left(2\beta^{10} n^{36\log\log{n}+10}\right)^{d}
\]

For each point $x\in N$, by Theorem \ref{thm:one-point-prob}
and our choice of $h$, we have that
\[
\frac{1}{\log{n}}
\sum_{i\in \left[n\right]} \gamma_q\left(t_i,y_i\right)
\leq
\sum_{i\in T^{\left(z + 1\right)}} w_i \gamma_q\left(t_i,y_i\right)
\leq
\log{n} \sum_{i\in \left[n\right]} \gamma_q\left(t_i,y_i\right).
\]
with probability at least
\[
1 - 
O\left( \log^{5}n \exp\left( - \frac{h}{c \log^2{n}} \right) \right)
\geq
1
-
O\left(\log^5{n} \left(\frac{1}{2\beta n} \right)^{\alpha d\log\log(n)}\right).
\]

Therefore by union bounding over all the points of $N$, we have
\[
\left\vert \sum_{i \in [n]} w_i \gamma_q\left(t_i,y_i\right)
- \sum_{i\in \left[n\right]} \gamma_q\left(t_i,y_i\right)
\right\vert
\leq
\frac{1}{2}\sum_{i\in \left[n\right]} \gamma_q\left(t_i,y_i\right)
\]
for all $y=Ax$ such that $x\in N$ with probability at least
\[
1 - \abs{N}
O\left(\log^{5}n \right) \left(\frac{1}{2\beta n} \right)^{\alpha d\log\log{n}}
\geq
1 - O\left(\log^5{n} \exp\left(
-\left(\alpha - 37\right) d\log\log{n}
\log\left(2\beta n\right)\right)\right)
\]
Picking $\alpha$ to be a large enough constant,
the above probability is at least $1-O(\exp(-d))$.

It remains to bound the additional additive errors from
rounding points onto $N$.
Consider a point $x\in\mathbb{R}^d$ not in $N$ such that $\norm{x}_2\in(\frac{1}{\beta n^{10}}, \beta n^{10})$.
This bound gives $\|Ax\|_2 \leq \beta^2 n^{10}$,
and also there exits $\tilde{x}\in N$ such that
\[
\norm{A\left(x-\tilde{x}\right)}_1
\leq
\sqrt{n} \norm{A\left(x-\tilde{x}\right)}_2
\leq
\frac{1}{\beta^8 n^{36\log\log{n}-2}}.
\]

Let $y=Ax$ and $\tilde{y}= A\tilde{x}$.
Therefore by Lemma \ref{lemma:gamma-close},
\begin{equation}
\label{eq:close-points}
\left\vert\sum_{i\in \left[n\right]}
\gamma_q\left(t_i,\tilde{y}_i\right)
-
\sum_{i\in\left[n\right]} \gamma_q\left(t_i,y_i\right)\right\vert
\leq
\frac{8}{\beta^6 n^{36\log\log(n)-13}}
\end{equation}
Moreover by the assumption on the condition number of $A$ and $\norm{x}_2$,
$\norm{y}_q\geq\norm{y}_2 \geq \frac{1}{\beta^2 n^{10}}$.
So we have
$\frac{1}{8\beta}\norm{y}_2^2 \geq \frac{1}{8\beta^5 n^{20}}$
and
$\frac{1}{8n}\norm{y}_q^q \geq \frac{1}{8\beta^4 n^{21}}$,
and in turn by Lemma~\ref{lemma:lower-bound-gamma},
\begin{equation}
\label{eq:gamma-lower-bound}
\sum_{i\in[n]} \gamma_q(t_i,y_i) \geq \frac{1}{8\beta^5 n^{21}}.
\end{equation}
Now, we assume over the course of the Algorithm 1, for any $k\in[z]$, and for any $i\in T^{(k)}$, $p_i^{(k)}\geq \frac{1}{n}$. We can achieve this by slightly modifying the algorithm and taking $p_i^{(k)}=\min\left\{1,\max\left\{\frac{1}{n}, h\sqrt{\tau_i^{(k)}}\right\}\right\}$. Note that this only increases the expected value of the number of picked rows by one and because the number of iterations is $\log\log(n)$, it only increases the number of picked rows by at most about $\log(n)$ rows. Now note that by this assumption we have
\[
w_i \leq n^{\log\log(n)}.
\]

Decomposing the difference via triangle inequality gives:
\begin{align*}
&\left\vert \sum_{i \in [n]} w_i \gamma_q(t_i,y_i) - \sum_{i\in [n]} \gamma_q(t_i,y_i)\right\vert \\ & =\left\vert \sum_{i \in [n]} w_i \gamma_q(t_i,\tilde{y}_i) + \left( \sum_{i \in [n]} w_i \gamma_q(t_i,y_i) - \sum_{i \in [n]} w_i \gamma_q(t_i,\tilde{y}_i) \right)  - \sum_{i\in [n]} \gamma_q(t_i,\tilde{y}_i) - \left( \sum_{i\in [n]} \gamma_q(t_i,y_i) - \sum_{i\in [n]} \gamma_q(t_i,\tilde{y}_i) \right)\right\vert \\ &
\leq
\left\vert \sum_{i \in [n]} w_i \gamma_q(t_i,\tilde{y}_i) - \sum_{i\in \left[n\right]} \gamma_q(t_i,\tilde{y}_i)\right\vert
+
\left\vert \sum_{i \in [n]} w_i \gamma_q(t_i,y_i) - \sum_{i \in [n]} w_i \gamma_q(t_i,\tilde{y}_i) \right\vert
+
\left\vert \sum_{i\in \left[n\right]}
\gamma_q\left(t_i,y_i\right) -
\sum_{i\in \left[n\right]} \gamma_q\left(t_i,\tilde{y}_i\right) \right\vert.
\end{align*}

Now by upper bounds on $w_i$ obtained from lower
bounds on $p_i$, and the triangle inequality, we get
\begin{multline*}
\left\vert \sum_{i \in \left[n\right]}
w_i \gamma_q\left(t_i,y_i\right)
-
\sum_{i \in \left[n\right]} w_i \gamma_q\left(t_i,\tilde{y}_i\right)
\right\vert
\leq \sum_{i \in [n]} w_i \left\vert \gamma_q\left(t_i,y_i\right)
- \gamma_q\left(t_i,\tilde{y}_i\right) \right\vert
\\\leq 
n^{\log\log{n}} \sum_{i \in \left[n\right]}
\left\vert \gamma_q\left(t_i,y_i\right)
-
\gamma_q\left(t_i,\tilde{y}_i\right) \right\vert
\leq
n^{\log\log{n}}
\sum_{i\in \left[n\right]} \left\vert \gamma_q\left(t_i,y_i\right)
- \gamma_q\left(t_i,\tilde{y}_i\right) \right\vert.
\end{multline*}
So we can bound the overall error by at most
\[
\leq \left\vert \sum_{i \in [n]} w_i \gamma_q(t_i,\tilde{y}_i) - \sum_{i\in [n]} \gamma_q(t_i,\tilde{y}_i)\right\vert + \left(n^{\log\log(n)} + 1 \right) \sum_{i\in [n]} \left\vert \gamma_q(t_i,y_i) - \gamma_q(t_i,\tilde{y}_i) \right\vert
\]
under the assumption of success over the entire $\epsilon$-net,
and the associated point-wise error bound.
Incorporating the distances between a point and its
closest point on the $\epsilon$-net from
Equation~\eqref{eq:close-points} then gives:
\[
\leq \frac{1}{2}\sum_{i\in [n]} \gamma_q(t_i,\tilde{y}_i) + \frac{16}{\beta^6 n^{35\log\log{n}-13}}
\leq
\frac{1}{2}\sum_{i\in [n]} \gamma_q\left(t_i,y_i\right) + \frac{20}{\beta^6 n^{35\log\log{n}-13}}
\]
after which the final bound follows from the
lower bound on objective value from
Equation~\eqref{eq:gamma-lower-bound}:
\[
\leq
\frac{1}{2}\sum_{i\in \left[n\right]} \gamma_q\left(t_i,y_i\right)
+ \frac{1}{n} \sum_{i\in \left[n\right]} \gamma_q\left(t_i,y_i\right)
\leq \frac{3}{4}\sum_{i\in \left[n\right]} \gamma_q\left(t_i,y_i\right).
\]
\end{proof}

\subsection{Solving the $p$-norm problem via the residual problem}

In this section, we show how the $p$-norm regression problem can be solved by approximately solving instances of the following residual problem. This section is adapted from~\cite{AdilKPS19}. However we consider a more general problem as the following.

\label{sec:p-norm-via-residual}
\begin{definition}
Let $A\in\mathbb{R}^{n\times d}$, $C \in \mathbb{R}^{d\times d}$. For a $p$-norm regression problem of the form
\begin{align}
\label{eq:tall-p-norm-w-constraint}
    \min_{Cx = v} \norm{Ax - b}_p^p,
\end{align}
we define the residual problem at point $x$ as
\begin{align}
    \label{eq:residual-problem}
    \max_{C \Delta = 0} g^T A \Delta - \frac{p-1}{p 2^p} \gamma_p(|Ax - b|, A\Delta),
\end{align}
where $g=p \left(\left\vert A x - b\right\vert^{p-2} \odot (Ax - b) \right)$, where raising to the power of $p-2$ is element-wise and $\odot$ is the Hadamard (element-wise) product. Let $\Delta^*$ be the optimum solution of \eqref{eq:residual-problem}. We say $\widetilde{\Delta}$ is an $\alpha$-approximate solution if $C\widetilde{\Delta} =0$ and 
\begin{align*}
g^T A \widetilde{\Delta} - \frac{p-1}{p 2^p} \gamma_p(|Ax-b|, A\widetilde{\Delta}) \geq \frac{1}{\alpha} \left( g^T A \Delta^* - \frac{p-1}{p 2^p} \gamma_p(|Ax - b|, A\Delta^*) \right)
\end{align*}
\end{definition}

The following two results regarding the $\gamma_p$ function (Definition \ref{def:quadratic-smooth}), proved by \cite{AdilKPS19}, are useful for bounding the difference of the $p$-norm of two close points.

\begin{lemma}[\cite{AdilKPS19}]
\label{lemma:bregman-divergence-bound}
Let $1< p < \infty$. Then for any $t,y\in\mathbb{R}$,
\begin{align*}
|t|^p + ly + \frac{p-1}{p 2^p} \gamma_p(|t|, y) \leq |t+y|^{p} \leq |t|^p + ly + 2^p \gamma_p(|t|, y),
\end{align*}
where $l = p |t|^{p-2} t$ is the derivative of $|t|^p$.
\end{lemma}

\begin{lemma}[\cite{AdilKPS19}]
\label{lemma:weird-gamma}
Let $1< p < \infty$. Then for any $t,y\in\mathbb{R}$ and $\lambda>0$,
\begin{align*}
\min\{\lambda^2, \lambda^p\} \gamma_p(|t|, y) \leq \gamma_p(|t|, \lambda y) \leq \max\{\lambda^2, \lambda^p\} \gamma_p(|t|, y)
\end{align*}
\end{lemma}

The following result gives a bound on the difference of $p$-norm of two points that are close to each other based on the $\gamma_p$ function.

\begin{lemma}
\label{lemma:bound-res-value}
Let $1<p<\infty$, $A\in \mathbb{R}^{n}$, $b\in\mathbb{R}^n$, and $\lambda$ be such that $\lambda^{\min\{1,p-1\}} \leq \frac{p-1}{p 4^p}$. Then for any $x$ and $\Delta$,
\begin{align*}
\norm{Ax-b}_p^p - \lambda g^T A \Delta + \frac{p-1}{p 2^p} \gamma_p(|Ax - b|, \lambda A\Delta) & \leq \norm{Ax - b - \lambda A \Delta}_p^p \\ & \leq \norm{Ax-b}_p^p - \lambda \left( g^T A \Delta + \frac{p-1}{p 2^p} \gamma_p(|Ax - b|, A\Delta)\right).
\end{align*}
\end{lemma}

\begin{proof}
First by setting $t=(Ax-b)_i$, $y=(-\lambda A\Delta)_i$ in Lemma \ref{lemma:bregman-divergence-bound} and summing over $i\in [n]$,
we have
\begin{align*}
\norm{Ax-b}_p^p - \lambda g^T A \Delta + \frac{p-1}{p 2^p} \gamma_p(|Ax - b|, \lambda A\Delta) \leq \norm{Ax - b - \lambda A \Delta}_p^p,
\end{align*}
which is the first part of the result. For the other inequality note that by Lemma \ref{lemma:bregman-divergence-bound},
\begin{align*}
\norm{Ax - b - \lambda A \Delta}_p^p \leq \norm{Ax-b}_p^p - \lambda g^T A \Delta + 2^p \gamma_p(|Ax - b|, \lambda A\Delta).
\end{align*}
Moreover by Lemma \ref{lemma:weird-gamma}, since $0<\lambda<1$, we have
\begin{align*}
\gamma_p(|Ax - b|, \lambda A\Delta) \leq \lambda^{\min\{2,p\}} \gamma_p(|Ax - b|, A\Delta).
\end{align*}
Therefore 
\begin{align*}
\norm{Ax-b}_p^p - \lambda g^T A \Delta + 2^p \gamma_p(|Ax - b|, \lambda A\Delta) \leq \norm{Ax-b}_p^p - \lambda g^T A \Delta + \lambda^{\min\{2,p\}} 2^p \gamma_p(|Ax - b|, A\Delta).
\end{align*}
Hence
\begin{align*}
\gamma_p(|Ax - b|, \lambda A\Delta) \leq \norm{Ax-b}_p^p - \lambda \left( g^T A \Delta - \lambda^{\min\{1,p-1\}} 2^p \gamma_p(|Ax - b|, A\Delta) \right)
\end{align*}
The result follows by noting that $\lambda^{\min\{1,p-1\}} \leq \frac{p-1}{p 4^p}$.
\end{proof}

The following result implies that the solution to the linear regression problem can be used as a good initial point for solving the $p$-norm regression problem.

\begin{lemma}
\label{lemma:initial-point}
Let $A\in\mathbb{R}^{n\times d}$,
\begin{align*}
x^* = \argmin_{Cx = v} \norm{Ax - b}_p^p,
\end{align*}
and 
\begin{align*}
x^{(0)} = \argmin_{Cx = v} \norm{Ax - b}_2^2
\end{align*}
Then $\norm{Ax^{(0)} - b}_p^p \leq n^{(p-2)/2} \norm{Ax^* - b}_p^p$.
\end{lemma}
\begin{proof}
First, by definition $\norm{Ax^* - b}_2 \leq \norm{Ax^{(0)} - b}_2$. Moreover by Holder's inequality, 
\begin{align*}
\norm{Ax^{(0)} - b}_2^p \leq n^{(p-2)/2} \norm{Ax^{(0)} - b}_p^p,
\end{align*}
and
\begin{align*}
\norm{Ax^* - b}_p \leq \norm{Ax^* - b}_2.
\end{align*}
The result follows by combining these inequalities.
\end{proof}

Now, we are equipped to show that the $p$-norm regression problem can be solved by approximately solving the instances of the residual problem.

\begin{theorem}
\label{thm:res-problem-number}
The $p$-norm regression problem \eqref{eq:tall-p-norm-w-constraint} can be solved to $\epsilon$ accuracy, by solving $O_p(\alpha \log(n/\epsilon))$ many calls to an $\alpha$-approximate solver for the residual problem \eqref{eq:residual-problem}.
\end{theorem}
\begin{proof}
Let
\begin{align*}
x^* = \argmin_{Cx = v} \norm{Ax - b}_p^p,
\end{align*}
and $\text{OPT} = \norm{A x^* - b}_p^p$. Let $\widetilde{\Delta}$ be an $\alpha$-approximate solution to the residual problem at $x$, i.e.,
\begin{align*}
g^T A \widetilde{\Delta} - \frac{p-1}{p 2^p} \gamma_p(|Ax - b|, A\widetilde{\Delta}) \geq \frac{1}{\alpha} \max_{C \Delta = 0} g^T A \Delta - \frac{p-1}{p 2^p} \gamma_p(|Ax - b|, A\Delta).
\end{align*}
This implies
\begin{align*}
g^T A \widetilde{\Delta} - \frac{p-1}{p 2^p} \gamma_p(|Ax - b|, A\widetilde{\Delta}) \geq \frac{1}{\alpha}\left( g^T A \Delta - \frac{p-1}{p 2^p} \gamma_p(|Ax - b|, A\Delta) \right),
\end{align*}
for any $\Delta$ such that $C\Delta = 0$. Taking $\Delta = \widetilde{x} - x^*$, by the first inequality of Lemma \ref{lemma:bound-res-value}, we have
\begin{align*}
g^T A \Delta - \frac{p-1}{p 2^p} \gamma_p(|Ax - b|, A\Delta)
&
\geq 
\norm{A\widetilde{x} - b}_p^p - \norm{A\widetilde{x} - b - A\Delta}_p^p 
\\ &
= 
\norm{A\widetilde{x} - b}_p^p - \norm{Ax^* - b}_p^p 
= 
\norm{A\widetilde{x} - b}_p^p - \text{OPT}.
\end{align*}
Therefore 
\begin{align}
\label{eq:main-res-thm-bound}
g^T A \widetilde{\Delta} - \frac{p-1}{p 2^p} \gamma_p(|Ax - b|, A\widetilde{\Delta}) \geq \frac{1}{\alpha} \left( \norm{A\widetilde{x} - b}_p^p - \text{OPT} \right)
\end{align}
Now let 
\begin{align*}
    \lambda = \left( \frac{p-2}{p 2^p}\right)^{\frac{1}{\min\{1, p-1\}}}.
\end{align*}
Note that the value of $\lambda$ only depends on $p$ and $\lambda = \Omega_p(1)$. By Lemma \ref{lemma:bound-res-value}, we have
\begin{align*}
\norm{Ax - b - \lambda A \widetilde{\Delta}}_p^p \leq \norm{Ax-b}_p^p - \lambda \left( g^T A \widetilde{\Delta} + \frac{p-1}{p 2^p} \gamma_p(|Ax - b|, A\widetilde{\Delta})\right).
\end{align*}
Therefore
\begin{align}
\label{eq:main-convergence-bound}
\nonumber\norm{Ax - b - \lambda A \widetilde{\Delta}}_p^p - \text{OPT} & \leq - \lambda \left( g^T A \widetilde{\Delta} + \frac{p-1}{p 2^p} \gamma_p(|Ax - b|, A\widetilde{\Delta})\right) + \norm{Ax-b}_p^p - \text{OPT} \\ & \nonumber \leq
-\frac{\lambda}{\alpha} \left( \norm{Ax-b}_p^p - \text{OPT} \right) + \norm{Ax-b}_p^p - \text{OPT} \\ & = 
\left( 1- \frac{\lambda}{\alpha} \right) \left( \norm{Ax-b}_p^p - \text{OPT} \right)
\end{align}
where the second inequality follows from \eqref{eq:main-res-thm-bound}. Let 
\begin{align*}
x^{(0)} = \argmin_{Cx = v} \norm{Ax - b}_2^2
\end{align*}
and $x^{(t+1)} = x^{(t)} - \lambda \widetilde{\Delta}^{(t)}$, where $\widetilde{\Delta}^{(t)}$ is an $\alpha$-approximate solution to the residual problem at point $x^{(t)}$. Then by \eqref{eq:main-convergence-bound},
\begin{align*}
\norm{Ax^{(t)} - b}_p^p - \text{OPT} \leq  \left( 1- \frac{\lambda}{\alpha} \right)^t \left(\norm{Ax^{(0)} - b}_p^p - \text{OPT}\right)  \leq \left( 1- \frac{\lambda}{\alpha} \right)^t \left( n^{(p-2)/2} - 1 \right) \text{OPT},
\end{align*}
where the second inequality follows from Lemma \ref{lemma:initial-point}.
\end{proof}

\subsection{Input Sparsity Time Algorithm via Sampling}
\label{sec:input-sparsity-via-sampling}

We start by proving Theorem \ref{thm:almost-input-sparsity-p-norm}. \cite{AdilKPS19} showed the following.
\begin{lemma}[\cite{AdilKPS19}]
\label{lemma:gamma-bound-adil}
For any $q\geq 2$, $t\geq 0$, and $y\in\mathbb{R}$, we have $\gamma_q(t,y)\geq \abs{y}^q$ and $\gamma_q(t,y) \geq \frac{q}{2} t^{q-2} y^2$.
\end{lemma}
This implies the following.
\begin{lemma}
\label{lemma:gamma-bound}
For any $q\geq 2$, $t\geq 0$, and $y\in\mathbb{R}$, we have
\begin{align*}
t^{q-2} y^2 + \abs{y}^q \leq 2 \gamma_q(t,y) \leq q \left( t^{q-2} y^2 + \abs{y}^q \right).
\end{align*}
\end{lemma}
\begin{proof}
First note that because $q\geq 2$, $\frac{q}{2} t^{q-2} y^2 \geq t^{q-2} y^2$ and therefore by Lemma \ref{lemma:gamma-bound-adil}, $\gamma_q(t,y)\geq t^{q-2} y^2$. Moreover by Lemma \ref{lemma:gamma-bound-adil}, $\gamma_q(t,y)\geq \abs{y}^q$. Therefore 
\[
t^{q-2} y^2 + \abs{y}^q \leq 2 \gamma_q(t,y).
\]
Now note that $\gamma_q(t,y)$ is either equal to $\abs{y}^p+(\frac{q}{2}-1)t^p$ or is equal to $\frac{q}{2}t^{q-2} y^2$. We have two cases.

\textbf{Case 1.} $t\leq |y|$. In this case 
\[
\gamma_q(t,y) = \abs{y}^p+(\frac{q}{2}-1)t^p \leq \frac{q}{2}\abs{y}^p \leq \frac{q}{2}(\abs{y}^p + t^{q-2} y^2),
\]
where the first inequality follows from the case assumption.

\textbf{Case 2.} $\abs{y} <t$. In this case
\[
\gamma_q(t,y) = \frac{q}{2}t^{q-2} y^2 \leq \frac{q}{2}(\abs{y}^p + t^{q-2} y^2),
\]
The result follows from the above case analysis.
\end{proof}

\begin{theorem}[\cite{AdilKPS19}]
\label{thm:solve-res-with-gamma}
The residual problem can be solved to $O_p(\alpha^{1/(\min\{2,p\}-1)})$ approximation by solving $O_p(\log(d))$ instances of the following problem to $\alpha$-approximation:
\begin{align}
    \min_{\Delta} ~~ & \gamma_p(\abs{Ax-b}, A\Delta) \\
    \text{s.t.} ~~ & C \Delta = 0, \\
    & g^T A \Delta = z.
\end{align}
\end{theorem}
The above result implies that we can solve the residual problem only by having a function $\widetilde{\gamma}_p$ that is within a constant factor of $\gamma_p$. We can find such a function by sampling.

\begin{theorem}[\cite{AdilS20}]
\label{thm:q-to-p}
The $p$-norm regression problem can be solved to $\epsilon$ accuracy by $\Otil(p n^{\max\{1/q, 1/(p-1)\}} \log^2(1/\epsilon))$ calls to a smoothed $q$-norm solver.
\end{theorem}

\almostInputSparsityPNorm*
\begin{proof}
Picking $q<p-1$, by Theorem \ref{thm:q-to-p}, we only need to solve $\Otil(pn^{1/q}\log^2(1/\epsilon))$ many $q$-norm problems. 
Let $t = \abs{Ax-b}$ and $y = A\Delta$. Let $T$ be a diagonal matrix such that its $T_{ii}$ is equal to $t_i^{(q-2)/2}$.
For the first part of the result, note that by taking Lemma \ref{lemma:gamma-bound} on all the entries and summing them together, we have
\begin{align*}
\norm{Ty}_2^2 + \norm{y}_q^q \leq2\gamma_q(t, y) \leq q(\norm{Ty}_2^2 + \norm{y}_q^q).
\end{align*}
Now let $\tau_i$ be the leverage scores of the matrix $TA$ and $w_i$ be the $q$-norm Lewis weights of $A$. Then sampling (and rescaling) $\Otil(d^{q/2})$ rows according to $p_i = \max\{\tau_i, w_i\}$ to obtain $\tilde{y}$ and $\tilde{T}$, by Lemmas \ref{lemma:sampling-leverage} and \ref{lemma:sampling-lewis}, with high probability, we have
\begin{align*}
    \frac{1}{2}(\norm{Ty}_2^2 + \norm{y}_q^q) \leq \norm{\tilde{T}\tilde{y}}_2^2 + \norm{\tilde{y}}_q^q \leq \frac{3}{2} (\norm{Ty}_2^2 + \norm{y}_q^q).
\end{align*}
therefore 
\[
\frac{2}{3}(\norm{\tilde{T}\tilde{y}}_2^2 + \norm{\tilde{y}}_q^q) \leq 2\gamma_q(t, y) \leq 2q\norm{\tilde{T}\tilde{y}}_2^2 + \norm{\tilde{y}}_q^q.
\]
Therefore $\gamma_q(t,y)$ is within an $O_q(1)$ factor of $\norm{\tilde{T}\tilde{y}}_2^2 + \norm{\tilde{y}}_q^q$.

Therefore the algorithm is to solve $\Otil_p(n^{1/q})$ many $q$-norm problems. To solve each such problem, first compute the leverage scores of $TA$ and Lewis weights of $A$ (by Theorems \ref{lemma:sampling-leverage} and \ref{lemma:sampling-lewis}, this can be done in $\Otil_p(\text{nnz}(A) + d^{q/2+C})$ time). Then sample (and rescale) and solve the sampled problem to within a constant approximation (this step can be done in $\Otil_q(d^{q/2+1})$ time.)

For the second part of the result note that the dual of $\min_{A^\top x=b} \norm{x}_p$ is $\max_{\norm{A y}_{p/(p-1)} \leq 1 } b^\top y$ which is equivalent to solving
\begin{align}
\label{eq:dualProb}
\min_{b^\top y = 1} \norm{A y}_{p/(p-1)}.
\end{align}
Solving this problem to polynomial accuracy (which we can do by the first part of the result) is equivalent to solving $\min_{A^\top x=b} \norm{x}_p$ to polynomial accuracy --- see Section 7.2 of \cite{AdilKPS19}. Therefore applying the first part of the result to \eqref{eq:dualProb}, the result follows.
\end{proof}

Now we prove our main input sparsity result.

\inputSparsityPNorm*
\begin{proof}
By Theorem \ref{thm:main-sampling}, there is an algorithm that with high probability returns a vector of weights $w$ with $\Otil(d^3)$ nonzeros such that for any $x$ in a polynomial range
\begin{align*}
\frac{1}{4} \sum_{i\in [n]} \gamma_q(t_i,y_i) \leq \sum_{i \in [n]} w_i \gamma_q(t_i,y_i) \leq \frac{5}{4} \sum_{i\in [n]} \gamma_q(t_i,y_i).
\end{align*}
Therefore to solve the residual problem to a constant approximation, we only need to solve the residual problem for an $\Otil(d^3)\times d$ matrix. This can be done in $\Otil(d^4)$ time. Moreover, note that by Theorems \ref{thm:res-problem-number} and \ref{thm:solve-res-with-gamma}, we only need to solve a constant number of such residual problems. The $\text{nnz}(A)$ term comes from the sampling algorithm which needs access to the leverage scores of the matrix --- see Algorithm 1 and Theorem \ref{thm:main-sampling}.
\end{proof}

\section{Regression faster than matrix multiplication}

In this section, we first show how to find a spectral approximation of a matrix $A$ with $\tilde{O}(d)$ rows using a fast sparse linear solver. Then we use this to show that linear regression can be solved faster than matrix multiplication. We also use this result to find spectral approximations for the $p$-norm regression problem for $p$ close to two. Finally, in Section~\ref{subsec:p-norm}, we show how to use inverse maintenance together with sparse linear solvers to go below the matrix multiplication runtime. We hope that these three applications illustrate the versatility and intricacies of using recently developed sparse linear solvers.

\subsection{Spectral Approximation}
\label{subsec:spectral_approx}

Our approach to finding a constant-factor spectral approximation of a matrix $A$ is to first find a ``good'' overestimate of leverage scores of rows of $A$. Lemma \ref{lemma:uniformSampling} clearly demonstrate that if we find a vector of overestimates $u$ such that $\|u\|_1=O(d)$, then with $\tilde{O}(d)$ samples from rows of $A$, we can recover a spectral approximation of $A$ with a high probability. Before discussing how to find such a vector of overestimates, we need the following definitions and results.

\begin{definition}
Let $u\in\mathbb{R}^n_{\geq 0}$. Let $\alpha$ and $c$ be positive constants. Let $p_i:=\min\{1,\alpha\cdot u_i c\log d\}$. We define the function $\textsc{Sample}(u,\alpha, c)$ that outputs a random diagonal matrix $S$ where each element $S_{ii}$ is $\frac{1}{\sqrt{p_i}}$ with probability $p_i$ and zero otherwise.
\end{definition}

In order to prove Theorem \ref{thm:spectral_approx}, we show that one can find a good overestimate of leverage scores in $O(\text{nnz}(A)+d^\theta)$ time using the following lemma from \cite{CohenLMMPS15}.

\begin{lemma}[\cite{CohenLMMPS15}]
\label{lemma:uniformSampling}
Let $0<\mu<1$, and $u$ be a vector of leverage score overestimates, i.e., $\tau_i(A)\leq u_i$. Let $\mu$ be a sampling rate parameter and let $c$ be a fixed positive constant. Let $S=\textsc{Sample}(u,\mu^{-2}, c)$.
Then $S$ has at most $2 \norm{u}_1 \mu^{-2} c \log d$ nonzero entries and $\frac{1}{\sqrt{1+\mu}}S A$ is a $\left(\frac{1+\mu}{1-\mu}\right)$-spectral approximation for $A$ with probability at least $1-d^{-c/3} - (3/4)^d$.
\end{lemma}

The following results are useful.

\begin{lemma}
\label{lemma:sum-leverage-bound}
For a matrix $A\in\mathbb{R}^{n\times d}$,
$
\sum_{i=1}^n \tau_i(A) \leq d.
$
\end{lemma}

The proof of the following lemma is similar to that of Theorem 3 of \cite{CohenLMMPS15}.

\begin{lemma}
\label{lemma:undersample}
Let $u$ be a vector of leverage score overestimates. For some undersampling factor $\alpha\in(0,1]$, let $S=\sqrt{3\alpha/4}\cdot \textsc{Sample}(u,9\alpha,c)$, where $c$ is a constant. Let $\tau_i^{S A}(A)\leq v_i\leq (1+\beta) \tau_i^{S A}(A)$, for all $i\in[n]$, where $\beta \geq 0$. Let $u'_i=\min\{v_i, u_i\}$. Then with a probability of $1-d^{-c/3}-(3/4)^d$, $u'_i$ is a leverage score overestimate, for all $i\in[n]$, $\sum_{i=1}^n u'_i\leq \frac{3d(1+\beta)}{\alpha}$, and the number of nonzeros of $S$ is $O(\alpha \norm{u}_1 \log d)$.
\end{lemma}

\begin{proof}
Let $S'=\frac{1}{\sqrt{3\alpha/4}}S$. Note that the number of nonzeros of $S$ and $S'$ are equal. By Lemma \ref{lemma:uniformSampling}, $\frac{1}{\sqrt{1+1/(3\sqrt{\alpha})}} S' A$ is a $\left( \frac{\sqrt{1+1/(3\sqrt{\alpha})}}{\sqrt{1-1/(3\sqrt{\alpha})}} \right)$-spectral approximation of $A$ and $S'$ has at most $18\alpha c \norm{u}_1 \log d = O(\alpha \norm{u}_1 \log d)$ with a probability of at least $1-d^{-c/3} - (3/4)^d$. Therefore
with a probability of $1-d^{-c/3} - (3/4)^d$,
\[
\frac{1}{(1+1/(3\sqrt{\alpha}))(3\alpha/4)} A^\top S^2 A=\frac{1}{1+1/(3\sqrt{\alpha})}A^\top (S')^2 A \preceq A^\top A.
\]
Now note that $(1+1/(3\sqrt{\alpha}))(3\alpha/4)\leq 1$ for $\alpha\in(0,1]$. Therefore $A^\top S^2 A \preceq A^\top A$. Hence, for all $i\in [n]$, $\tau_i(A) \leq \tau_i^{SA}(A)$. Therefore with a probability of $1-d^{-c/3}-(3/4)^d$, for all $i\in [n]$, $u'_i$ is a leverage score overestimate. 

Now we bound $\sum_{i=1}^n u'_i$. Note that $\textsc{Sample}(u, 9\alpha, c)$ and $\textsc{Sample}(\alpha u, 9, c)$ are equal in distribution. Therefore by Lemma \ref{lemma:uniformSampling}, $S' A$ is a $\frac{1}{\sqrt{4/3}}$-spectral approximation of $A$ with probability of $1-d^{-c/3}-(3/4)^d$ --- note that this does not add to the probability of failure because one of $1/\sqrt{4/3}$ and $\left( \frac{\sqrt{1+1/(3\sqrt{\alpha})}}{\sqrt{1-1/(3\sqrt{\alpha})}} \right)$ is smaller than the other one and we can use the probability of success of the tighter bound which would imply the other one. Therefore 
\[
\frac{1}{2} A^\top A \preceq \frac{3}{4} A^\top (S')^2 A = \frac{1}{\alpha} A^\top S^2 A
\]
Hence, for all $i\in [n]$ such that $a_i \perp \text{ker}(SA)$,
\[
\tau_i^{SA}(A) \leq \frac{2}{\alpha} \tau_i(A),
\]
Now we have
\begin{align*}
\sum_{i=1}^n u'_i =\sum_{i=1}^n \min\{u_i,v_i\} = \sum_{i=1}^n v_i \leq \sum_{i=1}^n (1+\beta) \tau_i^{SA}(A) \leq \sum_{i=1}^n (1+\beta) \frac{2}{\alpha} \tau_i(A) \leq \frac{2d(1+\beta)}{\alpha},
\end{align*}
where the last inequality follows from Lemma \ref{lemma:sum-leverage-bound}.
\end{proof}

Now, we are equipped to give a high-level view of our algorithm and prove Theorem \ref{thm:spectral_approx}. The high-level description of the algorithm is as the following.
\begin{enumerate}
    \item Start from the vector of overestimates $u=\vec{1}$.
    \item Repeat the following process for $\log(n/d)$ iterations.
    \begin{enumerate}
        \item Sample $\tilde{O}(d)$ rows from $A$ based on the vector of overestimates $u$ to form $\overline{A}$.
        \item Update the vector of overestimates of leverage scores using $\overline{A}$, i.e.,
        \[
        u_i \leftarrow a_i^\top (\overline{A}^\top \overline{A})^+ a_i.
        \]
    \end{enumerate}
    \item Return $\tilde{O}(d)$ rows of $A$ sampled based on $u$.
\end{enumerate}

We use Lemma \ref{lemma:undersample} and choose our parameters so that in each iteration of this algorithm, we cut the $\ell_1$ norm of $u$ by a half. So after $\log(n/d)$ iterations, the $\ell_1$ norm of $u$ is about $d$ which means $u$ is a good vector of overestimates of leverage scores. If we perform step (b) of the algorithm naively, then the cost of each update is $d^2$ and total cost of each iteration is $n d^2$. However one can use random projection to do such updates more efficiently.

\begin{proof}[Proof of Theorem \ref{thm:spectral_approx}]
We show Algorithm \ref{alg:spectral_approx} finds a spectral approximation in time $O(\text{nnz}(A) + d^\theta)$. The technique follows that of \cite{CohenLMMPS15} that finds the leverage scores of a matrix in a recursive fashion by updating the overestimates. We first prove the correctness of the algorithm assuming that all the randomized steps have succeeded. We then analyze the running time. Finally we bound the failure probability.

\textbf{Correctness.} Algorithm \ref{alg:spectral_approx} starts with a vector of leverage score overestimates $u^{(0)} = \mathbbm{1}_{[n]}$. Therefore $\norm{u^{(0)}}_1 = n$. To compute the generalized leverage scores we need to compute the following
\begin{align*}
a_i^\top ((S^{(i)}A)^\top (S^{(i)}A))^+ a_i & = a_i^\top ((S^{(i)}A)^\top (S^{(i)}A))^+ ((S^{(i)}A)^\top (S^{(i)}A)) ((S^{(i)}A)^\top (S^{(i)}A))^+ a_i \\ & = \norm{(S^{(i)}A) ((S^{(i)}A)^\top (S^{(i)}A))^+ a_i}_2^2
\end{align*}
We compute the pseudo-inverse using Theorem \ref{thm:sparse_inverse}.
By Lemma \ref{lemma:random-projection} and Theorem \ref{thm:sparse_inverse}, we have
\begin{align*}
\norm{(S^{(i)}A) ((S^{(i)}A)^\top (S^{(i)}A))^+ a_i}_2^2 & \leq \left(1+\frac{1}{n^9-1}\right) \frac{1}{0.9 r}\norm{ G (S^{(i)} A) Y^{(i)} (Z^{(i)})^\top a_i }_2 \\ & \leq \left(1+\frac{1}{n^{8}}\right)\left(\frac{11}{9}\right) \norm{(S^{(i)}A) ((S^{(i)}A)^\top (S^{(i)}A))^+ a_i}_2^2 \\ & \leq 
2 \norm{(S^{(i)}A) ((S^{(i)}A)^\top (S^{(i)}A))^+ a_i}_2^2
\end{align*}

Therefore by Lemma \ref{lemma:undersample}, $u_j^{(i)}$'s are leverage score overestimates and 
\[\norm{u^{(i)}}_1 \leq \frac{6d}{12d/\norm{u^{(i-1)}}_1} = \frac{\norm{u^{(i-1)}}_1}{2}\]
as long as $\norm{u^{(i-1)}}_1 \geq 12d$. Hence $\norm{u^{(z)}}_1 = O(d)$ and $u^{(z)}$ is a vector of leverage score overestimates.
Thus by Lemma \ref{lemma:uniformSampling}, $\frac{1}{\sqrt{3/2}} SA$ is a $3$-spectral approximation of $A$ with high probability.

\textbf{Running time.}
By Lemma \ref{lemma:undersample}, $S^{(i)}$ has
\[
O\left(\alpha_i \norm{u_{i-1}}_1 \log d \right) = O\left(\frac{12d}{\norm{u_{i-1}}_1} \norm{u_{i-1}}_1 \log d\right) = O(d\log d)
\]
nonzeros. Therefore $S^{(i)} A$ has $O(d\log d)$ nonzero rows. Hence $S^{(i)} A$ has $\tilde{O} (\text{nnz}_d(A))$ nonzero entries. Although we want to find an inverse for $(S^{(i)} A)^\top (S^{(i)} A)$, we do not perform this matrix multiplication because it is too costly. Note that by Theorem \ref{thm:sparse_inverse}, we only need to be able to do matrix-vector multiplication to find the inverse operator. Note that because the algorithm only has a logarithmic number of iterations, we only need to bound the cost of each iteration. Sampling $S^{(i)}$ given the vector $u^{(i-1)}$ can be done in $O(n)$ time. By Theorem \ref{thm:sparse_inverse}, finding the sparse inverse operator $Z_{(S^{(i)} A)^\top (S^{(i)} A)}$ takes $\tilde{O}((d\cdot \text{nnz}_d(A)\cdot m + d^2\cdot m^3+ (\frac{d}{m})^\omega m^2)\log(\kappa))$. By Theorem \ref{thm:sparse_inverse}, $M^{(i)}$ can be computed in time $O((\text{nnz}_d(A) \cdot m + d^2) (\log d + \log(n/d)))$ because $G$ has $O(\log d + \log(n/d))$ number of rows. Note that the entries of $M^{(i)}$ only need $\tilde{O}(1)$ bits because the number of bits required for the entries of $Z_{(S^{(i)} A)^\top (S^{(i)} A)}$ after the multiplication is $\tilde{O}(1)$. Therefore for each $j\in[n]$, the norm $\norm{M^{(i)} a_j}$ can be computed in $\tilde{O}(\text{nnz}(a_j))$, where $a_j$ is the $j$'th row of $A$. So the leverage score overestimate can be updated in time $\tilde{O}(\text{nnz}(A))$ in each iteration. Hence the total running time of the algorithm is 
\[
\Otil\left(\text{nnz}\left(A\right) + \left(d\cdot \text{nnz}_d\left(A\right)\cdot m + d^2\cdot m^3+ \left(\frac{d}{m}\right)^\omega m^2\right)\log\left(\kappa\right)\right).
\]

\textbf{Failure probability.} In each iteration of the for loop in Algorithm \ref{alg:spectral_approx}, there are three sources of randomness: 1) sampling the matrix $S^{(i)}$; 2) the sparse linear system solver to find the inverse of $(S^{(i)} A)^\top (S^{(i)} A)$; 3) and the random JL projection to update the leverage score estimates. We bound the failure probability of each of these steps. Finally, the algorithm samples $S$ and returns $SA$ as the spectral approximation. We bound the failure probability of this step as well.

In each iteration, the probability that $S^{(i)} A$ is not a spectral approximation of $A$ or it does not have $O(d\log d)$ rows is less than $d^{-c/3}+(3/4)^d$. The probability that the sparse inverse method cannot does not find an inverse with the desired property is less than $d^{-10}$. By Lemma \ref{lemma:random-projection} and the union bound, the probability that the projected vectors (with the Guassian matrix) do not have a norm in the right interval is less than
\[
n \cdot 2e^{-(\epsilon^2-\epsilon^3)r/4} = n \cdot 2e^{-(0.01-0.001)(4000/9)\cdot (\log d + \log (n/d))/4} = 2 n \cdot \frac{d}{n} \cdot \frac{1}{d^{11}} = 2 d^{-10}
\]
Therefore, with a logarithmic number of iterations, the total probability failure is $\tilde{O}(d^{-10} + (3/4){d})$.
\end{proof}

\RestyleAlgo{algoruled}
\IncMargin{0.15cm}
\setcounter{algocf}{1}
\begin{algorithm}[h]
\footnotesize
\textbf{Input:} $A\in \mathbb{R}^{n\times d}$, $b\in \mathbb{R}^n$, $\epsilon>0$, m\\
$u^{(0)} \leftarrow \mathbbm{1}_{[n]}$, $c \leftarrow 30$, $z \leftarrow \log(n/d)$, $r\leftarrow (4000/9)\cdot (11\log d + \log(n/d))$\\
\For {$i=1,\ldots,z$}{
$\alpha_i \leftarrow \frac{12d}{\norm{u^{(i-1)}}_1}$\\
$S^{(i)} \leftarrow \sqrt{3\alpha_i / 4} \cdot \textsc{Sample}(u^{(i-1)}, 9\alpha_i, c)$\\
Find a sparse inverse operator $Z_{(S^{(i)} A)^\top (S^{(i)} A)}$ with $m$ blocks such that $\norm{Z_{(S^{(i)} A)^\top (S^{(i)} A)} - ((S^{(i)} A)^\top (S^{(i)} A))^{-1}}_F \leq \kappa^{-10} d^{-10}$ with high probability via Theorem \ref{thm:sparse_inverse}.\\
$G\leftarrow$ random $r\times d$ Gaussian matrix.\\
$M^{(i)} \leftarrow G (S^{(i)} A) Z_{(S^{(i)} A)^\top (S^{(i)} A)}$\\
\ForAll{$j\in [n]$}{
$u^{(i)}_j \leftarrow \min\{(1+\frac{1}{n^9-1})\frac{1}{0.9r}\norm{M^{(i)} a_j}_2^2, u^{(i-1)}_j\}$ 
}}
$S \leftarrow \frac{1}{\sqrt{3/2}}\textsc{Sample}(u^{(z)}, 4, c)$\\
\Return $SA$
\caption{Spectral Approximation}
\label{alg:spectral_approx}
\end{algorithm}

\subsection{Tall Linear Regression (\texorpdfstring{$p=2$}{p=2})}
\label{subsec:lin-reg}

In the case of linear regression, the idea is to use Algorithm \ref{alg:spectral_approx} to find a spectral approximation of the matrix and then we can find an inverse of the spectral approximation using Theorem \ref{thm:sparse_inverse}. Then we use Richardson's algorithm (Lemma \ref{lemma:richardson}) to solve the regression problem. The high-level view of the algorithm is as the following.

\begin{enumerate}
    \item Find a $\lambda$-spectral approximation $\widetilde{A}$ of the matrix $A$.
    \item Set $x\leftarrow \vec{0}$.
    \item Repeat the following for $O(\lambda \log \left( \frac{\kappa \norm{b}_2}{\epsilon OPT} \right))$ iterations.
    \[
    x \leftarrow x - (\lambda\widetilde{A}^\top \widetilde{A})^{-1} (A^\top A x - A^\top b)
    \]
\end{enumerate}

For Step 1 of this algorithm, we use Algorithm \ref{alg:spectral_approx} to find the spectral approximation. We show that this algorithm finds the desired solution.

\begin{proof}[Proof of Theorem \ref{thm:linear-regression}]
We assume that $\widetilde{A}$ is found using Algorithm \ref{alg:spectral_approx} and therefore it has $\tilde{O}(d)$ rows. Moreover $M=\lambda\widetilde{A}^\top \widetilde{A}$ for the Richardson's iterations. Hence $A^\top A \preceq M = \lambda \widetilde{A}^\top \widetilde{A} \preceq \lambda A^\top A$. Therefore by Lemma \ref{lemma:richardson}, after $k$ steps, we have
\[
\norm{x^{(k)}-x^*}_M \leq \left(1-\frac{1}{\lambda}\right)^k\norm{x^*}_M
\]
Now we need to show that for the right choice of $k$, we have 
\[
\norm{Ax^{(k)}-b}_2^2 \leq (1+\epsilon) \norm{Ax^* - b}_2^2.
\]
To do so, we show that it is enough to pick $k$ such that
\[
\norm{x^{(k)}-x^*}^2_{M} \leq \epsilon \norm{Ax^*-b}_2^2
\]
Note that $A^\top b= A^\top A x^*$. Therefore
\begin{align}
\label{eq:optimal-obj-lin-reg}
\nonumber 
\norm{Ax^* - b}_2^2 
& 
= 
(Ax^* - b)^\top (Ax^* - b) 
= 
(x^*)^\top A^\top A x^* + b^\top b - 2(x^*)^\top A^\top b 
\\ &
= 
(x^*)^\top A^\top A x^* + b^\top b - 2(x^*)^\top A^\top A x^* 
= 
b^\top b - (x^*)^\top A^\top A x^*
\end{align}
Hence
\[
(x^*)^\top A^\top b = (x^*)^\top A^\top A x^* = \norm{Ax^*}_2^2 \leq \norm{b}_2^2.
\]
Moreover because $M \preceq \lambda A^\top A$,
\begin{align*}
\norm{ x^{(0)} - x^{*}}_M^2 & = b^\top A M A^\top b \leq \lambda b^\top A A^\top A A^\top b = \lambda \norm{A A^\top b}_2^2
\end{align*}
Therefore by Lemma \ref{lemma:richardson},
\begin{align}
\norm{x^{(k)}-x^*}_M \leq (1-\frac{1}{\lambda})^k \norm{x^{(0)} - x^*}_M \leq (1-\frac{1}{\lambda})^k \lambda \norm{A A^\top b}_2^2
\end{align}
Moreover
\begin{align*}
\norm{A x^{(k)} - b}_2^2 & = (A x^{(k)} - b)^\top (A x^{(k)} - b) = (x^{(k)})^\top A^\top A x^{(k)} + b^\top b - 2(x^{(k)})^\top A^\top b \\ & = (x^{(k)})^\top A^\top A x^{(k)} + b^\top b - 2(x^{(k)})^\top A^\top A x^*.
\end{align*}
Therefore if $\norm{x^{(k)}-x^*}^2_{M} \leq \epsilon \norm{Ax^*-b}_2^2$, because $A^\top A \preceq M$, then
\begin{align*}
(x^{(k)}-x^*)^\top A^\top A (x^{(k)}-x^*) \leq \norm{x^{(k)}-x^*}_M^2 \leq \epsilon\norm{Ax^* - b}_2^2   
\end{align*}
Hence by \eqref{eq:optimal-obj-lin-reg},
\begin{align*}
(x^{(k)})^\top A^\top A x^{(k)} - 2(x^{(k)})^\top A^\top A x^* + (x^*)^\top A^\top A x^* \leq \epsilon (b^\top b - (x^*)^\top A^\top A x^*)
\end{align*}
Thus
\begin{align*}
(x^{(k)})^\top A^\top A x^{(k)} - 2(x^{(k)})^\top A^\top A x^* + b^\top b \leq (1+\epsilon) (b^\top b - (x^*)^\top A^\top A x^*)
\end{align*}
Therefore by $A^\top b = A^\top A x^*$,
\begin{align*}
\norm{Ax^{(k)}-b}_2^2 \leq (1+\epsilon) \norm{Ax^* - b}_2^2
\end{align*}
Thus it is enough to set the number of iterations to $k \geq \lambda \log \left( \frac{\lambda\norm{AA^\top b}_2^2}{\epsilon \norm{Ax^*-b}_2^2} \right)$. Moreover note that each iteration of the algorithm takes $\tilde{O}(\text{nnz}(A) + \text{nnz}_d(A) \cdot m + d^2)$ time. The other terms of the running time come from Theorem \ref{thm:spectral_approx} to find a spectral approximation. Therefore the total running time of the algorithm is 
\[
\tilde{O} \left(
\left( \text{nnz}\left(A\right) + d \cdot \text{nnz}_d(A) \cdot m + d^2 \cdot m^3 + \left( \frac{d}{m} \right)^\omega m^2 \right)
\cdot \log^{2} (\kappa) \log \left( \frac{\lambda\norm{AA^\top b}_2^2}{\epsilon \norm{Ax^*-b}_2^2} \right) \right).
\]
The result follows by picking $m=\min\{ d \cdot \text{nnz}_d(A)^{\frac{1}{\omega-1}}, d^{\frac{\omega-2}{\omega+1}}\}$.
\end{proof}

\subsection{\texorpdfstring{$p$}{p}-Norm Regression for $p$ close to $2$}
\label{subsec:p-close-to-two}

\cite{bubeck2018homotopy} showed that the following for the $\gamma_p$ function (Definition \ref{def:quadratic-smooth}).

\begin{lemma}[\cite{bubeck2018homotopy}]
The function $\gamma$ has the following properties.
\begin{enumerate}
    \item $\gamma_p(0,x)=|x|^p$. 
    \item $\gamma_p(t,\cdot)$ is quadratic on $[-t,t]$; 
    \item  $\gamma_p$ is in $C^1$.
\end{enumerate}
\end{lemma}

Using this function, \cite{bubeck2018homotopy} developed a homotopy based algorithm for solving the $p$-norm regression problem --- see Algorithm \ref{alg:tall-p-norm}. The algorithm starts with a large $t$ and decreases $t$ over a logarithmic number of phases. The reason that this algorithm works is that the following ``quadratic extension'' is well-conditioned on a box ($l\leq x\leq u$) around the optimal solution of $t_k$ which includes the optimal solution for $t_{k+1}$.

\begin{definition}
\label{def:quad-ext}
\[
f_{t,\ell,u}(s)=\begin{cases}
\gamma_p(t,s) & \text{ if } \ell \leq s \leq u \\
\gamma_p(t,u)+\frac{d}{ds}\gamma_p(t,u)\cdot(s-u)+\frac{1}{2}\frac{d^2}{ds^2}\gamma_p(t,u)\cdot(s-u)^2 & \text{ if } u \leq s \\
\gamma_p(t,\ell)+\frac{d}{ds}\gamma_p(t,\ell)\cdot(s-\ell)+\frac{1}{2}\frac{d^2}{ds^2}\gamma_p(t,\ell)\cdot(s-\ell)^2 & \text{ if } s \leq \ell
\end{cases}
\]
\end{definition}

Now we are equipped to give a high-level description of the algorithm for tall $p$-norms.

\begin{enumerate}
    \item Set $t_0 = 2 \norm{b}_2$ and $x(t_0) = \argmin_x \norm{Ax-b}_2^2$.
    \item Repeat the following for $k=1,\ldots, O(1)\cdot \log (n p t_0^p/\epsilon)$
    \begin{enumerate}
        \item Set $t_k = \left(1-\frac{1}{2p}\right) t_{k-1}$
        \item Set $x(t_k)=\argmin_x \gamma_p(t_k, Ax-b)$
    \end{enumerate}
\end{enumerate}

First of all if $t_k$ is small enough, then $x(t_k)$ is close to the optimal solution of $\min \norm{Ax-b}_p^p$ (see Lemma 5 of \cite{bubeck2018homotopy}).
Second, the crux of the above algorithm is to implement Step (b). In general the condition number of $\gamma$ function can be large. Therefore instead of minimizing $\gamma_p(t_k,Ax-b)$ itself, we minimize the quadratic extension $f_{t_k,l,u}(Ax-b)$ for the appropriate bounds $l,u$ (see Definition \ref{def:quad-ext}). The functions $\gamma$ and $f_{t,l,u}$ have unique minimizers because of their strict convexity property. Therefore if we pick $l$ and $u$ such that $l \leq A(\argmin_x \gamma_p(t_k,Ax-b))-b \leq u$, then finding the minimum of $f_{t_k,l,u}(Ax-b)$ is equivalent to finding the minimum of $\gamma_p(t_k,Ax-b)$. Moreover the condition number of $f_{t_k,l,u}(Ax-b)$ is equal to the condition number of $\gamma_p(t_k,Ax-b)$ restricted to the set $\{x:l\leq Ax-b\leq u\}$. We pick $l$ and $u$ that determine a neighborhood around $Ax(t_{k-1})-b$ that contains $Ax(t_k)-b$. The algorithm works because $Ax(t_{k-1})-b$ and $Ax(t_{k})-b$ are close to each other. Therefore $\{x:l\leq Ax(t_{k-1})-b\leq u\}$, that contains $Ax(t_{k})-b$. is small enough so that $\gamma_p(t_k, Ax-b)$ on this set has a small condition number --- see Section 2.2 of \cite{bubeck2018homotopy}. 

Finally, note that $f_{t_k,l,u}$ is well-conditioned with respect to $Ax-b$ and not necessarily with respect to $x$. Therefore we need to use a preconditioner $A P^{(k)}$ such that $f_{t_k,l,u}(A P^{(k)} y -b)$ is well-conditioned with respect to $y$. For this we pick $P^{(k)} = (\widetilde{A}^\top \widetilde{A})^+ \widetilde{A}^\top$ where $\widetilde{A}$ is a constant-factor spectral approximation (with $\tilde{O}(d)$ rows) of $\sqrt{D^{(k)}}A$ and $D^{(k)}$ is a diagonal matrix such that $D^{(k)}_{ii} = \frac{p-1}{2} \max \{t_k^{p/2}, |(Ax^{(k)}-b)_i|^{p/2}-\text{sign}(p-2)\gamma \}^{2-4/p}$ --- see Sections 2.2 and 3 of \cite{bubeck2018homotopy} for details. To find the spectral approximation, we use Theorem \ref{thm:spectral_approx} and to find the inverse of $\widetilde{A}^\top \widetilde{A}$, we use the sparse linear system solver (Theorem \ref{thm:sparse_inverse}).

In summary, \cite{bubeck2018homotopy} proves the following result.

\begin{theorem}[\cite{bubeck2018homotopy}]
Algorithm \ref{alg:tall-p-norm} returns $\overline{x}$ such that
\[
\norm{A\overline{x} - b}_p^p \leq \epsilon + \min_{x\in\mathbb{R}^d}  \norm{Ax-b}_p^p,
\]
\end{theorem}

\RestyleAlgo{algoruled}
\IncMargin{0.15cm}
\begin{algorithm}[h]
\textbf{Input:} $A\in \mathbb{R}^{n\times d}$, $b\in \mathbb{R}^n$, $p\in (1,\infty)$, $\epsilon>0$\\
$t_0 \leftarrow 2 \norm{b}_2$, $\gamma  \leftarrow \left( 1 + \frac{p^2}{2(p-1)} \sqrt{n} \right) t^{p/2}$\\
Find a $2$-spectral approximation $\widetilde{A}$ of $A$ and set $x^{(0)} = \argmin_x \norm {Ax-b}_2^2$ \tcp*{The solution to this regression problem is found using Theorem \ref{thm:linear-regression}, i.e., Richardson's iterations that use $2 \widetilde{A}^\top \widetilde{A}$ as the preconditioner.}
    $z\leftarrow O(1)\cdot \log (np t_0^p /\epsilon)$\\
    \ForAll{$k=0,1,\ldots z$}{
    $t_{k+1} \leftarrow (1-\frac{1}{2p}) t_k$\\
    Set $D^{(k)}$ to a diagonal matrix where $D^{(k)}_{ii} = \frac{p-1}{2} \max \{t_k^{p/2}, |(Ax^{(k)}-b)_i|^{p/2}-\text{sign}(p-2)\gamma \}^{2-4/p}$\\
Find a constant-factor spectral approximation $\widetilde{A}^{(k)}$, with $\tilde{O}(d)$ rows of $\sqrt{D^{(k)}} A$.\\
$P^{(k)} \leftarrow ((\widetilde{A}^{(k)})^\top \widetilde{A}^{(k)})^{+} (\widetilde{A}^{(k)})^\top$ \tcp*{Use the sparse linear system solver (Theorem \ref{thm:sparse_inverse}) to find the inverse.}
    Given $x^{(k)},P^{(k)}$, find $y^{(k+1)}$ by minimizing the following function using mini-batch Katyusha \cite{allen2017katyusha} on
        \begin{align}
        g^{(k)}(y) := \sum_{i=1}^n f_{t_{k+1}, (|(Ax^{(k)}-b)_i|^{p/2}-\gamma)^{2/p}, (|(Ax^{(k)}-b)_i|^{p/2}+\gamma)^{2/p}}((AP^{(k)} y-b)_i),
        \end{align}
        For each iteration of mini-batch Katyusha, we compute the following, for a set $S\subseteq [n]$ of size $\eta$,
        \[
        \sum_{i\in S} \nabla F_i(x) = P^{(k)} \sum_{i\in S} f'_{(1-h)t_k, (|s_i(t)|^{p/2}-\gamma)^{2/p}, (|s_i(t)|^{p/2}+\gamma)^{2/p}}(a_i\cdot P^{(k)} y - b_i) a_i
        \]
     Compute $x^{(k+1)}$ by the formula $x^{(k+1)} \leftarrow P^{(k)}y^{(k+1)}$.\\
    }
    \Return $x^{(z)}$
\caption{Tall $p$-norm regression}
\label{alg:tall-p-norm}
\end{algorithm}

We show that, using the sparse inverse solver to find the initial solution and preconditioners in this algorithm improves the running time of the algorithm to better than $d^\omega$.

\pCloseToTwo*
\begin{proof}
In this proof, for brevity, we show the running time of the sparse linear system solver with $d^\theta$. 
We show that Algorithm \ref{alg:tall-p-norm} runs in the mentioned time complexity.
First note that, by Lemma \ref{lemma:richardson} and Theorem \ref{thm:spectral_approx}, we can compute $x^{(0)}$ in time $\tilde{O}(\text{nnz}(A) + d^\theta)$. Moreover, $D^{(k)}$ can be computed in $O_p(n)$ because it is a diagonal matrix. 

We can find this spectral approximation by Theorem \ref{thm:spectral_approx} in time $\tilde{O}(\text{nnz}(A)+d^{\theta})$. 
Note that to multiply a vector $y$ with $P^{(k)}$ we need to first multiply by $(\widetilde{A}^{(k)})^\top$, which takes $\tilde{O}(\text{nnz}_d(A))$ because $\widetilde{A}^{(k)}$ contains $\tilde{O}(d)$ (scaled) rows of $\sqrt{D^{(k)}}A$. We then have to multiply $(\widetilde{A}^{(k)})^\top y$ with $((\widetilde{A}^{(k)})^\top \widetilde{A}^{(k)})^+$. By Theorem \ref{thm:sparse_inverse}, 
this process takes 
$\tilde{O}(\text{nnz}_d(A) \cdot m + d^2)$ time. Then for a set $S\subseteq [n]$ of size $\eta$, $\sum_{i\in S} \nabla F_i(x)$ can be computed in time $O(\text{nnz}(A) \frac{\eta}{n})$ after computing $P^{(k)} y$. Therefore each iteration of mini-batch Katyusha takes $O(\text{nnz}(A) \frac{\eta}{n}+ \text{nnz}_d(A) \cdot m + d^2)$ in expectation. Note that $\text{nnz}_d(A) \cdot m \leq d^2$ by assumption. Moreover in each phase of the algorithm, we pay a preprocessing time of $O(\text{nnz}(A)+d^{\theta})$ to find the linear operator for $P^{(k)}$, i.e, the inverse operator of $(\widetilde{A}^{(k)})^\top \widetilde{A}^{(k)}$. Moreover as discussed in \cite{bubeck2018homotopy} the smoothness and strong convexity parameters of the function are equal to $L=O_p(n^{|1-2/p|})$ and $\sigma=\Omega(1)$, respectively. Also the sum of smoothness parameters of $F_i$ functions is equal to $\sum_{i\in [n]} L_i=O_p(n^{|1-2/p|}d)$. Let $\kappa=L/\sigma$ be the condition number of the Hessian. Then mini-batch Katyusha takes $\tilde{O}_p(\frac{n}{\eta} +\sqrt{\kappa} +\frac{1}{\eta}\sqrt{n\kappa d})$ iterations. Let $Z=\text{nnz}(A)$. Then the total running time of the algorithm is
\begin{align*}
& \tilde{O}_p\left[\left(\frac{n}{\eta}+\sqrt{\kappa} +\frac{1}{\eta} \sqrt{n\kappa d}\right)\left(Z \frac{\eta}{n}+ d^2\right)+Z+d^{\theta} \right]
\\ & = \tilde{O}_p\left[Z\left(1+\sqrt{\frac{\kappa d}{n}}\right) + d^{\theta} + d^2 \sqrt{\kappa} + \frac{d^2\sqrt{n}}{\eta}\left(\sqrt{\kappa d} +\sqrt{n}\right)+ Z\sqrt{\kappa}\frac{\eta}{n} \right]
\end{align*}
We now optimize $\eta$ over $\frac{     d^2\sqrt{n}}{\eta}(\sqrt{\kappa d} +\sqrt{n})+ Z\sqrt{\kappa}\frac{\eta}{n}$. If $\kappa d \geq n$, then we choose $\eta=\lceil \sqrt{\frac{n^{3/2} d^{5/2}}{Z}} \rceil$. Then 
\[
\frac{d^2\sqrt{n}}{\eta}\left(\sqrt{\kappa d} +\sqrt{n}\right)+ Z\sqrt{\kappa}\frac{\eta}{n} = O\left(\sqrt{Z} d^{5/4} n^{-1/4} \sqrt{\kappa}\right) \leq O\left(Z \sqrt{\frac{\kappa d}{n}} + d^2 \sqrt{\kappa}\right),
\]
where the inequality follows from the AM-GM inequality. Therefore the total cost is
\[
\tilde{O}_p\left[Z\left(1+\sqrt{\frac{\kappa d}{n}}\right) + d^{\theta} + d^2 \sqrt{\kappa} \right]
\]
Because $\kappa\leq O(n^{|1-2/p|})$, $\kappa d \geq n$ implies $n\leq O_p(d^{p/2})$ if $p\geq 2$, and $n\leq O_p(d^{\frac{1}{2-2/p}})$ if $p\leq 2$. Therefore for $p\geq 2$, we have
\begin{align*}
\tilde{O}_p\left[Z\left(1+\sqrt{\frac{\kappa d}{n}}\right) + d^{\theta} + d^2 \sqrt{\kappa} \right] & = \tilde{O}_p\left[Z+nd\sqrt{\frac{\kappa d}{n}} + d^{\theta} + n^{1/2-1/p} d^2  \right] \\ & \leq
\tilde{O}_p\left[Z+d^{p/2+1} + d^{\theta} + n + (d^2)^{\frac{1}{1/2+1/p}}  \right] \\ & =
\tilde{O}_p\left[Z+d^{p/2+1} + d^{\theta} + n +  d^{\frac{4}{1+2/p}}  \right] \\ & = \tilde{O}_p\left[Z+ n + d^{p/2+1} + d^{\theta}  \right],
\end{align*}
where the first inequality follows from the weighted AM-GM. The second inequality follows from $\frac{4}{1+2/p}\leq \frac{p}{2}+1$.
For $p\leq 2$, similarly, we can show that
\begin{align*}
    \tilde{O}_p\left[Z\left(1+\sqrt{\frac{\kappa d}{n}}\right) + d^{\theta} + d^2 \sqrt{\kappa} \right] & = \tilde{O}_p\left[Z + n + d^{\theta} + d^{\frac{p}{2(p-1)}+1} \right]
\end{align*}
\noindent
If $\kappa d \leq n$, then we choose $\eta=\lceil \sqrt{\frac{ n^2 d^2}{Z \sqrt{\kappa}}} \rceil$ and then
\[
\frac{ d^2\sqrt{n}}{\eta}\left(\sqrt{\kappa d} +\sqrt{n}\right)+ Z\sqrt{\kappa}\frac{\eta}{n} =O\left(\sqrt{Z} d \kappa^{1/4}\right) \leq O\left(Z + d^2 \sqrt{\kappa}\right).
\]
Therefore the cost is 
\[
\tilde{O}_p\left[Z\left(1+\sqrt{\frac{\kappa d}{n}}\right) + d^{\theta} + d^2 \sqrt{\kappa} \right] = \tilde{O}_p\left(Z+d^{\theta} + n +  d^{\frac{4}{1+2/p}}\right) \leq \tilde{O}_p\left(Z+d^{\theta} + n +  d^{\frac{p}{2}+1}\right),
\]
where the last inequality follows from $\frac{4}{1+2/p}\leq \frac{p}{2}+1$.

Therefore the total running time of the algorithm is
\[
\tilde{O} \left(
\left(
\text{nnz}\left(A\right) + d^{0.5\max\left\{p,\frac{p}{p-1}\right\}+1}
+ d \cdot \text{nnz}_d(A) \cdot m + d^2 \cdot m^3 + \left( \frac{d}{m} \right)^\omega m^2 \right)
\log^{2} \left(\kappa/\epsilon\right) \log \left( \frac{\kappa\norm{b}_2}{\epsilon OPT}
\right)
\right),
\]
The result follows by picking $m=\min\{ d \cdot \text{nnz}_d(A)^{\frac{1}{\omega-1}}, d^{\frac{\omega-2}{\omega+1}}\}$.
\end{proof}

\subsection{\texorpdfstring{$p$}{p}-Norm Regression for any $p>1$ faster than matrix multiplication}
\label{subsec:p-norm}

In this section, we consider (P2) which is of the form
\[
\min_{A^\top x = b} \norm{x}_p^p.
\]

\cite{AdilKPS19} showed that for numbers $x,\Delta\in\mathbb{R}$,
\[
|x|^p + \Delta \frac{d}{d x}|x|^p + \frac{p-1}{p^{2^p}} \gamma_p(|x|,\Delta) \leq |x+\Delta|^p \leq |x|^p + \Delta \frac{d}{d x}|x|^p + 2^p \gamma_p(|x|,\Delta),
\]
where $\gamma$ is the quadratically smoothed $p$-norm function --- see Definition \ref{def:quadratic-smooth}. Note that both $2^p$ and $\frac{p-1}{p^{2^p}}$ are $O_p(1)$.
This inequality suggests the following iterative scheme for $p$-norm regression problem.
\begin{enumerate}
    \item Start from an initial point $x \in \mathbb{R}^n$ such that $A^\top x=b$.
    \item Repeat the following
    \begin{enumerate}
        \item Find $\Delta\in \mathbb{R}^n$ such that $A^\top \Delta = 0$ and minimizes
        \begin{align}
        \label{eq:exact-res-problem}
         \sum_{j=1}^d \Delta_j \frac{d}{d x_j} |x_j|^p + O_p(1) \gamma_p(|x_j|,\Delta_j)
        \end{align}
        \item Update $x$ to $x+\Delta$
    \end{enumerate}
\end{enumerate}

It is shown that by a logarithmic number of iterations of the above algorithm, one can solve the $p$-norm regression problem to $\epsilon$ accuracy.
Moreover, instead of iteration (a) in the above algorithm, we can guess the value of $z=\Delta^\top \nabla \norm{x}_p^p$ (in a binary search fashion) and solve a logarithmic number of problems of the following form (see Theorems \ref{thm:res-problem-number} and \ref{thm:solve-res-with-gamma})
\begin{align}
\label{eq:res-problem}
    \min_{\Delta} \gamma_p(x, \Delta) \\
    A^\top \Delta = 0 \nonumber \\
    g^\top \Delta = z \nonumber,
\end{align}
where $g$ is the gradient vector, $\nabla \norm{x}_p^p$,
and we overload the notation for $\gamma$ to denote $\sum_{j=1}^d \gamma_p(x_j, \Delta_j)$ with $\gamma_p(x, \Delta)$. Note that by doing line search on the value of $\Delta^\top \nabla \norm{x}_p^p$, we can also remove the $O_p(1)$ term completely. \cite{AdilKPS19} has shown that it is enough to solve $O_p(\alpha \log(\frac{n}{\epsilon}))$ many problems of form \eqref{eq:res-problem} to $\alpha$ approximation to solve the $p$-norm regression problem. Theorem 5.8 of \cite{AdilKPS19} states that \eqref{eq:res-problem} can be solved by solving $\tilde{O}_p(n^{\frac{p-2}{3p-2}})$ many problems of the form
\begin{align}
\label{eq:weighted-lin-reg}
    \min_{\Delta} \frac{1}{2} \Delta^\top R \Delta \\
    A^\top \Delta = 0 \nonumber\\
    g^\top \Delta = z, \nonumber
\end{align}
where $R$ is a diagonal matrix. Note that \eqref{eq:weighted-lin-reg} is a weighted linear regression problem. To solve \eqref{eq:res-problem} using instances of \eqref{eq:weighted-lin-reg}, one starts from an initial $R$ and then repeats the following.
\begin{enumerate}
    \item Solve \eqref{eq:weighted-lin-reg} with $R$ to find $\Delta^*$.
    \item Update $R$ based on $\Delta^*$ via a multiplicative weights update algorithm.
\end{enumerate}
Therefore by the above discussion, one can solve the $p$-norm regression problem by solving $\tilde{O}_p (\alpha n^{\frac{p-2}{3p-2}} \log(\frac{n}{\epsilon}))$ many instances of $\eqref{eq:weighted-lin-reg}$ --- see \cite{AdilKPS19}.
One caveat of this result is that the number of problems needed to be solved is exponential in $p$. This was improved by Adil and Sachdeva \cite{AdilS20} to $O(p \alpha n^{\frac{p-2}{3p-2}} \log^2 (n/\epsilon))$ solves of instances of \eqref{eq:weighted-lin-reg}. This is achieved by showing that a smoothed $p$-norm problem can be solved by solving $p n^{\max\{\frac{1}{q},\frac{1}{p-1}\}} \log^2 (n/\epsilon)$ instances of the smoothed $q$-norm problem and using a homotopy approach to solve the problem for the following norms $2^{-k}p, 2^{-k+1}p,\ldots, \frac{p}{2}, p$. The $q$ is then picked to be $ \sqrt{\log(n)}$ which adds a factor of $n^{o(1)}$ to the running time.
The overall result can be summarized as the following.

\begin{theorem}[\cite{AdilKPS19,AdilS20}]
\label{thm:reg-using-res}
The problem of
\[
\min_{A^\top x=b} \norm{x}_p^p
\]
can be solved by solving $O(p \alpha \log^2 (\frac{n}{\epsilon}))$ instances of the following residual problem each to an $\alpha$ approximation, where the objective value of the optimal solution is less than or equal to one and $n^{-1/p}\leq t_j\leq 1, \forall j$.
\begin{align*}
    \min_{\Delta} \gamma_p(t, \Delta) \\
    A^\top \Delta = 0 \nonumber \\
    g^\top \Delta = z \nonumber
\end{align*}
\end{theorem}

Therefore, we can focus on solving the residual problems of the form \eqref{eq:res-problem} by solving instances of \eqref{eq:weighted-lin-reg}. First, it is easy to analytically find the solution of the weighted regression problem by the method of Lagrange multipliers.

\begin{theorem}[\cite{AdilKPS19}]
\label{thm:sol-weighted-lin-reg}
The solution to the problem
\begin{align*}
    \min_{\Delta} \frac{1}{2} \Delta^\top R \Delta \\
    A^\top \Delta = 0 \\
    g^\top \Delta = z.
\end{align*}
is
\begin{align}
\label{eq:res-solution}
\Delta = R^{-1} \left( A + \frac{z-g^\top R^{-1} A}{g^\top R^{-1} g} \right) v,
\end{align}
where
\begin{align*}
v = \frac{z (A^\top R^{-1} A)^{-1} A^\top R^{-1} g}{g^\top R^{-1} g - g^{\top} R^{-1} A (A^\top R^{-1} A)^{-1} A^\top R^{-1} g},
\end{align*}
\end{theorem}

Note that $\Delta$ can be find in time needed to compute $v$ plus $O(n^2)$ to do a matrix-vector multiplication. Moreover to find the vector $(A^\top R^{-1} A)^{-1} A^\top R^{-1} g$, we can use a $\tilde{O}(1)$ spectral approximation of $(A^\top R^{-1} A)^{-1}$ and use Richardson's iteration (Lemma \ref{lemma:richardson}) to find $(A^\top R^{-1} A)^{-1} A^\top R^{-1} g$ with high accuracy in $\tilde{O}(1)$ iterations. Moreover if $\widetilde{R}$ is within an $\tilde{O}(1)$ factor of $R$, then $(A^\top \widetilde{R}^{-1} A)$ is a $\tilde{O}(1)$ spectral approximation of $(A^\top R^{-1} A)$. Therefore it is enough to maintain the inverse $(A^\top \widetilde{R}^{-1} A)^{-1}$ such that $\widetilde{R}$ is within an $\tilde{O}(1)$ factor of $R$ and apply this inverse in the Richardson's iteration to the vector $A^\top R^{-1} g$. The reason that this gives improvements is that the entries of $R$ change slowly. Therefore, we can use the following identity to perform the low-rank updates.

\begin{lemma}[Sherman-Morrison-Woodbury identity]
For an invertible $n \times n$ matrix $M$ and matrices $U\in \Rbb^{n \times r},C \in \Rbb^{r \times r},V \in \Rbb^{r \times n}$, we have
\[
(M+UCV)^{-1} = M^{-1} - M^{-1} U (C^{-1} + V M^{-1} U)^{-1} V M^{-1}.
\]
\end{lemma}

After finding the vector $(A^\top R^{-1} A)^{-1} A^\top R^{-1} g$, we can multiply it by $z$ or $g^\top R^{-1} A$ to find the terms we need for \eqref{eq:res-solution}.

Now we are equipped to state the algorithm for solving the residual problem of the form \eqref{eq:res-problem}. Note that as we mentioned, we only need to solve about $n^{(p-2)/(3p-2)}$ instances of weighted linear regression \eqref{eq:weighted-lin-reg} to solve \eqref{eq:res-problem}. If we naively find the inverse of $A^\top R^{-1} A$ for each instance separately, the cost becomes about $n^{\omega+(p-2)/(3p-2)}$ which is too high. The entries of the diagonal matrix $R$ change slowly. So, we can use the inverse maintenance technique based on the Sherman-Morrison-Woodbury identity to perform low-rank updates to the inverse in order to maintain a spectral approximation of the true inverse. This spectral approximation can then be used as a preconditioner in the Richardson's iteration (Lemma \ref{lemma:richardson}) to compute \eqref{eq:res-solution}.

There are two important differences between our approach to inverse maintenance and the previous one used by \cite{AdilKPS19}.

\begin{enumerate}
    \item We cannot update the inverse matrix directly because we only have access to it via a linear operator --- see Section \ref{sec:accessInverse}. Therefore we keep a dense matrix $Q$ in which the result of low-rank updates is accumulated. Hence our spectral approximation of the inverse is of the form $Y+Q$ where $Y$ is a linear operator for the inverse of $A^\top (\hat{R}^{(0)})^{-1} A$ computed by the sparse inverse solver of Peng and Vempala \cite{PengVempala21}. Note that because $Q$ is the result of multiplication of different parts of $Y$, each of its entries have $\tilde{O}(1)$ bits.
    
    \item Because of the cost of access to the inverse operator (see Theorem \ref{thm:sparse_inverse}, we cannot allow updates of rank more than about $(n/m)$ where $m$ is the number of blocks of the block Krylov space used for the sparse inverse. Therefore, once every $(n/m)^{(p-2)/(3p-2)}$ iterations, we compute the sparse inverse from scratch.
\end{enumerate}

Algorithm \ref{alg:res-problem} illustrates the pseudocode of our process. The red lines show the main differences between our algorithm and \cite{AdilKPS19}. The next theorem states that Algorithm \ref{alg:res-problem} solves the residual problem and gives a bound on the size of low-rank updates which we use to bound the running time of the algorithm.

\begin{theorem}[\cite{AdilKPS19}]
\label{thm:res-k}
Suppose for problem \eqref{eq:res-problem}, the objective of the optimal solution is less than one and $n^{-1/p}\leq t_j\leq 1, \forall j$, then
Algorithm \ref{alg:res-problem} returns a solution $\overline{\Delta}$ with high probability such that $A^\top\overline{\Delta} = 0$, $g^\top \Delta = z$, and $\gamma_p(t, \overline{\Delta})$ is within an $O_p(1)$ factor of the optimal objective value.

Moreover, let $k_{i,\eta}$ be the number of indices $j$ that are added to $E$ at iteration $i$ due to changes between $2^{-\eta}$ and $2^{-\eta+1}$. Let $t=\tilde{\Theta}_p(n^{\frac{p-2}{3p-2}})$ be the number of iterations. Then
\begin{align}
\sum_{i=1}^t k_{i,\eta} = \begin{cases}
0 & ~  ~ \text{ if } 2^\eta > t \\
\tilde{O}_p \left(n^{\frac{p+2}{3p-2}} 2^{2\eta}\right) & ~ ~ \text{ otherwise.}
\end{cases}
\end{align}
\end{theorem}

\RestyleAlgo{algoruled}
\IncMargin{0.15cm}

\begin{algorithm}[!h]
\textbf{Input:} $A\in \mathbb{R}^{n\times d}$, $t\in \mathbb{R}^n$, $g\in \mathbb{R}^n$, $z\in \mathbb{R}$, $p\in (1,\infty)$\\
$w_j^{(0)} \leftarrow 0, \forall j\in[n]$\\
$x \leftarrow \vec{0}$\\
$\rho \leftarrow \tilde{\Theta}_p(n^{\frac{(p^2-4p+2)}{p(3p-2)}})$\\
$\beta \leftarrow \tilde{\Theta}_p(n^{\frac{p-2}{3p-2}})$\\
$\alpha \leftarrow \tilde{\Theta}_p(n^{-\frac{(p^2-5p+2)}{p(3p-2)}})$\\
$\tau \leftarrow \tilde{\Theta}_p(n^{\frac{(p-1)(p-2)}{(3p-2)}})$\\
$T = \alpha^{-1} n^{1/p} = \tilde{\Theta}_p(n^{\frac{p-2}{3p-2}})$\\
$r_j^{(0)} \leftarrow (n^{1/p} t_j)^{p-2}, \forall j\in [n]$\\
\ForAll{$i = 0,\ldots, t$}{
{\bf \underline{(1) Recompute the sparse inverse operator.}}\\
\If{\underline{$i$ is a multiple of $(n/m)^{(p-2)/(3p-2)}$}}{
\underline{$\hat{r}_j \leftarrow r_j^{(i)}, \forall j\in [n]$}\\
\underline{$c_{j, \eta} \leftarrow 0$ for all $j\in [n]$ and $\eta \in [\log(t)] \cup \{0\}$}\\
\underline{Use Theorem~\ref{thm:sparse_inverse} to find
inverse operator for $A^\top \hat{R}^{-1} A$
with error $\kappa^{-10} n^{-10}$, $Y$.}\\
\underline{Set $Q$ to a $d \times d$ matrix of all zeros.}
}

{\bf (2) Solve the weighted linear regression by Richardson's iteration and preconditioning. See Section \ref{subsec:spectral_approx}, Lemma \ref{lemma:richardson}, and Theorem \ref{thm:linear-regression}.}
\\
$\Delta^* = \argmin_{\Delta} \sum_{j\in [n]} r_j \Delta_j^2$ s.t. $A^\top \Delta = 0$ and $g^\top \Delta = z$. \DontPrintSemicolon \tcp*{The solution is given by \eqref{eq:res-solution} and we can use $Y+Q$ as a preconditioner to find $\Delta^*$ with a high accuracy.}
{\bf (3) Update the weights.}\\
\If{$\norm{ \Delta^* }_p^p \leq \tau$}{
$w^{(i+1)}_j \leftarrow w^{(i)}_j + \alpha |\Delta_j^*|, \forall j \in [n]$\\
$x \leftarrow x + \alpha \Delta^*$\\
} \Else{
For all $j\in [n]$ with $|\Delta_j| \geq \rho$ and $r_j \leq \beta$ do \\
~  ~ ~ $w_j^{(i+1)} \leftarrow 4^{1/(p-2)} \max \{n^{1/p}, w_j^{(i)}\}$\\
For rest of $j\in[n]$ do $w_j^{(i+1)} \leftarrow w_j^{(i)}$
}
$r_j^{(i+1)} \leftarrow (n^{1/p} t_j)^{p-2} + w_j^{p-2}$\\
{\bf (4) Find the significant buckets.}\\
For all $j\in [n]$ find the least non-negative integer $\eta_j$ such that $\frac{1}{2^{\eta_j}} \leq \frac{r_j^{(i+1)} - r_j^{(i)}}{\hat{r}_j}$\\
For all $j \in [n]$, $c_{j,\eta_j} \leftarrow c_{j,\eta_j} + 1$\\
$E \leftarrow \cup_{\eta: i+1 \mod{2^\eta}\equiv 0} \{j: c_{j,\eta} \geq 2^\eta\}$\\
$\hat{r}_{j} \leftarrow r_j^{(i+1)}, \forall j \in E$\\
$c_{j,\eta} \leftarrow 0$ for all $(j,\eta)$ such that $j\in E$.\\
{\bf \underline{(5) Update the preconditioner.}}\\
\underline{$Q \leftarrow Q - (Y + Q) (A_E)^\top ((R^{(i+1)})_{E,E}^{-1} + A_E (Y + Q) (A_E)^\top)^{-1} A_E (Y + Q)$.} \DontPrintSemicolon \tcp*{$A_E$ is a matrix obtained by taking the rows of $A$ in $E$.}
}
\Return $d^{-1/p} x$
\caption{Algorithm for the Residual Problem --- Steps (1) and (5), underlined, are new.}
\label{alg:res-problem}
\end{algorithm}

The above theorem states that after $k$ iterations of the algorithm at most about $k^3$ of the weights have changed significantly.
Now, we are equipped to bound the time complexity of solving the residual problem \eqref{eq:res-problem}.

\pNorm*
\begin{proof}
First, note that the only randomness of the algorithm comes from finding the inverse matrices in iterations that are multiples of $(n/m)^{(p-2)/(3p-2)}$. Because the number of iterations of the algorithm is $\tilde{O}_p(n^{(p-2)/(3p-2)})$, using the sparse inverse approach of~\cite{PengVempala21},
i.e, Theorem \ref{thm:reg-using-res},
the algorithm succeeds with high probability.

Next, we need to bound the time complexity of Algorithm \ref{alg:res-problem}. Note that the time complexity of this algorithm is bounded by finding the sparse inverse operators (Line 14), solving the weighted linear regression problems (Line 16), and performing low-rank updates to the inverse (Line 29). In the following, we bound the running time of these.

\textbf{Running time of finding the sparse inverse operators.} The algorithm has at most $\tilde{O}_p(n^{(p-2)/(3p-2)})$ iterations and we compute the sparse inverse once every $\tilde{O}_p((n/m)^{(p-2)/(3p-2)})$ iterations. Therefore, by Theorem \ref{thm:sparse_inverse}, the total cost of computing sparse inverse operators over the course of the algorithm is
\[
\tilde{O}_p\left(m^{(p-2)/(3p-2)}\left(n\cdot \text{nnz}(A)\cdot m + n^2 \cdot m^3 + n^{\omega} m^{2-\omega}\right)\right)
\]

\textbf{Running time of solving weighted linear regression problems using the preconditioner $Y+Q$.}
As discussed in the beginning of this section, to find the solution \eqref{eq:res-solution} to the weighted regression problem \eqref{eq:weighted-lin-reg}, it is enough to have an inverse $(A^\top \widetilde{R}^{-1} A)^{-1}$, where $\widetilde{R}$ is within a factor of $\tilde{O}(1)$ of $R$, and use this inverse in the Richardson's iteration (Lemma \ref{lemma:richardson}). We call $(A^\top \widetilde{R}^{-1} A)^{-1}$ a preconditioner for $(A^\top R^{-1} A)^{-1}$. Note that $Y+Q$ (see Algorithm \ref{alg:res-problem}), provides such a preconditioner. The reason is that the algorithm checks once every $2^\eta$ iterations whether the number of changes of size between $2^{-\eta}$ and $2^{-\eta+1}$ to an entry is more than $2^\eta$. This way the algorithm guarantees the contribution of such changes to an entry is at most $2^{-\eta+1} \cdot 2^{\eta} \cdot 2 = 4$. Moreover there are a logarithmic number of different $\eta$'s. Therefore $\widetilde{R}$ can be at most $\tilde{O}(1)$ far from $R$.

By Theorem \ref{thm:sol-weighted-lin-reg} and Lemma \ref{lemma:richardson} to solve the weighted linear regression problems, we need to compute $A^\top R^{-1} g$ and do a logarithmic number of matrix vector multiplications with the spectral approximation of $(A^\top R^{-1} A)^{-1}$ that is provided by $Y + Q$. Note that computing $A^\top R^{-1} g$ takes $O(\text{nnz}(A))$ time. By Theorem \ref{thm:sparse_inverse}, the cost of the multiplications is $\tilde{O}(n^2 + \text{nnz}(A) \cdot m)=\tilde{O}(n^2)$ and by assumption $\text{nnz}(A) \cdot m \leq n^2$. Therefore because the algorithm has $\tilde{O}_p(n^{(p-2)/(3p-2)})$ iterations, the total cost of solving weighted linear regression problems is
\[
\tilde{O}_p(n^{(p-2)/(3p-2)} n^2)
\]

\textbf{Running time of low rank updates.} Because we find the sparse inverse operator once every $(n/m)^{(p-2)/(3p-2)}$ iterations. No low rank update happens due to $\eta$ that $2^\eta > (n/m)^{(p-2)/(3p-2)}$.

We list the operations and the respective running times needed to do an update of rank $r$ in the following.

\begin{enumerate}
    \item Computing $(Y+Q)(A_E)^\top$. By theorem \ref{thm:sparse_inverse}, the cost of multiplying the sparse inverse with a $d\times r$ matrix is $\tilde{O}(r\cdot \text{nnz}(A) \cdot m + d^{2} r^{\omega-2})$. Moreover, $Q$ is an $d\times d$ matrix such that each entry of which has $\tilde{O}(1)$ bits. Therefore multiplying $Q$ by a $d\times r$ matrix takes
    \[
    \tilde{O}(\textsc{MM}(d,d,r))\leq  \tilde{O}\left(\left(\frac{d}{r}\right)^2 \textsc{MM}(r,r,r)\right) = \tilde{O}(d^2 r^{\omega-2}) \leq \tilde{O}(n^2 r^{\omega-2})
    \]
    time. Computing $A_E (Y+Q)$ is similar.
    \item Computing $((R^{(i+1)})_{E,E}^{-1} + A_E (Y + Q) (A_E)^\top)^{-1}$. Computing $A_E (Y + Q) (A_E)^\top$
    is a left multiply by $A_E$ which has size $r \times d$. This multiplication takes
    \[
    \tilde{O}(\textsc{MM}(r,d,r)) \leq \tilde{O}\left(\frac{d}{r}\textsc{MM}(r,r,r)\right) = \tilde{O}(d r^{\omega-1}) \leq \tilde{O}(n^2 r^{\omega-2})
    \]
    time, where the last inequality follows from $r, d\leq n$. Finally $((R^{(i+1)})_{E,E}^{-1} + A_E (Y + Q) (A_E)^\top)$ is an $r\times r$ matrix and each of its entries have $\tilde{O}(1)$ bits. Therefore computing its inverse takes $\tilde{O}(r^\omega) \leq \tilde{O}(n^2 r^{\omega-2})$.
    
    \item Computing $(Y + Q) (A_E)^\top ((R^{(i+1)})_{E,E}^{-1} + A_E (Y + Q) (A_E)^\top)^{-1} A_E (Y + Q)$. For this we need to multiply a $d\times r$ matrix with an $r\times r$ matrix and then multiply a $d \times r$ matrix with an $r\times d$ matrix. This takes
    \[
    \tilde{O}(\textsc{MM}(d,r,d)+\textsc{MM}(d,r,r))=\tilde{O}(\textsc{MM}(d,r,d))\leq  \tilde{O}\left(\left(\frac{d}{r}\right)^2 \textsc{MM}(r,r,r)\right) = \tilde{O}(d^2 r^{\omega-2}) \leq \tilde{O}(n^2 r^{\omega-2})
    \]
    time.
\end{enumerate}

Therefore the cost of an update of rank $r$ is $\tilde{O}(r\cdot \text{nnz}(A) \cdot m + n^{2} r^{\omega-2})$. Hence, by Theorem \ref{thm:res-k}, the total cost of low rank updates over the course of the algorithm is

\begin{align*}
& \sum_{\eta=0}^{\log (n/m)^{(p-2)/(3p-2)}}
\sum_{i=0}^z \left(k_{i,\eta} \cdot \text{nnz}(A) \cdot m + n^2 \left(k_{i,\eta}\right)^{\omega-2}\right) \\
& = \left(\text{nnz}\left(A\right) \cdot m\right)
\left(\sum_{\eta=0}^{\log (n/m)^{\left(p-2\right)/\left(3p-2\right)}}
\sum_{i=0}^z k_{i,\eta} \right)
+
n^2 \sum_{\eta=0}^{\log (n/m)^{(p-2)/(3p-2)}}
\sum_{i=0}^z \left(k_{i,\eta}\right)^{\omega-2} \\
& \leq \tilde{O}_p \left(\text{nnz}\left(A\right)\cdot m \cdot n^{(p+2)/(3p-2)} \cdot \left(\frac{n}{m}\right)^{2(p-2)/(3p-2)}\right)
+
n^2 \sum_{\eta=0}^{\log (n/m)^{(p-2)/(3p-2)}} \tilde{O}_p\left(n^{\frac{p-(10-4\omega)}{3p-2}} 2^{\eta(3\omega-7)}\right) \\
& \leq \tilde{O}_p\left(\text{nnz}(A) \cdot n \cdot m^{(p+2)/(3p-2)}\right) +
\tilde{O}_p\left(n^2 n^{\frac{p-(10-4\omega)}{3p-2}}
\left(1+ \frac{n^{\frac{(p-2)(3\omega-7)}{3p-2}}}{m^{(p-2)/(3p-2)}}\right)\right) \\ &
= \tilde{O}_p
\left(\text{nnz}(A) \cdot n \cdot m^{(p+2)/(3p-2)} + n^{2+\frac{p-(10-4\omega)}{3p-2}} + \frac{n^{\omega}}{m^{(p-2)/(3p-2)}}\right),
\end{align*}
where the first inequality follows from Theorem \ref{thm:sparse_inverse} and the concavity of the function $f(a)=a^{\omega-2}$, which implies that the maximum of the summation happens when all the summands are equal.
The second inequality follows from the fact that the maximum summand of the summation is either for $\eta=0$ or $\eta=\log(n/m)^{(p-2)/(3p-2)}$ depending on whether $3\omega-7$ is positive or negative.

\textbf{Numerical stability of inverse maintenance.}
The inverse operator $Y$ that we start with has some error (see Theorem \ref{thm:sparse_inverse}). We need to argue that this error does not increase over the iterations where we do inverse maintenance using the Sherman-Morrison-Woodbury identity.
Lemma \ref{lemma:smw-stability} shows that the inverse maintenance using Sherman-Morrison-Woodbury identity is numerically stable. The round-off error of finding the low-rank inverses does not increase the overall error by assuming that the round-off error is much smaller than the error of the sparse inverse solver. For the numerical stability of matrix operations, see \cite{DemmelDHK07,DemmelDH07}. For stability of inverse maintenance (in the context of linear programming), see \cite{renegar88}.

\end{proof}

\begin{lemma}[Numerical stability of inverse maintenance by Sherman-Morrison-Woodbury identity]
\label{lemma:smw-stability}
Let $Z, \widetilde{Z}, C$ be positive semi-definite matrices. Let $0 < \epsilon < 1$. Suppose 
\begin{align}
\label{eq:smw-assumption}
\frac{1}{1+\epsilon} Z^{-1} \preceq \widetilde{Z}^{-1} \preceq \frac{1}{1-\epsilon} Z^{-1}.
\end{align}
Then
\[
\frac{1}{1+\epsilon}(Z + U^\top C U)^{-1} \preceq \widetilde{Z}^{-1} - \widetilde{Z}^{-1} U (C^{-1} + U^\top \widetilde{Z}^{-1} U)^{-1} U^\top \widetilde{Z}^{-1} \preceq \frac{1}{1-\epsilon}(Z + U^\top C U)^{-1}
\]
\end{lemma}
\begin{proof}
First note that because $C$ is positive semi-definite $U^\top C U$ is also positive semi-definite. Moreover $(1-\epsilon)(U^\top C U) \preceq U^\top C U \preceq (1+\epsilon)(U^\top C U)$. Because $\epsilon<1$, $(1-\epsilon)(U^\top C U)$ is positive semi-definite. Therefore by assumption \eqref{eq:smw-assumption},
\[
(1-\epsilon) (Z + U^\top C U) \preceq (\widetilde{Z} + U^\top C U) \preceq (1+\epsilon) (Z + U^\top C U).
\]
Therefore because $Z + U^\top C U$ and $\widetilde{Z} + U^\top C U$ are positive semi-difinite matrices,
\begin{align}
\label{eq:lemma-stab-smw-1}
\frac{1}{1+\epsilon}(Z + U^\top C U)^{-1} \preceq (\widetilde{Z} + U^\top C U)^{-1} \preceq \frac{1}{1-\epsilon}(Z + U^\top C U)^{-1}
\end{align}
Moreover by Sherman-Morrison-Woodbury identity,
\begin{align}
\label{eq:lemma-stab-smw-2}
(\widetilde{Z} + U^\top C U)^{-1} = \widetilde{Z}^{-1} - \widetilde{Z}^{-1} U (C^{-1} + U^\top \widetilde{Z}^{-1} U)^{-1} U^\top \widetilde{Z}^{-1}.
\end{align}
The result follows by combining \eqref{eq:lemma-stab-smw-1} and \eqref{eq:lemma-stab-smw-2}.
\end{proof}

\section{Accessing the sparse block-Krylov inverse}
\label{sec:accessInverse}

In this section, we formalize, with error bounds,
the type of access one has to the inverse of projection operator
defined from sparse matrix.
Specifically, we describe the running time of solving a sparse matrix against a batch of vectors as stated
in Theorem~\ref{thm:sparse_inverse}.
The statements below are closely based on the top-level
claims in~\cite{PengVempala21}\footnote{Version 2, \url{https://arxiv.org/pdf/2007.10254v2.pdf}}. 

\begin{proof}[Proof of Theorem~\ref{thm:sparse_inverse}.]

Since $A W A^{\top}$ is already symmetrized,
we can ignore the outer step involving a multiplication
by the transpose of an asymmetric matrix.
So we will show how to give access to an
operator $Z_{AWA^{\top}}$ such that
\begin{align}
\norm{
Z_{AWA^{\top}}
-
\left( A W A^{\top} \right)^{-1}
}_{F}
\leq
\epsilon
\end{align}

The algorithm that computes access to this $Z$ was given in Section 7
of~\cite{PengVempala21}.
\begin{enumerate}
    \item Perturb with random Gaussian $R$ to form the perturbed matrix
    \[
    \widehat{A}
    =
    AWA^{\top}
    +
    R
    \]
    \item Generate Krylov space with $\Otil(m)$ extra columns,
    \[
    K
    = 
    \left[
    \begin{array}{ccccc}
    G^{S}
    &
    \widehat{A} G^{S}
    &
    \widehat{A}^2 G^{S}
    &
    \ldots
    &
    \widehat{A}^{m - 1} G^{S}
    \end{array}
    \right],
    \]
    which is padded with a dense, $n$-by-$n - ms$ dense Gaussian $G$ to form $Q = [K, G]$.
    \item Replace the inverse of the block Krylov space portion, $H = K^{\top} \widehat{A} K$
    using the block-Hankel inverse.
    \item Complete the inverse using another Schur complement
    / low rank perturbation,
    and further multiplications by $Q$ on the outside.
\end{enumerate}

Specifically, for step (3),
the $Z_H$ generated by the block-Hankel solver
is the product of two explicit matrices, each with
$\Otil(m \log(\kappa))$ bits,
\[
Z_H
=
X_{H} Y_{H}
\]
such that the cost of 
computing $X_{H} B$, $Y_{H} B$, $X_{H}^{\top} B$, $Y_{H}^{\top} B$
for some $sm$-by-$r$ matrix $B$ with up to
$\Otil(m \log{\kappa})$ bits per entry is
$\Otil(m^2 \log{\kappa}
\textsc{MM}(\frac{n}{m}, \frac{n}{m}, r))$
by Lemma 6.6 of~\cite{PengVempala21}\footnotemark[1],

Then in step (4),
$Z_H$ is extended to the full inverse for
$AWA^{\top}$,
$Z_{AWA^{\top}}$, via
the operator:
\begin{equation}
Z_{A W A^{\top}}
=
Q
\left[\begin{array}{c|c}
I & Z_{HG}\\
\hline
0 & I
\end{array}\right]
\left[\begin{array}{c|c}
Z_H & 0 \\
\hline
0 & Z_{GG}
\end{array}\right]
\left[\begin{array}{c|c}
I & 0 \\
\hline
Z_{GH} & I
\end{array}\right]\\
\end{equation}
where the intermediate matrices $Z_{GH}$, $Z_{GG}$, and $Z_{HG}$ are given by:
\begin{align}
Z_{GH} = Z_{HG}^{\top} & =
- \left( \widehat{A} G \right)^{\top} \widehat{A} K Z_{H}\\
Z_{GG}
& =
\left[ 
\left( \widehat{A} G \right)^{\top} \widehat{A} G
- \left( \widehat{A} G \right)^{\top} \widehat{A} K
Z_{H} \left( \widehat{A} K \right)^{\top} \widehat{A} G
\right]^{-1}
\end{align}

The last block has size $\Otil(m)$,
so the blocks get explicitly computed.
We can also extract out its effect, and treat it as a
separate perturbation to the overall matrix:
\[
Z_{A W A^{\top}}
=
Q
\left( 
Z_{H}
+
\left[\begin{array}{c|c}
Z_{HG} Z_{GG} Z_{GH} & Z_{HG} Z_{GG} \\
\hline
Z_{GG} Z_{GH} & Z_{GG}
\end{array}\right]
\right)
\left( \widehat{A} Q \right)^{\top}.
\]
Here we overloaded notation
by extending $Z_{H}$ onto the full coordinates
(filling the extra with $0$s).
Observe the second matrix is
\[
\left[
\begin{array}{c}
Z_{HG}\\
I
\end{array}
\right]
Z_{GG}
\left[
\begin{array}{cc}
Z_{GH} & I
\end{array}
\right]
\]
(this is, in fact, excatly what Sherman-Moorison-Woodbury gives).
So we can treat the whole thing as a rank-$\Otil(m)$
perturbation to $Z_{H}$.
Substituing in the factorization of $Z_{H}$ as
$X_{H} Y_{H}^{\top}$, we get back
\[
Z_{A W A^{\top}}
=
\left[
\begin{array}{cc}
X_{H}
& \left[
\begin{array}{c}
Z_{HG}\\
I
\end{array}
\right]
Z_{GG}
\end{array}
\right]
\left[
\begin{array}{c}
Y_{H}\\
\left[
\begin{array}{cc}
Z_{GH} & I
\end{array}
\right]
\end{array}
\right]
\]

The cost of multiplying $Z$ against a $n$-by-$r$ matrix
$B$ is then broken down into three parts:
\begin{enumerate}
    \item The cost of multiplying $Y_{H}$ against an $sm$-by-$r$ matrix: by Lemma 6.6 of~\cite{PengVempala21}\footnotemark[1],
    this takes time $\Otil((\frac{n}{m})^{\omega} m^2 \log{\kappa})$.
    \item The cost of multipling $X_{H}$ against an $sm$-by-$r$ matrix, with $\Otil(m \log{\kappa})$ extra bits in the numbers.
    This takes the same time as above, since both $X_{H}$ and $Y_{H}$ already have $\Otil(m \log{\kappa})$ bits in their entries.
    \item Multiplying the extra matrices $Z_{HG}$,
    $Z_{GG}$, and $Z_{GH}$: these are $\Otil(m)$-by-$n$
    matrices (with $\Otil(m \kappa)$ bits per number),
    so the running times are lower order
    terms by the assumption of $m < n^{1/4}$.
    \item Mutliplying $n$-by-$r$ matrices with $\Otil(m \log(\kappa))$ bits by $Q$ and $(AQ)^{\top}$:
    this has two parts: multiplying by $G^{S}$, and by
    a degree $m$ polynomial in $\widehat{A}$.
    The former's cost is at most $O(n^2 m^{3})$ by the
    sparsity bound on $G^{S}$, while the latter's cost 
    is the cost of $O(m r)$ matrix-vector multiplies in 
    $A$ against vectors with $m \log\kappa$ bits.
\end{enumerate}

\end{proof}

\bibliographystyle{alpha}
\bibliography{ref}
\end{document}